\documentclass[
aps,
pra,
reprint,
footinbib,superscriptaddress,notitlepage,numbers]{revtex4-1}
\usepackage[top=0.9in,bottom=0.9in,left=0.9in,right=0.9in]{geometry}

\usepackage{graphicx}
\usepackage{indentfirst}
\usepackage{braket}
\usepackage{float}
\usepackage{amsmath}
\usepackage{amssymb}
\usepackage{amsfonts}
\usepackage{amsthm}
\usepackage{epstopdf}
\usepackage{CJK}
\usepackage{esint}
\usepackage{color}
\usepackage[T1]{fontenc}
\usepackage{subfigure}
\usepackage{amsfonts}
\usepackage{natbib}
\usepackage{footmisc}
\usepackage{verbatim}
\usepackage{multirow}
\usepackage[english]{babel}
\usepackage{url}
\usepackage{bm}
\usepackage{hyperref}
\definecolor{darkblue}{rgb}{0,0,0.5}
\hypersetup{
colorlinks=true,
linkcolor=black,
filecolor=blue,
citecolor=darkblue,  
urlcolor=black,
}

\usepackage{enumitem,kantlipsum}
\usepackage{refcount}

\newcounter{lemmacounter}
\newenvironment{lemma}[1][]{\refstepcounter{lemmacounter}
   {{\em Lemma~\thelemmacounter #1}.---} \rmfamily}
   
\newcounter{theoremcounter}
\newcounter{propositioncounter}
\newcounter{conjecturecounter}
\newcounter{definitioncounter}
\newenvironment{definition}[1][]{\refstepcounter{definitioncounter}
   {{\em Definition~\thedefinitioncounter #1}.---} \rmfamily}

\def\be{\begin{equation}}
\def\ee{\end{equation}}
\def\ba{\begin{eqnarray}}
\def\ea{\end{eqnarray}}
\newcommand{\calN}{{\cal N}}

\newcommand{\calG}{{\cal G}}

\newcommand{\tr}{{\rm Tr}}

\newcommand{\ketbra}[2]{|{#1}\rangle\!\langle{#2}|}
\newcommand{\state}[1]{\ketbra{#1}{#1}}

\newcommand{\eq}[1]{(\hyperref[eq:#1]{\ref*{eq:#1}})}

\renewcommand{\sec}[1]{\hyperref[sec:#1]{Section~\ref*{sec:#1}}}
\newcommand{\thrm}[1]{\hyperref[thm:#1]{Theorem~\ref*{thm:#1}}}
\newcommand{\lemm}[1]{\hyperref[lemm:#1]{Lemma~\ref*{lemm:#1}}}
\newcommand{\pro}[1]{\hyperref[pro:#1]{Proposition~\ref*{pro:#1}}}
\newcommand{\corr}[1]{\hyperref[corr:#1]{Corollary~\ref*{corr:#1}}}
\newcommand{\deff}[1]{\hyperref[deff:#1]{Definition~\ref*{deff:#1}}}
\newcommand{\fig}[1]{\hyperref[fig:#1]{\ref*{fig:#1}}}
\newcommand{\tbl}[1]{\hyperref[fig:#1]{\ref*{tbl:#1}}}

\DeclareMathOperator{\Tr}{Tr}

\begin{document}

\preprint{MIT-CTP/5005}

\title{Convex resource theory of non-Gaussianity}

\author{Ryuji Takagi}
\email{rtakagi@mit.edu}
\affiliation{Center for Theoretical Physics, Massachusetts Institute of Technology, Cambridge, Massachusetts 02139, USA}
\affiliation{Department of Physics, Massachusetts Institute of Technology, Cambridge, Massachusetts 02139, USA}
\author{Quntao Zhuang}
\email{quntao@mit.edu}
\affiliation{Department of Physics, Massachusetts Institute of Technology, Cambridge, Massachusetts 02139, USA}
\affiliation{Research Laboratory of Electronics, Massachusetts Institute of Technology, Cambridge, Massachusetts 02139, USA}
\date{\today}

\begin{abstract} 
Continuous-variable systems realized in quantum optics play a major role in quantum information processing, and it is also one of the promising candidates for a scalable quantum computer. 
We introduce a resource theory for continuous-variable systems relevant to universal quantum computation.
In our theory, easily implementable operations---Gaussian operations combined with feed-forward---are chosen to be the free operations, making the convex hull of the Gaussian states the natural free states.
Since our free operations and free states cannot perform universal quantum computation, genuine non-Gaussian states---states not in the convex hull of Gaussian states---are the necessary resource states for universal quantum computation together with free operations.
We introduce a monotone to quantify the genuine non-Gaussianity of resource states, in analogy to the stabilizer theory. 
A direct application of our resource theory is to bound the conversion rate between genuine non-Gaussian states. Finally, we give a protocol that probabilistically distills genuine non-Gaussianity---increases the genuine non-Gaussianity of resource states---only using free operations and postselection on Gaussian measurements, where our theory gives an upper bound for the distillation rate. In particular, the same protocol allows the distillation of cubic phase states, which enable universal quantum computation when combined with free operations.


\end{abstract} 


\maketitle

\section{Introduction}

Continuous-variable quantum information deals with continuous degrees of freedom, such as position and momentum quadratures, in quantum systems like optical fields or vibration modes. With a close connection to quantum optical experiments~\cite{walls2007quantum}, continuous-variable systems have been an important grounds for quantum information processing~\cite{Weedbrook_2014}, in parallel with discrete-variable systems (qudits).  

Gaussian states and Gaussian operations~\cite{giedke2002characterization,de2015normal,de2017gaussian} especially play important roles in continuous-variable quantum information processing. Despite being in an infinite-dimensional Hilbert space, Gaussian states often enable analytical results in the analysis of quantum information processing, due to their characteristic functions in a Gaussian form. As an example, Gaussian operations---operations that map Gaussian states to Gaussian states---are completely characterized by linear transforms of the mean and covariance~\cite{giedke2002characterization}. Besides the convenience of analytic treatment, it is desirable to restrict ourselves in the Gaussian regime, also because preparation of Gaussian states and application of Gaussian operations are readily accessible in quantum optical experiments~\cite{walls2007quantum}. Not only they are experimentally realizable, but they allow for useful quantum information processing protocols such as quantum teleportation~\cite{vaidman1994teleportation,braunstein1998teleportation,ralpha1998teleportation}, (noisy) quantum cloning~\cite{cerf2000optimal,lindblad2000cloning}, quantum-enhanced sensing~\cite{caves1981quantum,bondurant1984squeezed,tan2008quantum,zhuang2017optimum,zhuang2017distributed} and quantum key distribution~\cite{grosshans2002continuous}.


Unfortunately, such Gaussian schemes are limited in their power of continuous-variable quantum information processing. It has been shown that non-Gaussianity in the form of either non-Gaussian states or non-Gaussian operations are required for entanglement distillation~\cite{eisert2002distilling,giedke2002characterization,fiuravsek2002gaussian,zhang2010distillation}, error correction~\cite{niset2009no}, loophole-free~\footnote{If one trusts the device, then Gaussian states and operations suffices for Bell's inequality violation~\cite{ralph2000proposal,Oliver2018violation}.} violation of Bell's inequality~\cite{banaszek1998nonlocality,banaszek1999testing,filip2002violation,chen2002maximal,nha2004proposed,invernizzi2005effect,garcia2005loophole,ferraro2005nonlocality}, and universal quantum computation~\cite{lloyd1999quantum,bartlett2002universal,ohliger2010limitations,menicucci2006universal}. 
Recently, it has been also shown that any Gaussian quantum resource cannot be distilled if one is restricted to the Gaussian regime~\cite{lami2018gaussian}.  
To quantitatively take into account the necessary resource for these protocols, resource theories of non-Gaussianity have been proposed~\cite{marian2013relative,genoni2008quantifying,genoni2010quantifying,zhuang2018resource}. Resource theories are frameworks for quantifying the amount of resource, with respect to a given set of free states and free operations~\cite{Horodecki2013,Brandao2015,chitambar2018quantum}. Previous proposals establish resource theories in a state-driven manner. They first choose Gaussian states as the free states and define resource measures based on the deviation, e.g. measured by relative entropy, from the set of Gaussian states. Afterwards, Gaussian operations are naturally chosen as the free operations, which preserve the set of free states.

However, some non-Gaussian operations are actually easy to implement. A major class of such operations is the class of operations composed by Gaussian operations with adaptive feed-forward on measurement outcomes of Gaussian measurements. 
They can produce probabilistic mixtures of Gaussian states, which are non-Gaussian due to the non-convexity of the set of Gaussian states. 
In contrast to state-driven theories, setting Gaussian operations with feed-forward as free operations is more suitable, when accessible operations are being considered. 
The convex hull of the Gaussian states is a natural set of free states because it is invariant under the free operations and also the largest set of states generated by them.
It is easy to see that all such free states have non-negative Wigner functions~\footnote{The converse is not true. And various ways to test whether a quantum state is in the such a free set are being considered~\cite{filip2011detecting,genoni2013detecting,park2015testing,hughes2014quantum,jevzek2011experimental,Happ2018}}. This implies that the states or operations outside of the free sets are necessary for universal quantum computation, since quantum circuits involving only states with non-negative Wigner function and operations that cannot create negativity can be efficiently simulated classically~\cite{mari2012positive,Veitch2013}. Thus, a resource theory with this choice of free operations and free states is more relevant to universal quantum computation than previous proposals with the set of free states being Gaussian states. 
Indeed, in the discrete-variable case, similar resource theory of quantum computation has been established by taking the stabilizer operations with feed-forward as the free operations and the convex hull of the stabilizer states as free states~\cite{veitch2014resource,howard2017application}. 



In this paper, we establish a resource theory relevant to universal quantum computation with continuous variable. Its free operations are Gaussian operations combined with feed-forward, which are experimentally easy to implement and generate the convex hull of Gaussian states, which we take as the set of free states. States outside the convex hull of Gaussian states---{\it genuine non-Gaussian states}---are the resource states. As a quantifier of the resourcefulness in quantum computation, we consider {\it the logarithmic negativity of Wigner function}, the logarithm of the integral of the absolute value of the Wigner function. 
It is easily computable, and relevant to universal quantum computation---it bounds the classical simulability~\cite{pashayan2015estimating}.  
We show that it is a valid measure for {\it genuine non-Gaussianity}, which takes zero for the states in convex hull of Gaussian states and satisfies the monotonic property under free operations. 
We compute the logarithmic negativity for a number of common resource states and compare them with respect to their mean photon numbers. We find that number states, cubic phase states~\cite{gottesman2001encoding} as well as recently proposed ON states~\cite{sabapathy2018states} have similar values of logarithmic negativity for fixed mean photon numbers.
We further apply our theory to the protocol implementing the cubic phase gate using the ON state~\cite{sabapathy2018states} and show that their protocol is efficient in terms of genuine non-Gaussianity.

To facilitate the preparation of the resource state, we also provide the first distillation protocol that increases the genuine non-Gaussianity. We exploit a partial homodyne measurement idea considered in~\cite{Suzuki2006, Heersink2006} to non-deterministically extract larger genuine non-Gaussianity comparing to the initial value. Our approach is entirely different from the discrete-variable stabilizer quantum computation resource state (magic state) distillation protocols, which are based on the property of discrete-variable error correcting codes~\cite{Dennis2001,Bravyi2005,Reichardt2005,Bravyi2012,Meier2013,Eastin2013}. Due to this different approach, our protocol only requires a single copy of the input state, unlike usual distillation protocols. In analog to ref.~\cite{veitch2014resource}'s application of resource theory on magic-state distillation, we apply our resource theory to this protocol to give an upper bound of the obtainable logarithmic negativity with respect to the success rate of the protocol.
We verify it by numerically computing the logarithmic negativity for input and output states, when the input states are imperfect cubic phase states obtained by applying a cubic phase gate to finitely squeezed states. Conditioned on success, we find that the protocol not only increases the logarithmic negativity but also increases the fidelity from a better cubic phase state---a state obtained by applying a cubic phase gate on a finitely squeezed state with larger squeezing parameter. Thus, it works as ``state distillation'' as well as  ``genuine non-Gaussianity'' distillation. 

This paper is organized as follows. In Sec.~\ref{sec_pre}, we introduce Gaussian states, Gaussian channels and continuous-variable quantum computation. In Sec.~\ref{sec_framework}, we establish the resource theory framework---free states, free operations and a monotone. In Sec.~\ref{sec_states}, we compare the genuine non-Gaussianity of some resource states. In Sec.~\ref{sec_application}, we apply our resource theory on state conversion and distillation. Finally, we conclude in Sec.~\ref{sec_conclusions} by more discussions and future directions.

\section{Preliminaries}
\label{sec_pre}
\subsection{Gaussian states and Gaussian channels}

We use the notation in ref.~\cite{Weedbrook_2012}. 
An $N$-mode bosonic continuous-variable system is described by annihilation operators $\left\{\hat{a}_k, 1\le k \le N\right\}$, which satisfy the commutation relation $\left[\hat{a}_k,\hat{a}_j^\dagger\right]=\delta_{kj}, \left[\hat{a}_k,\hat{a}_j\right]=0$. One can also define real quadrature field operators $\hat{q}_k=\hat{a}_k+\hat{a}_k^\dagger, \hat{p}_k=i\left(\hat{a}_k^\dagger-\hat{a}_k\right)$ and formally define a real vector of operators
$\hat{x}=\left(\hat{q}_1,\hat{p}_1,\cdots, \hat{q}_N,\hat{p}_N\right)$, which satisfies the canonical commutation relation
$
\left[\hat{x}_i,\hat{x}_j\right]=2i {  \Omega}_{ij}.
$
Here ${  \Omega}=i \bigoplus_{k=1}^N {  Y} $, where ${  Y}$ is the Pauli matrix. The mean photon number (power) in mode $k$ is given by the expectation value of operator $\hat{a}^\dagger_k\hat{a}_k=\left(\hat{p}_k^2+\hat{q}_k^2\right)/4-1/2$.

A quantum state $\hat{\rho}$ can be described by its Wigner characteristic function
\be 
\chi\left({  \xi};\hat{\rho}\right)=\tr \left[\hat{\rho} \hat{D}\left({  \xi}\right)\right],
\label{Wigner_characterisitic_function}
\ee
where $  \xi=\left(\xi_1,\cdots \xi_{2N}\right)\in \mathbb{R}^{2N}$ and 
$
\hat{D}\left({  \xi}\right)=\exp\left(i \hat{x}^T {  \Omega} {  \xi} \right)$
is the Weyl displacement operator for all modes.
Under this convention, the displacement operator on position  $\hat{D}_q\left(x\right)\equiv \exp\left(-i\hat{p}x/2\right)$, which satisfies $\hat{D}_q\left(x\right)\ket{y}=\ket{y+x}$. Similarly, $\hat{D}_p\left(P\right)\equiv \exp\left(i\hat{q}P/2\right)$, which satisfies $\hat{D}_p\left(P\right)\ket{m}=\ket{m+P}$. For a pure state $\ket{\psi}$, for simplicity, we will write $\chi\left({  \xi};\ket{\psi}\right)$ and for other similar cases.



The Wigner function is defined as the Fourier transform of the Wigner characteristic function
\be
W\left({  x};\hat{\rho}\right)=\int \frac{d ^{2N}{  \xi}}{\left(2\pi\right)^{2N}}\exp\left(-i {x}^T {  \Omega} {  \xi}\right) \chi\left({  \xi};\hat{\rho}\right).
\ee
Note by definition, both the Wigner function and the Wigner characteristic function are linear in $\hat{\rho}$. 

A state $\hat{\rho}$ is Gaussian if its characteristic function has the Gaussian form
\be
\chi \left({  \xi};\hat{\rho}\right)=\exp\left(-\frac{1}{2}{  \xi}^T \left({  \Omega} {  \Lambda } {  \Omega}^T\right){  \xi}-i \left({  \Omega} \overline{  x}\right)^T {  \xi}\right).
\ee
Here the $\overline{  x}=\braket{  \hat{x}}_{\hat{\rho}}$ is the state's mean and 
$
{  \Lambda}_{ij}=\frac{1}{2} \braket{\{\hat{x}_i-\overline{x}_i,\hat{x}_j-\overline{x}_j\}}_{\hat{\rho}}$
is its covariance matrix, where $\{,\}$ is the anticommutator and $\braket{\hat{A}}_{\hat{\rho}}\equiv\tr \left(\hat{A}\hat{\rho}\right)$ for operator $\hat{A}$. 

Gaussian channels are complete-positive and trace-preserving (CPTP) maps that map any Gaussian state to a Gaussian state~\cite{giedke2002characterization,de2015normal,de2017gaussian}. They can be extended to Gaussian unitaries on the input and a vacuum environment (Stinespring dilation)~\cite{giedke2002characterization}, therefore we focus on Gaussian unitaries. 
A Gaussian unitary $\hat{U}_{  S,  d}$ transforms
\be
\hat{U}_{  S,  d}^\dagger \hat{  x} \hat{U}_{  S,  d}=  S \hat{  x}+{  d},
\ee
where ${  d}=\left(d_1,\cdots, d_{2N}\right)$ is the displacement and $  S$ is a matrix.
Commutation relation preserving requires that $\left[\sum_m{  S}_{am}{x}_m+{ d}_a, \sum_n{  S}_{bn}{x}_n+{d}_b\right]=\sum_{mn}{  S}_{am}(2i){  \Omega}_{mn} {  S}_{bn}=2i{  \Omega}_{ab }$, thus
$
{  S}{  \Omega} {  S}^T={  \Omega}.
$
Because ${  S}{  \Omega} \left({  S}^T{  \Omega} {  S}\right)=\left({  S}{  \Omega} {  S}^T\right){  \Omega} {  S}={  \Omega}^2 {  S}=-{  S}$, we have
$
{  S}^T {  \Omega} {  S}={  \Omega}.
$
Since $\det\left({  \Omega}\right)=1$, we have the Jacobian of the linear transform $|\det\left({  S}\right)|=1$. Also we have $\hat{U}_{{  S},{  d}}^\dagger=\hat{U}_{{  S}^{-1},-{  S}^{-1}{  d}}$. Here we give some examples of Gaussian unitaries: single-mode squeezing $\hat{S}\left(s\right)=\exp\left(s\left(\hat{a}^2-\hat{a}^{\dagger2}\right)/2\right)$, where $s$ is the squeezing strength; phase rotation $\hat{R}\left(\theta\right)=\exp\left(-i \theta \hat{a}^\dagger \hat{a}\right)$; displacement operator which has already been introduced.

\begin{lemma}\label{lemm:Wigner_Gaussian_unitary}
Gaussian unitary corresponds to a linear coordinate transform for the Wigner characteristic function and the Wigner function.
\ba
\chi \left({  \xi};\hat{U}_{  S,  d} \hat{\rho} \hat{U}_{  S,  d}^\dagger\right)&=&\chi \left({  S}^{-1}{  \xi};\hat{\rho}\right) \exp\left(i {  d}^T {  \Omega   \xi }\right),
\\
W\left(  x; \hat{U}_{  S,  d}\hat{\rho} \hat{U}_{  S,  d}^\dagger \right)&=&W\left(  S^{-1}\left(  x-  d\right);\hat{\rho}\right).
\ea
\end{lemma}
The proof is attached in Appendix~\ref{proof_lemma1}. In discrete-variable systems, it is known that the Clifford unitaries are permutations on the discrete Wigner functions~\cite{veitch2014resource}. \lemm{Wigner_Gaussian_unitary} shows the analog between Gaussian unitaries and Clifford unitaries. 

In following, we list the well-known properties of Wigner functions.

\begin{lemma}
\label{lemm:Wigner}
Wigner functions satisfy the following.
\begin{enumerate}[label=\text{(\ref{lemm:Wigner}.\arabic*)},wide, labelwidth=!,labelindent=0pt]
   \item  \label{Wigner_trace}
    Consider bipartite state $\hat{\rho}_{AB}$ with two parts $A$ and $B$. Let $x_A, x_B$ denote the variable in Wigner function associated with each part. Then 
    $
    W\left(x_B; {\rm Tr}_A \left(\hat{\rho}_{AB}\right) \right)= \int d x_A W\left(x_A,x_B; \hat{\rho}_{AB}\right).
    $ 
   \item  \label{Wigner_measurement}
    The probability of homodyne measurement on the $m$-th mode described by $\{\state{q_m}\}$ has result $q_m$ with probability
    $
    P_m\left(q_m\right)=\int d p_m \prod_{k\neq m}d^2x_k W\left(x; \hat{\rho} \right).
    $ 
    Similar results hold for projection on multiple modes.
   \item  \label{Wigner_product}
    Consider product state $\hat{\rho}_{AB}$ with two parts $A$ and $B$. Then 
    \be
    W\left(x_A,x_B; \hat{\rho}_{A}\otimes \hat{\rho}_{B} \right)= W\left(x_A; \hat{\rho}_{A} \right) W\left(x_B;  \hat{\rho}_{B} \right).
    \ee
   \item \label{Wigner_post}
    Suppose one performs a homodyne measurement on a single mode A and get the result $\tilde{q}_A$, the Wigner function of the post-measurement state $\hat{\rho}_{B|\tilde{q}_A}={\braket{\tilde{q}_A|\hat{\rho}_{AB}|\tilde{q}_A}}/\tr_B\left( \braket{\tilde{q}_A|\hat{\rho}_{AB}|\tilde{q}_A}\right) $ is
    \be 
    W\left(x_B; \hat{\rho}_{B|\tilde{q}_A}\right)=
    \frac{\int dp_A W\left(p_A, \tilde{q}_A,x_B; \hat{\rho}_{AB}\right)}{\int d x_B\int dp_A W\left(p_A, \tilde{q}_A,x_B; \hat{\rho}_{AB}\right)}.
    \ee 
    Similar results can be obtained for measuring multiple modes, by considering the measurement sequentially.
   \item \label{Wigner_trace_states}
   The trace of the product of two single-mode quantum states can be evaluated as follows,
   $
    \tr\left(\hat{\rho} \hat{\sigma}\right)=4\pi \int d x W\left(x;\hat{\rho}\right)W\left(x;\hat{\sigma}\right).
   $
   \item \label{Wigner_xspace}
   For the single-mode pure state with position space wave function $\psi\left(q\right)$, we have
   \be
   W\left(q,p; \ket{\psi} \right)=\frac{1}{2\pi}\int_{-\infty}^\infty dy \psi^\star \left(q-y\right)\psi\left(q+y\right) e^{-i p y}.
   \ee
  \end{enumerate}
\end{lemma}

Note that \ref{Wigner_xspace} is different from ref.~\cite{ghose2007non}, and \ref{Wigner_trace_states} has a factor of $4\pi$, because of the choice of $\hbar=2$, which is the convention in ref.~\cite{Weedbrook_2012}. To clarify these two points, we show their proof in Appendix~\ref{Convention_proof}. The proof of the other properties are well-known, straightforward and not presented here.

\subsection{Continuous-variable quantum computation}
\label{intro:cv_computation}
We consider the definition of continuous-variable universal quantum computation in ref.~\cite{lloyd1999quantum}, i.e. a set of operations is universal if by a finite number of applications of operations in the set, one can approach arbitrarily close to any unitary evolution generated by Hamiltonians of a polynomial form in operators $\hat{q}_k$'s and $\hat{p}_k$'s. Under this definition of universality, ref.~\cite{lloyd1999quantum} also shows that Gaussian operations alone are not universal, since Gaussian unitaries corresponds to generators of second order polynomials in $\hat{q}_k$'s and $\hat{p}_k$'s. However, an arbitrary extra unitary with generators of higher order than two, in addition to Gaussian operations, will be universal. Based on this finding, ref.~\cite{sefi2011decompose} developed a systematical way of performing the decomposition of any unitary generated by polynomial Hamiltonian to a basic set of Gaussian unitaries $\left\{e^{i\pi\left(\hat{p}^2+\hat{q}^2\right)/2}, e^{i t_1\hat{q}}, e^{i t_2\hat{q}^2} \right\}$ and the cubic phase gate
\be
\hat{V}(\gamma)=e^{i\gamma \hat{q}^3}.
\ee
The choice of the non-Gaussian unitary is not unique and cubic phase gate, generated by $\hat{q}^3$, is basically one of the most simple non-Gaussian unitaries, in the sense that it is generated by the lowest polynomial of order higher than two.

There are a number of experimental proposals of realizing the cubic phase gate, involving genuine non-Gaussian resource states and Gaussian operations combined with feed-forward~\cite{gottesman2001encoding,marek2011deterministic,ghose2007non,sabapathy2018states}. Ref.~\cite{gottesman2001encoding} (GKP scheme) uses the cubic phase state
\ba
 \ket{\gamma} = \hat{V}(\gamma)\ket{0}_p = \int dq e^{i\gamma q^3}\ket{q}
\ea
as the resource state, where $\ket{0}_p$ is the zero-momentum state at the infinite squeezing limit and unnormalizable. They also provide a scheme for preparing approximate cubic phase states by two-mode squeezing, displacement and photon number counting. Ref.~\cite{ghose2007non} analyzes the above approximate preparation scheme in details, and provide alternative schemes with number state $\ket{N}$ as the resource state. Ref.~\cite{marek2011deterministic} provides a scheme by sequential photon number subtractions and displacement to produce approximate weak cubic phase states, which is experimentally implemented in ref.~\cite{yukawa2013emulating}. Ref.~\cite{sabapathy2018states} further introduces the ON state 
\be
\ket{ON}=\frac{1}{\sqrt{1+|a|^2}}\left(\ket{0}+a\ket{N}\right)
\label{eq:ON}
\ee
as a replacement of cubic phase state to realize the cubic phase gate.  

Ref.~\cite{arzani2017polynomial} shows that instead of genuine non-Gaussian states as the resource, it is possible to have non-Gaussian measurements such as photon number counting to enable cubic phase gates. This is similar to the idea of Gaussian cluster state measurement-based quantum computation~\cite{menicucci2006universal}, where the measurement, e.g. photon number counting, is non-Gaussian. In this paper, we focus on the case where genuine non-Gaussian states are the necessary resource for universal quantum computing combined with Gaussian operations and feed-forward.

\section{Resource theory framework}
\label{sec_framework}

\subsection{Free states and free operations}
Resource theories are frameworks that deal with quantification and manipulation of a quantity that is considered resource under some setting. The states that are not resourceful are called free states and operations that cannot create any resource from free states are called free operations. The set of free states is invariant under the free operations. 
In some situations where one is more interested in the property of states, it is natural to define the set of free states first and choose a corresponding set of free operations. An example of such resource theories is the theory of coherence~\cite{aberg2006quantifying, Baumgratz2014}, where free states are naturally defined as incoherent states---the set of diagonal states with respect to some fixed basis. Depending on the focus, one can investigate various sets of free operations. For example, in coherence theory, many sets of free operations have been considered such as incoherent operations~\cite{Baumgratz2014}, strictly incoherent operations~\cite{winter2016operational, Yadin2016}, maximally incoherent operations~\cite{aberg2006quantifying, Chitambar2016} etc., and each set has its own resource quantifier and restriction on state transformation rules. 

On the other hand, for other settings where one is more interested in physically available operations, it is more natural to define the set of free operations as easily accessible operations and choose a set of free states invariant under the application of free operations. An example of such setting is the case when two parties are restricted to the set of local operation and classical communications (LOCC)~\cite{nielsen1999conditions}. As we set LOCC as free operations, separable states are naturally found to be a set of free states and other states---called entangled states---are the resource states. 

Here, we take the latter view to consider continuous-variable systems relevant to quantum optics. In such a setting, easily accessible operations are Gaussian operations and homodyne measurements. It is also reasonable to allow feed-forward---controlling further operations conditioned on the outcome of the measurements. We call the above operations combined as {\it Gaussian protocols} and we take them as our free operations. 

\begin{definition}
\label{deff:Gprotocol}
  An operation is a Gaussian protocol if it is composed by the following operations. 
  \begin{enumerate}
  \renewcommand{\labelenumi}{(\arabic{enumi}).}
   \item  \label{item:unitary}
  Gaussian unitaries, $\hat{\rho}\to \hat{U}_{S,d}\hat{\rho}\hat{U}_{S,d}^\dagger$.
   \item  \label{item:ancilla}
   Composition with ancillary vacuum states, $\hat{\rho}\to \hat{\rho}\otimes \state{0}$
   \item  \label{item:homodyne}
   Homodyne measurement described by POVM $\{\state{q}\}$.
   \item \label{item:partial}
   Partial trace.
   \item \label{item:conditional}
   The above quantum operations conditioned on the outcomes of homodyne measurements.
  \end{enumerate}
\end{definition}

Because general Gaussian measurement is equivalent to first performing a general Gaussian unitary and then homodyne~\cite{giedke2002characterization}, one only needs to consider homodyne measurement. Note that the above set of operations includes performing the above operations conditioned on classical randomness. This is because one can generate classical randomness by performing homodyne measurements on ancillae. By dividing the measurement outcomes on ancillae into proper regions, random numbers with an arbitrary probability distribution can be generated. The above definition of Gaussian protocols based on elementary operations is equivalent to the following definition. 

\begin{lemma} \label{lemm:Gprotocol2}
The set of Gaussian protocols defined in \deff{Gprotocol} is equivalent to the set of operations composed by the following.
\begin{enumerate}
\renewcommand{\labelenumi}{(\arabic{enumi}).}
 \item
  Composition with ancillary vacuum states, $\hat{\rho}\to \hat{\rho}\otimes \state{0}$
 \item 
  Homodyne measurement on one mode and conditional Gaussian channel on the other modes, $\int dq \Phi_q \otimes M_q$ where $\{\Phi_q\}$ are Gaussian channels and $M_q(\cdot) = \state{q}\cdot \state{q}$.  
\end{enumerate}
\end{lemma}

\begin{proof}
Let $\mathcal{O}_1$ be the set of Gaussian protocols defined in \deff{Gprotocol} and $\mathcal{O}_2$ be the set of operations realized by the above two elementary operations.  
Since (\ref{item:unitary}), (\ref{item:ancilla}), and (\ref{item:partial}) in \deff{Gprotocol} can realize any Gaussian channel in dilution form, together with (\ref{item:conditional}) we get $\mathcal{O}_2\subseteq \mathcal{O}_1$. To show $\mathcal{O}_1 \subseteq \mathcal{O}_2$, notice that (\ref{item:unitary}), (\ref{item:homodyne}), and (\ref{item:partial}) in \deff{Gprotocol} are Gaussian channels. They can be realized by attaching ancillary vacuum state, making a homodyne measurement on the vacuum state, and applying the corresponding Gaussian channels no matter what the measurement result is. (\ref{item:conditional}) can be straightforwardly realized again noticing that (\ref{item:unitary})-(\ref{item:partial}) are Gaussican channels. 
\end{proof}

Now that we set up our free operations, we shall find a set of free states that is invariant under free operations. The set of Gaussian states is not appropriate for it due to its non-convexity property. A simple probabilistic mixture of Gaussian operations, which are free operations, may transform a Gaussian state to non-Gaussian states composed of a mixture of Gaussian states. Indeed, an appropriate set of free states is the convex hull of Gaussian states.

\begin{definition}
Denote the set of Gaussian states as $\calG$, the states outside are called non-Gaussian. We consider the convex hull  
\be
\bar{\calG}=\left\{\int d\lambda  P_\lambda \hat{\rho}_\lambda\mid \hat{\rho}_\lambda\in \calG, P_\lambda \geq 0, \int d\lambda  P_\lambda=1 \right\}.
\ee 
We call states in $\bar{\calG}$ as convex-Gaussian states and states outside $\bar{\calG}$ as genuine non-Gaussian states. 
\end{definition}

It is clear by definition that any two convex-Gaussian states are connected by some Gaussian protocol and $\bar{\calG}$ is the largest set of states that can be prepared by Gaussian protocols from Gaussian states. The following lemma greatly simplifies the expression of $\bar{\calG}$. 

\begin{lemma} 
The convex hull of all Gaussian pure states equals $\bar{\calG}$.
\end{lemma}

\begin{proof}
This lemma is mentioned in ref.~\cite{genoni2013detecting} without proof. For completeness, we make the proof explicit. Any Gaussian state can be transformed to a thermal state by Gaussian unitary. For single mode case, $\hat{\rho}_\lambda=\hat{U}_\lambda \hat{\sigma}^{\rm th}_\lambda \hat{U}_\lambda=\int d^2\alpha p_\alpha \hat{U}_\lambda\ket{\alpha}\bra{\alpha}\hat{U}_\lambda^\dagger $, where $p_\alpha>0$ is the P-function of the thermal state and $\ket{\alpha}$ is the coherent state. The multi-mode case is similar.
\end{proof}

All states in $\bar{\calG}$ have non-negative Wigner functions. Denote the set of sates with non-negative Wigner functions as $W_+$~\footnote{We include states of Wigner function with negative values at measure zero points in this set. These measure zero negative points should be irrelevant to any experiments in reality. }. We have~\cite{filip2011detecting,genoni2013detecting} 
$
\calG \subsetneq \bar{\calG} \subsetneq W_+.
$

This choice of free operations and free states is the most relevant to universal quantum computation with continuous-variable systems. 
Ref.~\cite{mari2012positive} showed that when the initial state and the Choi matrix of all operations are in $W_+$, the quantum computation can be classically simulated efficiently. Genuine non-Gaussian states (with negative Wigner function) are naturally considered resource for universal quantum computation. Note states in $W_+\setminus\bar{\calG}$ are analogs of bound genuine non-Gaussian states, which do not enable universal quantum computation together with free states and free operations.

\subsection{Monotone}

Besides free states and free operations, another important concept of resource theories is monotones. Monotones are maps from a quantum state to a real number that is meant to quantify how resourceful the state is. Here, we are interested in quantum states acting on an infinite-dimensional Hilbert space. If a quantum state describes a $N$-mode bosonic system, the state is expressed by a density operator $\hat{\rho}$ acting on the Hilbert space $\mathcal{H}^{\otimes N}$. We formally define monotones as following.

\begin{definition}
 Let $D(\mathcal{H}^{\otimes N})$ be a space of the density operators for $N$-mode bosonic systems and  $\mathcal{M}_N:D(\mathcal{H}^{\otimes N})  \rightarrow \mathbb{R}$ be a map from the set of quantum states for $N$-bosonic systems to real numbers. Let $\mathcal{M}(\hat{\rho}) \equiv \mathcal{M}_N (\hat{\rho})$ for $\hat{\rho} \in D(\mathcal{H}^{\otimes N})$ so that $\mathcal{M}$ automatically takes into account the number of the modes in the system that the input state belongs to. 
 $\mathcal{M}$ is called genuine non-Gaussian monotone if it satisfies the following.
 \begin{enumerate}
  \item 
 $\mathcal{M}(\hat{\rho}) \geq 0$ and $\mathcal{M}(\hat{\rho}) = 0,\ \forall \hat{\rho} \in \bar{\calG}$  
 \item 
  It does not increase under a Gaussian protocol. To clarify what it means, consider the elementary operations for the Gaussian protocols in \lemm{Gprotocol2}. We say that $\mathcal{M}$ does not increase under a Gaussian protocol if (1) $\mathcal{M}(\hat{\rho})\geq \mathcal{M}(\hat{\rho} \otimes \state{0})$ and (2) $\mathcal{M}(\hat{\rho}) \geq  \mathcal{M}\left(\int dq P(q)\hat{\sigma}_q\right)$ and $\mathcal{M}(\hat{\rho}) \geq \int dq P(q) \mathcal{M}\left(\hat{\sigma}_q\right)$ where $P(q) = \Tr[(\hat{I} \otimes \hat{M}_q) \hat{\rho}]$ and $ \hat{\sigma}_q = \Phi_q \otimes M_q (\hat{\rho})/{P(q)}$. 
 \end{enumerate}
 \label{deff:mono_average}
\end{definition}

The last inequality is sometimes called selective monotonicity because it states that expectation value after a selective measurement cannot increase. It is not a very standard requirement for other resource theories; it is more common to only require the monotonicity under free operations. Nevertheless, we take the selective monotonicity as a requirement as well because it will be relevant to a non-deterministic distillation protocol which we shall discuss later. 
Although there are many candidates for monotones, in this paper we focus on the {\it logarithmic negativity of the Wigner function}.  
We pick it up because it is easily computatble and it is a relevant resource measure for universal quantum computation. The definition and properties are given as follows (proof in Appendix.~\ref{proof_mana})

\begin{lemma}
\label{lemm:mana}
Logarithmic negativity of the Wigner function
\be 
\calN_L\left(\hat{\rho}\right)=\ln \int d^{2N} x |W\left(x; \hat{\rho}\right)|.
\ee 
satifies the following properties.

\begin{enumerate}[label=\text{(\ref{lemm:mana}.\arabic*)},wide, labelwidth=!,labelindent=0pt]

\item \label{mana_unitary}
Invatiant under Gaussian unitaries.

\item \label{mana_partial_trace}
Non-increasing under partial trace.

\item \label{mana_additive}
Additive. $\calN_L\left(\hat{\rho}_{A}\otimes \hat{\rho}_{B}\right)=\calN_L\left(\hat{\rho}_{A}\right)+\calN_L\left( \hat{\rho}_{B}\right)$.

\item \label{mana_zero}
$\calN_L\left(\hat{\rho}\right)=0$, iff $\hat{\rho}\in W_+$. 

\item \label{mana_channel}
Non-increasing under Gaussian channel $\Phi_\calG$. $\calN_L\left(\Phi_\calG\left(\hat{\rho}\right)\right)\le \calN_L\left(\hat{\rho}\right)$.


\item  \label{mana_prob}
Non-increasing under free operations in the sense of \deff{mono_average}. 

\end{enumerate}
\end{lemma}
Since $\bar{\calG} \subsetneq W_+$ and states in $W_+\setminus\bar{\calG}$ have zero logarithmic negativity, logarithmic negativity is not a faithful monotone. However, states in $W_+\setminus \bar{\calG}$ cannot be resources for universal quantum computation, and should be considered as bound genuine non-Gaussian states, as the analog of bound entangled states~\cite{horodecki1998mixed,divincenzo2000evidence} and bound magic states~\cite{Veitch2012}, where the distillable resource is zero. 
In general, although logarithmic negativity is usually not analytically calculable, it can be easily obtained numerically when the Wigner function is available analytically.

\section{Resource states}
\label{sec_states}
\begin{figure}[htbp]
    \centering
    \includegraphics[width=0.3\textwidth]{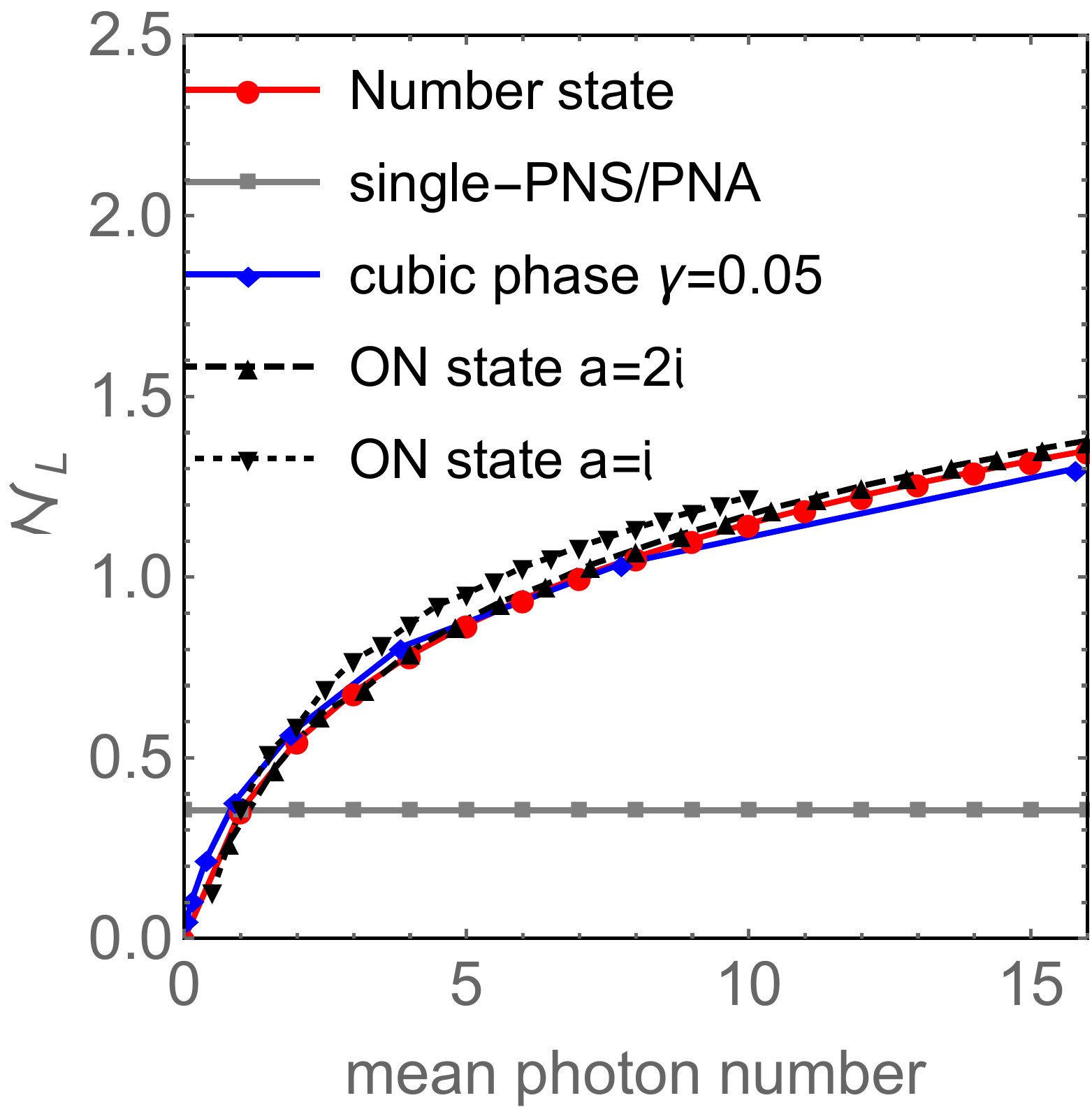}
    \caption{Genuine non-Gaussianity measured by logarithmic negativity of some resource states. PNS: photon-number subtraction. PNA: photon-number addition. Note that for single-PNS/PNA, we are plotting the maximum genuine non-Gassianity.} 
    \label{fig:NLall}
\end{figure}

As discussed in section~\ref{intro:cv_computation}, cubic phase states, number states and ON states can be used as resource states to facilitate continuous-variable universal quantum computation, when combined with Gaussian protocols. In this section, we will compute their genuine non-Gaussianity measured by the logarithmic negativity; we will also compare their genuine non-Gaussianity when each state has the same mean photon number. For completeness, we also consider photon-number added or subtracted states, which have been shown to improve the entanglement of two-mode Gaussian states~\cite{navarrete2012enhancing}. 
We numerically obtain the logarithmic negativity in Fig.~\fig{NLall}. Surprisingly, number states, ON states, and cubic phase states with the same mean photon number (same energy) have very close amount of the negativity even though the description of the cubic phase states looks very different from the other two.
Below, we make comments on these resource states to clarify the meaning of Fig.~\fig{NLall}.

{1. Number states}. Photon number state $\ket{N}$ is the most common non-Gaussian source. And it has been known that for fixed photon number, number states and its superpositions maximize the non-convex version of non-Gaussianity~\cite{genoni2010quantifying,genoni2008quantifying}. The Wigner function is given in Appendix~\ref{app_wigners}.

{2. Single photon added and subtracted states}. Ideal single photon-number addition (PNA) operation and photon-number subtraction (PNS) operation can be described by the annihilation and creation operators $\hat{a}$ and $\hat{a}^\dagger$~\cite{kim2008scheme,navarrete2012enhancing}. Experimental schemes of PNS and PNA can be found in Refs.~\cite{parigi2007probing,fiuravsek2009engineering,marek2008generating,kitagawa2006entanglement,namekata2010non,fiuravsek2005conditional,wakui2007photon}. Conditioned on success, they map a pure state $\ket{\psi}$ to another pure state $\propto\hat{a}\ket{\psi}$ (PNS) or $\propto\hat{a}^\dagger\ket{\psi}$ (PNA). For single-photon-subtracted or added zero-mean Gaussian state, the Wigner functions are analytically calculable~\cite{walschaers2017entanglement,walschaers2017statistical} (details in Appendix~\ref{app_wigners}). Here we consider the zero-mean single-mode photon-number subtracted or added pure state with the maximum logarithmic negativity. It suffices to consider PNS or PNA on the squeezed pure state $\ket{\theta,s}=\hat{R}\left(\theta\right)\hat{S}\left(s\right)\ket{0}$. Note that $\hat{a}\hat{R}\left(\theta\right)=e^{i\theta}\hat{R}\left(\theta\right)\hat{a}$ and that Gaussian unitary does not change logarithmic negativity, we only need to consider $\theta=0$. Afterwards, a simple numerical calculation shows that for all values of $s$, both the single-photon-added state and single-photon-subtracted state has
\be
\calN_L=\ln\left(4/\sqrt{e}-1\right)\simeq 0.354.
\ee
This equals $\calN_L\left(\ket{1}\right)$ exactly, which coincides with the conclusion by using non-convex version of non-Gaussianity~\cite{genoni2007measure}.

{3. Cubic phase states}.
The ideal cubic phase state cannot be normalized, thus is not physical. The unnormalized Wigner functions of $\ket{\gamma}$ is given in Appendix~\ref{app_wigners}.
We instead consider an imperfect cubic phase state created by applying the cubic phase gate to the finitely squeezed state with squeezing parameter $s$, namely,
\ba
 \ket{\gamma,P,s}  = \hat{V}(\gamma)\ket{P,s}_p,
 \label{eq:cubic_phase_definition}
\ea
where $\ket{P,s}_p=\hat{D}_p\left(P\right)\hat{S}\left(-s\right)\ket{0}$. The wave function of $\ket{\gamma, P, s}$ in q-space is given by
\begin{align}
 \psi(q) = (2\pi e^{2s})^{-1/4}\exp\left(i\gamma q^3 - \frac{q^2}{4e^{2s}}+i\frac{P q}{2}\right).
 \label{eq:wfunc_cubic}
\end{align}
As the squeezing strength $s$ increases, it approaches the ideal cubic phase state.
The Wigner function of $\ket{\gamma,P,s}$ is given in Appendix~\ref{app_wigners}. The mean photon number can be analytically obtained (details also in Appendix~\ref{app_wigners}) as
\be
N_S=\frac{1}{2}\left(\cosh(2s)-1\right)+18 \gamma^2 e^{4s}+\frac{1}{4}\left(P+6\gamma e^{2s}\right)^2.
\label{eq:Ns_cubic_phase}
\ee
We consider the state with $P=-6\gamma e^{2s}$ in Fig.~\fig{NLall}, which has the minimum mean photon number, when we compare the genuine non-Gaussianity.

{4. ON states}. The wavefunction is given in Eq.~\eq{ON}. The Wigner functions of ON states are given in Appendix~\ref{app_wigners}. The mean photon number of the ON state is $\frac{|a|^2}{1+|a|^2}N$.

\section{Applications}
\label{sec_application}

\subsection{State conversions}
Consider a state transformation $\hat{\rho} \rightarrow \int dq P(q) \hat{\sigma}_q$ under a Gaussian protocol where $P(q) = \Tr\left[\Phi_q \otimes M_q(\hat{\rho})\right]$ in the sense of \lemm{Gprotocol2}. The monotonicity of the logarithmic negativity $\calN_L(\hat{\rho}) \geq \int dq P(q) \calN_L (\hat{\sigma}_q)$ serves as a necessary condition for such a transformation to be possible under a Gaussian protocol. As an example, let us take a closer look at the protocol in ref.~\cite{sabapathy2018states} which implements a cubic phase gate using ON state, which is defined in Eq.~\eq{ON}.  It starts with input state $\ket{\psi}$ and ancillary ON state $\ket{ON}$ and apply the continuous version of controlled-NOT gate from $\ket{\psi}$ to $\ket{ON}$. The homodyne measurement is made on the ancillary system and Gaussian feed-forward operation is applied to the other system. 
They showed that when $N=3$ and $|a|\ll 1$ where $a$ is the parameter of $\ket{ON}$ in Eq.~\eq{ON}, the output state can be approximated by $\frac{1}{P(\tilde{q})}\hat{A}_{\tilde{q}} \hat{V}(\gamma)\ket{\psi}$ where $\tilde{q}$ is the outcome of the homodyne measurement, $P(\tilde{q})$ is the probability density of obtaining the outcome $P(\tilde{q})$, and $\hat{A}_{\tilde{q}}=\exp\left[-(\hat{q} +\tilde{q})^2/4\right]$ is a noise factor. $\gamma$ is related to $a$ by $\gamma = \left(-i/\sqrt{6}\right)a$. 
It can be seen as a state transformation under a Gaussian protocol with $\hat{\rho}$ being $\ket{\psi}\otimes \ket{ON}$ and $\sigma_{\tilde{q}}$ being $\frac{1}{P(\tilde{q})}\hat{A}_{\tilde{q}} \hat{V}(\gamma)\ket{\psi}$, so the monotonicity relation applies. 
For instance, suppose $\ket{\psi}$ is the infinitely-squeezed state. Since the squeezed state has zero logarithmic negativity and the logarithmic negativity is additive, the initial logarithmic negativity is $\calN_L(\ket{03})$ where $\ket{03}$ is the ON state with $N=3$. 
$\hat{\sigma}_{\tilde{q}}$ in this case is proportional to $\int dq \exp\left[-(q+\tilde{q})^2/4 + i\gamma q^3\right]\ket{q}$. 
In ref.~\cite{sabapathy2018states}, they examined $P(\tilde{q})$ for various $\ket{\psi}$ and observed that for squeezed states the distributions are Gaussian-like and become flatter as the squeezing level increases.
We obtained that for $\gamma = 0.1$, $\calN_L(\ket{03}) = 0.11$ and $\calN_L(\hat{\sigma}_{\tilde{q}}) \sim 0.09$ for various $\tilde{q}$ we examined. It not only confirms the monotonicity relation but asserts that it is quite an efficient protocol in terms of genuine non-Gaussianity in this case. The monotonocity relation would be useful in general to get insights to the relationship between output states and corresponding probability density for various input states $\ket{\psi}$. 

If $\ket{\psi}$ is a squeezed state and the outcomes $\tilde{q} \sim 0$ are postselected, it works as a non-derterministic conversion to the imperfect cubic phase state in Eq.~\eq{wfunc_cubic} and the success probability of such a protocol is bounded by the ratio between the initial logarithmic negativity to the output logarithmic negativity. 
 As we shall see in the next section, the imperfect cubic phase state in Eq.~\eq{wfunc_cubic} may be further purified by another protocol, and the level of purification and the success probability have a similar constraint due to the monotonicity of the logarithmic negativity. 
 
 Finally, although in ref.~\cite{sabapathy2018states} only the ON state of $N=3$ is discussed for the implementation of the cubic phase gate, it is possible that larger $N$ would be helpful to realize a better cubic phase gate since it may allow tuning up to higher order terms. The logarithmic negativity of the ON states for various $N$ in Fig.~\fig{NLall} serves as a bound for quality of the implemented gate and success probability of such protocols.   

\subsection{Distillation of genuine non-Gaussianity}
Although genuine non-Gaussianity is necessary for some tasks such as universal quantum computation, it is usually hard to prepare states with large genuine non-Gaussianity. One can then ask whether it is possible to non-deterministically increase the genuine non-Gaussianity of states only by Gaussian protocols and postselection. Here, we provide such a protocol which only consists of a beam splitter, homodyne measurement and postselection. It is so simple that it should be readily implementable in experiments. The protocol is based on the partial homodyne measurement that allows for distillation of coherent-state superpositions~\cite{Suzuki2006} and distillation of squeezing under non-Gaussian noise~\cite{Heersink2006}.
Intuitively, it works as a filter function that focuses on some region of the Wigner plane and reduce the contributions from the other regions. By tuning the strategy of postselection, one can somewhat engineer the output state. We apply this idea to increase the genuine non-Gaussianity by focusing on the region on the Wigner plane that has more negativity than the other. It turns out that the same protocol also allows for distillation of cubic phase state in the sense that it increases the fidelity of imperfect cubic phase states from a ``better'' cubic phase state, which may be of interest on its own.

\begin{figure*}
    \centering
    \subfigure[]{
    \includegraphics[width=0.22\textwidth]{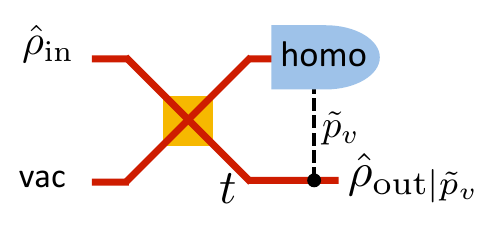}
    \label{fig:distillation}
    }
    \subfigure[]{
    \includegraphics[scale=0.19]{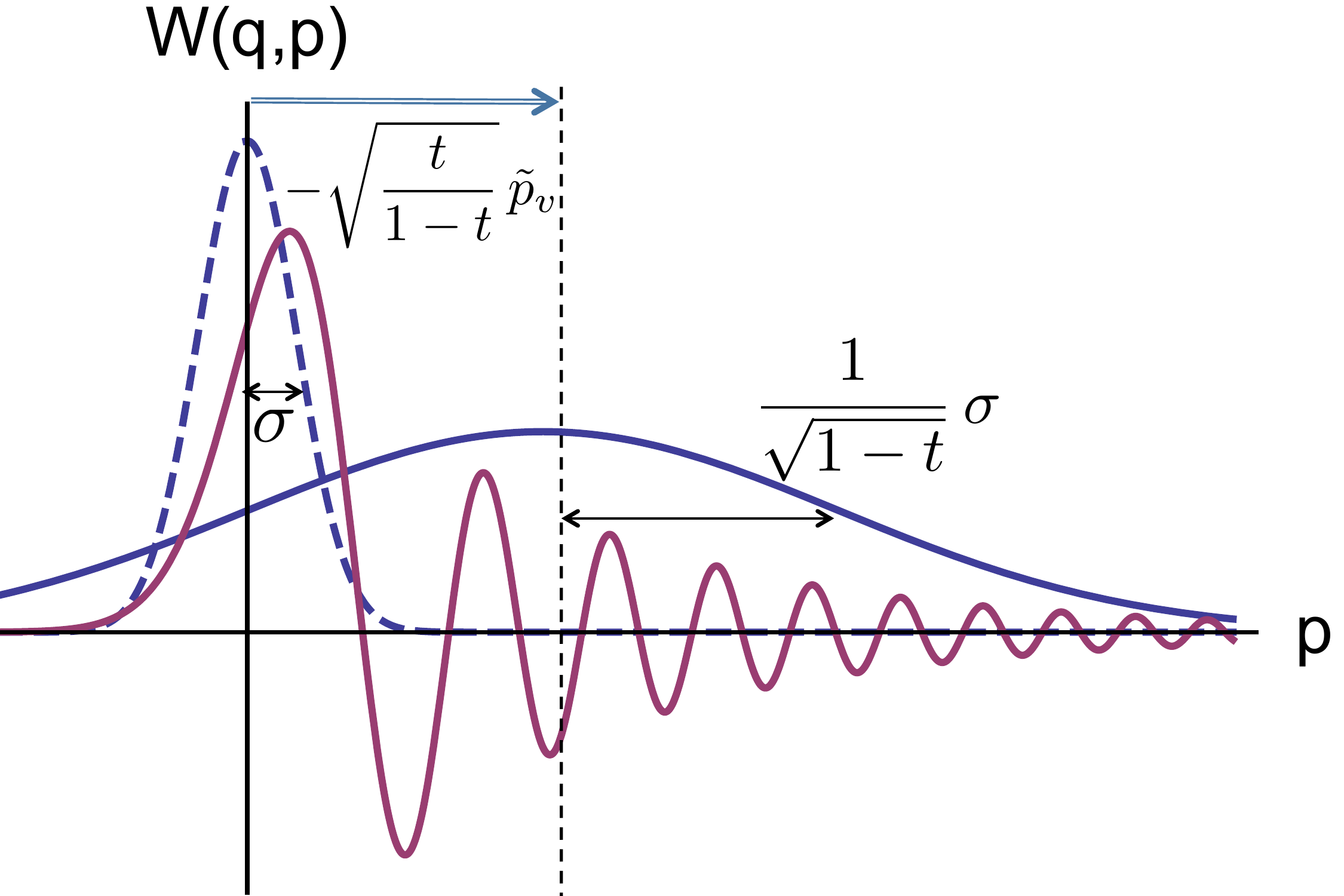}
    \label{fig:homodyne_mechanism}
    }
    \centering
    \subfigure[]{
    \includegraphics[width=0.19\textwidth]{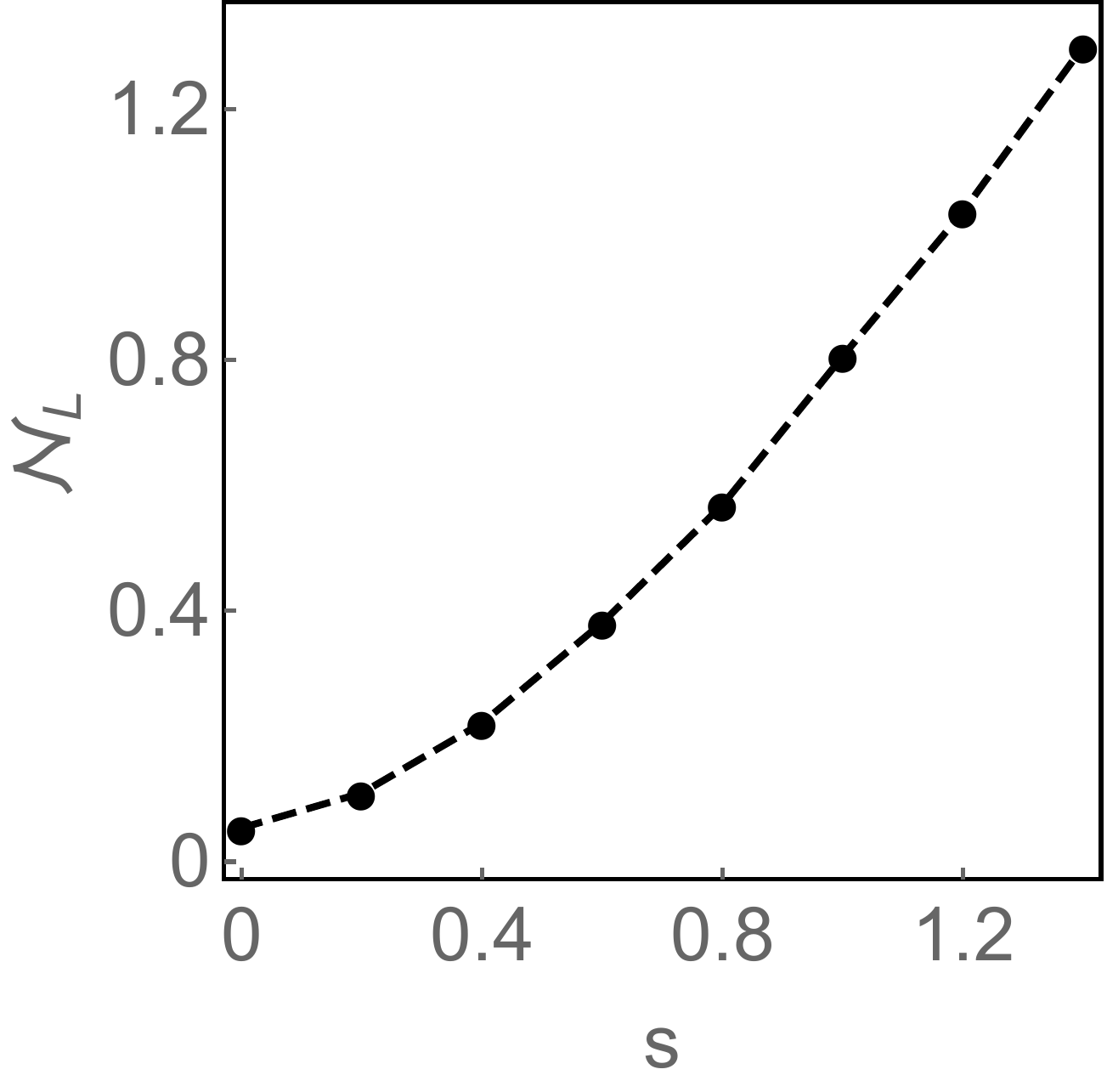}
    \label{fig:NL_s}
    }
    \centering
    \subfigure[]{
    \includegraphics[width=0.23\textwidth]{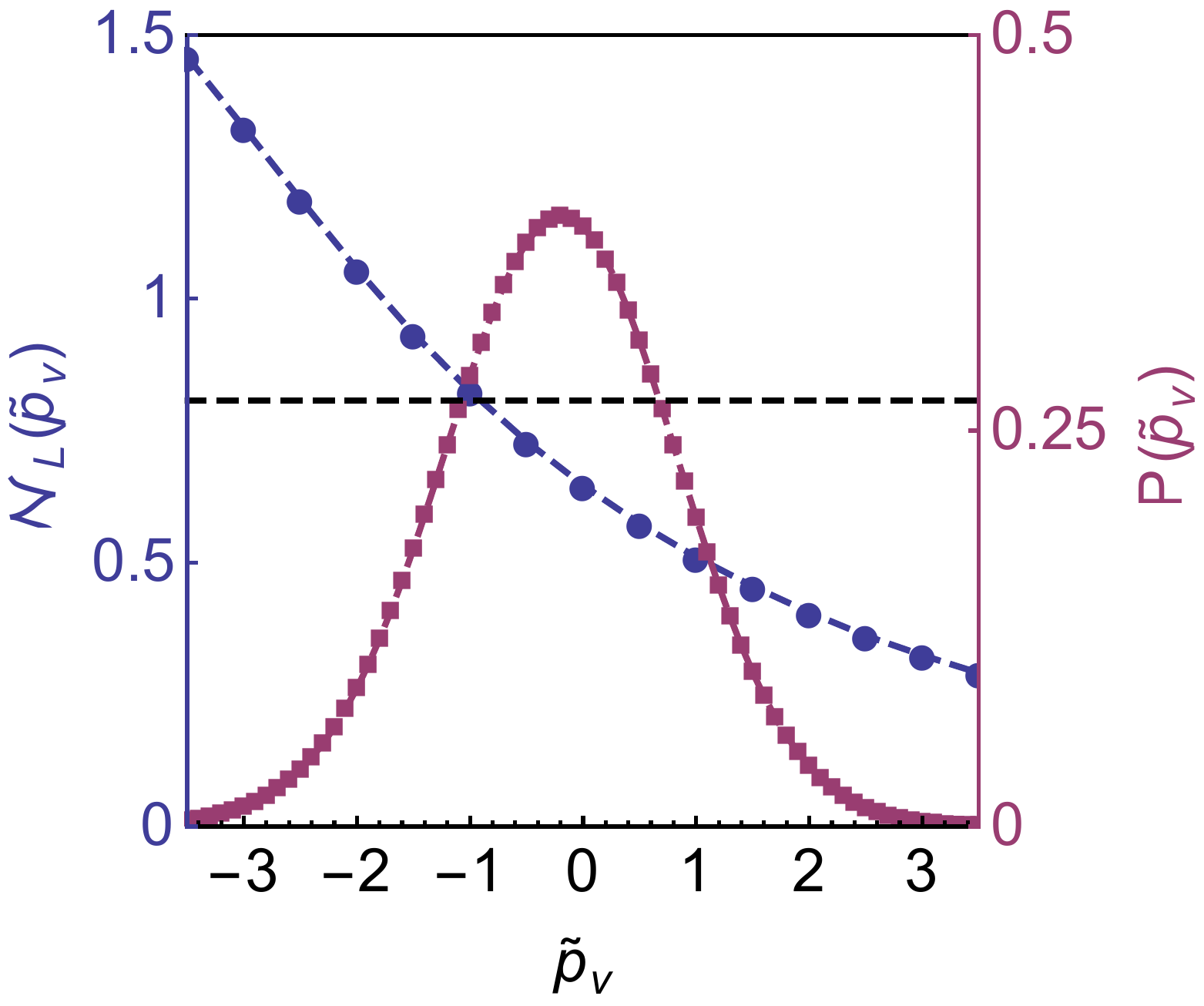}
    \label{fig:NL_prob_pvmeas}
    }
    \caption{
    (a) Schematic of the distillation protocol. The yellow box is a beam splitter with transmittance $t$. ``Homo'' means homodyne measurement. The dashed line denotes the postsection by the homodyne measurement result $\tilde{p}_v$. ``Vac'' denotes the vacuum ancilla. 
    (b) Mechanism of the distillation  protocol. The cross section of the Wigner function at $q=0$ is shown as functions of $p$. Dotted blue curve represents the input vacuum state with standard deviation $\sigma=1$ and purple curve represents the input resource state. In this figure, the Wigner function of an imperfect cubic phase state is plotted. The partial homodyne shifts the Gaussian distribution by $-\sqrt{t/(1-t)}\,\tilde{p}_v$ where $t$ is the transmittance of the beam splitter and $\tilde{p}_v$ is the outcome of the homodyne measurement. The case of a negative $\tilde{p}_v$ is shown. The Wigner function of the resource state is also shifted in the opposite direction by $-\sqrt{(1-t)/t}\,\tilde{p}_v$, but it is not shown in the figure because the shift is small when $1-t\ll 1$. The Wigner function of the output state is the product of the shifted Gaussian distribution and the slightly shifted distribution of the resource state. The shifted Gaussian works as a filter function that focuses on the region which has a large contribution to the negativity.
    (c) Logarithmic negativity $\calN_L$ of imperfect cubic phase states $\ket{\gamma,P,s}$ with $\gamma=0.05$, $P=0$ in terms of squeezing parameter $s$.
    (d) Logarithmic negativity $\calN_L(\tilde{p}_v)$ (blue dots) and the probability density $P(\tilde{p}_v)$ (purple squares) for the case when the input state is $\ket{\gamma,P,s_{ini}}$ with $\gamma=0.05$, $P=0, s_{ini}=1$ in terms of outcome of the homodyne measurement $\tilde{p}_v$. The horizontal dotted line is the logarithmic negativity of the input state $\calN_L^{ini}$.
    } 
\end{figure*}

The setup of the protocol is described in Fig.~\fig{distillation}. Suppose we have a beam splitter with transmittance $t\simeq 1$ and homodyne detector that measures the momentum quadrature $\hat{p}$. The input state $\hat{\rho}_{in}$ and vacuum state $\ket{0}$ are mixed by the beam splitter. The transformation at the beam splitter is
$
 \hat{q}' = \sqrt{t} \hat{q} + \sqrt{1-t} \hat{q}_v,\ \  \hat{p}' = \sqrt{t} \hat{p} + \sqrt{1-t} \hat{p}_v, \ \
 \hat{q}_v' = \sqrt{t} \hat{q}_v - \sqrt{1-t} \hat{q}, \ \ \hat{p}_v' = \sqrt{t} \hat{p}_v - \sqrt{1-t} \hat{p},
$
where $\hat{q}$, $\hat{p}$, $\hat{q_v}$, $\hat{p_v}$ are the quadrature operators for input state and the vacuum and prime denotes the operator after the transform.
The momentum of the outgoing vacuum state is measured at the homodyne detector and postselection is made based on the measurement outcome. Specifically, the state is discarded if the measured momentum is not within the prespecified region, which can be tuned by the experimenter. From~\ref{Wigner_product}, the initial Wigner function of the input state and vacuum state is
\ba
 W_{\rm ini}(q,p,q_v,p_v) = W(q,p;\hat{\rho}_{\rm in})W(q_v,p_v;\ket{0}).
\ea
Using \lemm{Wigner_Gaussian_unitary}, the Wigner function of the total output state right after the beam splitter transformation is
\ba
W_{\rm f}(q,p,q_v,p_v) &= &W(q'_{inv}, p'_{inv};\hat{\rho}_{\rm in})
\nonumber
\\
&&\times W(q'_{v,inv}, p'_{v,inv}; \ket{0}),
\ea
where 
$
 q_{inv}' = \sqrt{t} q - \sqrt{1-t} q_{v}, \ \ p'_{inv} = \sqrt{t} p - \sqrt{1-t} p_v,
 \ \ 
 q_{v,inv}' = \sqrt{t} q_v + \sqrt{1-t} q, \ \ p_{v,inv}' = \sqrt{t} p_v + \sqrt{1-t} p.
$
From~\ref{Wigner_post}, the output state $\hat{\rho}_{\rm out}$ conditioned that the measurement result $\tilde{p}_v$ has the Wigner function 
\ba
 W(q,p;\hat{\rho}_{{\rm out}|\tilde{p}_v} ) =  \frac{\int dq_v W_{\rm f}(q,p,q_v,\tilde{p}_v)}{\int dqdpdq_v W_{\rm f}(q,p,q_v,\tilde{p}_v)}.
 \label{eq:Wout}
\ea

Note that since the exponent of $W_{\rm f}$ is quadratic with respect to $q_v$, the integration of $q_v$ can be carried out analytically. 
Let $\tilde{p}_v^{-}$ and $\tilde{p}_v^{+}$ be the thresholds for the postselection. We call the protocol successful and keep the output state when $\tilde{p}_v \in [\tilde{p}_v^{-}, \tilde{p}_v^{+}]$. The probability density of obtaining the measurement outcome $\tilde{p}_v$ is 
\ba
 P(\tilde{p}_v)=\int dqdpdq_v W_{\rm f }(q,p,q_v,\tilde{p}_v)
\ea
and the success probability $P_{suc}$ is obtained by
\ba
 P_{suc} = \int_{\tilde{p}_v^{-}}^{\tilde{p}_v^{+}} d\tilde{p}_v  P(\tilde{p}_v).
 \label{eq:psuc}
\ea

Here, we give an intuition why our protocol allows to increase the negativity. At the beginning, the input state and vacuum state are uncorrelated and the total Wigner function is just a product of these two. The beam splitter with transmittance $t\simeq 1$ generates a slight correlation between them, although it does not alter the input state significantly because we assume $1-t \ll 1$. However, it still allows the homodyne measurement to partially extract the information about the input state and give a kickback by the measurement. The Wigner function after the beam splitter depends on the measurement outcome $\tilde{p}_v$ and it is a product of the slightly shifted Wigner function of the input resource state and a shifted Gaussian distribution from the vacuum state. How much these two are shifted depends on the measured momentum $\tilde{p}_v$, and the shifted Gaussian distribution works as a filter function that passes the region which has a large contribution to the negativity.
Fig.~\fig{homodyne_mechanism} shows the mechanism of the protocol with the input state being an imperfect cubic phase state. 
Although we have only checked the genuine non-Gaussianity increase for imperfect cubic phase states having the form of Eq.~\eq{wfunc_cubic}, we expect that our protocol will work for large class of resource states in the same mechanism. One should be able to tune the thresholds for postselection so the shifted Gaussian distribution appropriately focuses on the region on the Wigner plane where many ripples occur.

Furthermore, specifically for cubic phase state, it can be expected that the output state gets closer to another imperfect cubic phase state with larger squeezing parameter as long as the input squeezing level is not too large. 
One can check that the intervals of ripples of the Wigner functions for imperfect cubic phase states do not vary much with squeezing parameters, but states with larger squeezing parameter have less difference in the amplitudes between neighboring peaks in the small $p$ domain.
As can be seen in Fig.~\fig{homodyne_mechanism}, the broad Gaussian distribution evens out the difference in the amplitudes in the small $p$ region if the amplitude difference and filtering of the Gaussian distribution balances out. In such cases, the filter function makes the shape of the Wigner function closer to the one for a larger squeezing parameter. However, when the initial squeezing is so large that ripples are still significantly large in a large $p$ domain, the amplification in the focused domain starts beating the ripple in the small $p$ domain, so it will not get closer to the state with larger squeezing parameter although it will still increase the negativity. 
We shall observe this tendency of the fidelity increase later in this section.

Let us first look at the negativity increase.
We are interested to compare the initial logarithmic negativity to the final logarithmic negativity. Let $\calN_L^{ini}$ be the initial logarithmic negativity. Let $\calN_L(\tilde{p}_v)$ denote the logarithmic negativity of state described by Wigner function in Eq.~\eq{Wout}. The average negativity conditioned on the postselection is given by
\ba
 \calN_L^{post} = \left< \calN_L^{fin} \right> /P_{suc},
 \label{eq:Wout_post}
\ea 
 where
\ba 
 \left< \calN_L^{fin} \right> = \int_{\tilde{p}_v^{-}}^{\tilde{p}_v^{+}} d\tilde{p}_v P(\tilde{p}_v) \calN_L(\tilde{p}_v).
\ea
\begin{figure*}
    \centering
    \subfigure[]{
     \includegraphics[width=0.25\textwidth]{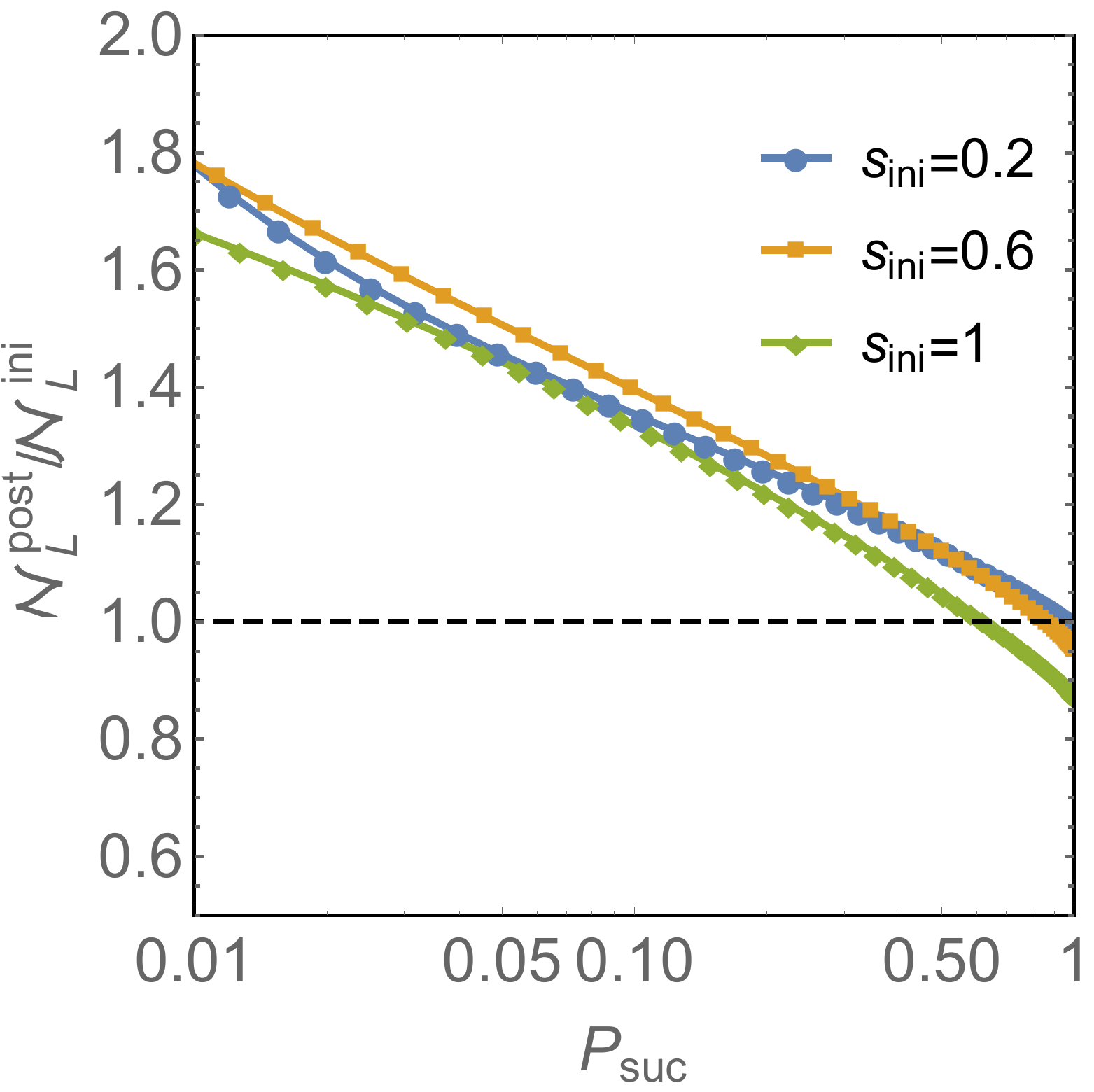}
     \label{fig:NL_post_psuc}
     }
     \centering
    \subfigure[]{
     \includegraphics[width=0.26\textwidth]{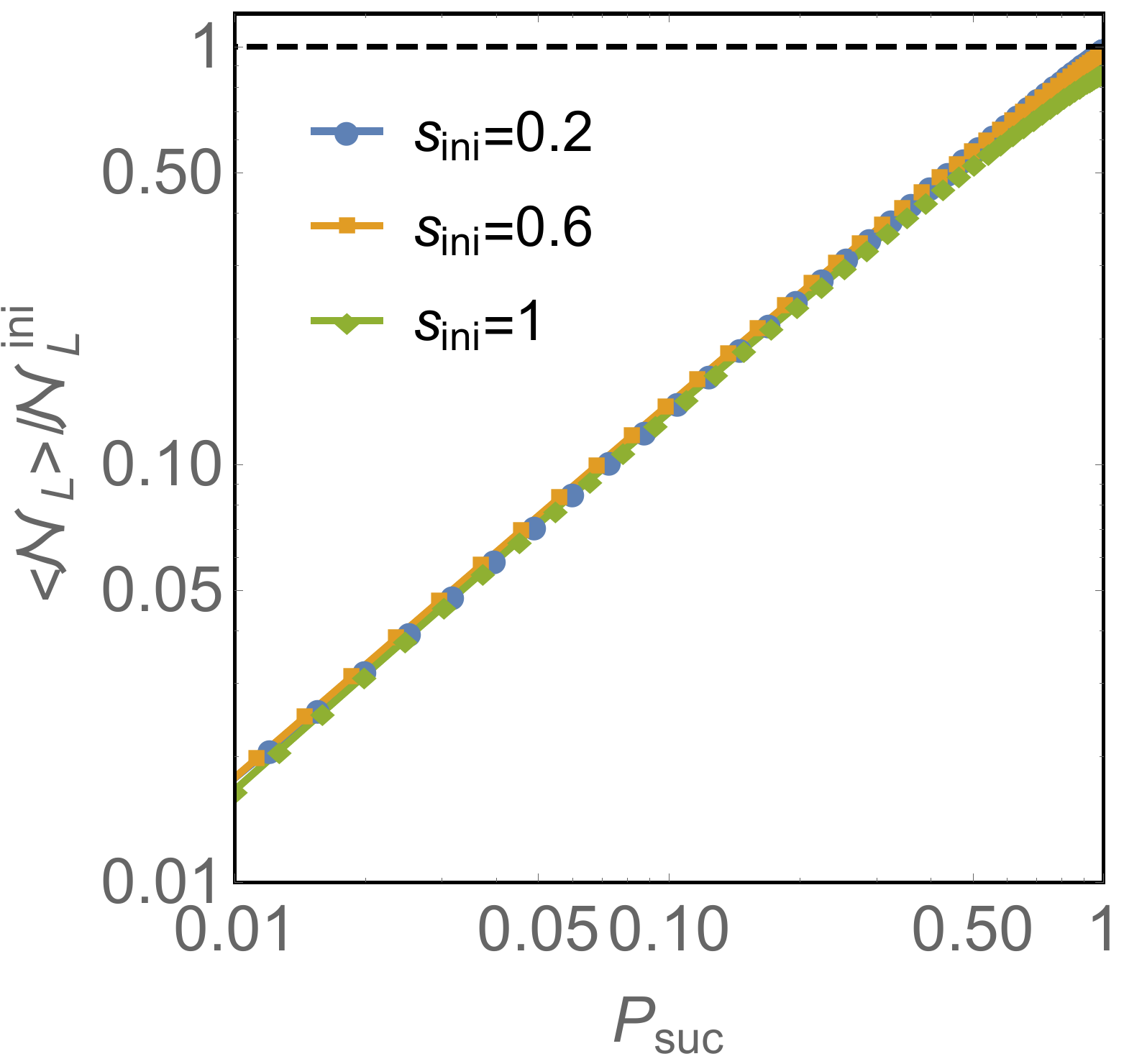}
     \label{fig:NL_ave_psuc}
     } 
    \centering
    \subfigure[]{
    \includegraphics[width=0.25\textwidth]{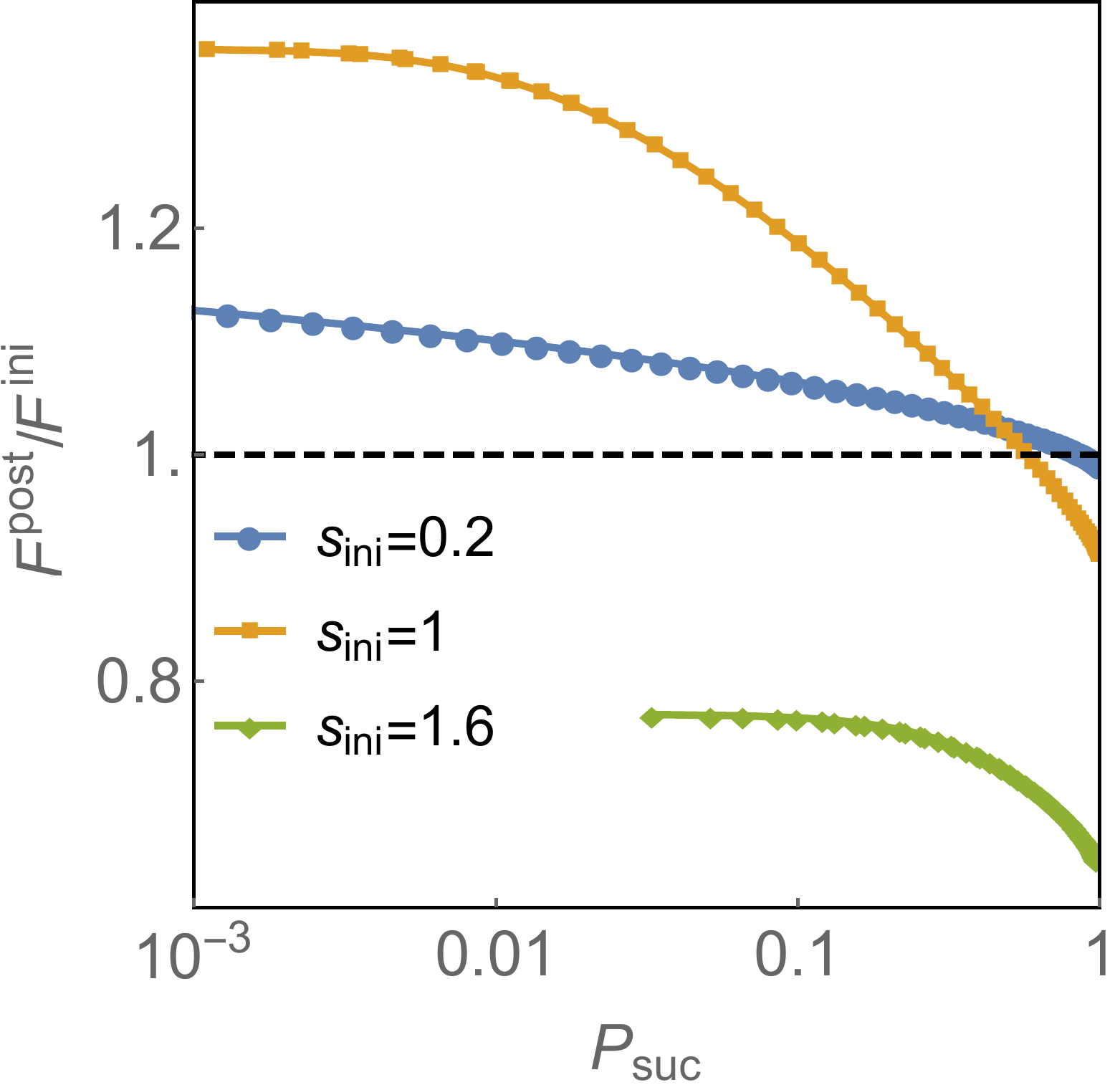}
    \label{fig:Fpost_psuc}
    }
     \caption{(a) Logarithmic negativity after postselection $\calN_L^{post}$ in terms of success probability $P_{suc}$ for different input squeezing parameters. Input states are $\ket{\gamma,P,s_{ini}}$ with $\gamma=0.05$, $P=0$. The initial logarithmic negativity for each $s_{ini}$ are $\calN_L^{ini} = 0.11, 0.38, 0.81$ for $s_{ini} = 0.2, 0.6, 1.0$ respectively. 
     (b) Average logarithmic negativity after the protocol $\left< \calN_L^{fin} \right>$ normalized by the initial logarithmic negativity.
     (c) Fidelity from the target state $\ket{\gamma, P',s_{targ}}$ after postselection $F^{post}$ in terms of success probability $P_{suc}$ for different input squeezing parameters. The target squeezing parameter is set as $s_{targ}=4.0$. Input states are $\ket{\gamma,P,s_{ini}}$ with $\gamma=0.05$, $P=0$. The initial fidelity for each $s_{ini}$ are $F^{ini} = 0.04, 0.10, 0.18$ for $s_{ini} = 0.2, 1.0, 1.6$ respectively.
     } 
\end{figure*}
The resource theory developed above allows one to put an upper bound on the average logarithmic negativity without a postselection. Namely,
\ba
 \left< \calN_L^{fin} \right> \leq \calN_L^{ini}
 \label{eq:bound}
\ea
and it holds for any choice of $\tilde{p}_v^{-}$ and $\tilde{p}_v^{+}$. One can see it as a trade-off relation between the output negativity and success probability $\calN_L^{post} \leq \calN_L^{ini}/P_{suc}$ by applying Eq.~\eq{bound} to Eq.~\eq{Wout_post}. Note that the left hand side of Eq.~\eq{bound} is monotonically increasing as the success region gets larger. It increases the success probability but may reduce the output negativity after postselection.

We choose an imperfect cubic phase state $\ket{\gamma, P, s_{ini}}$ as the input state $\hat{\rho}_{in}$, then $\calN_L^{ini}=\calN_L (\ket{\gamma, P, s})+\calN_L (\ket{0})=\calN_L (\ket{\gamma, P, s})$. Since $\calN_L$ is independent of $P$, we choose $P=0$ without loss of generality. As seen in the previous section, $\ket{\gamma, P, s}$ approaches the ideal cubic phase state as $s$ increases, so it is expected that the genuine non-Gaussianity also increases as $s$ increases. We find that it is indeed the case as shown in Fig.~\fig{NL_s}. We choose $\tilde{p}_v^{\pm}$ for given $P_{suc}$ such that it gives the maximum $\calN_L^{post}$ while satisfying Eq.~\eq{psuc}. $\calN_L(\tilde{p}_v)$ and $P(\tilde{p}_v)$ are plotted in terms of $\tilde{p}_v$ in Fig.~\fig{NL_prob_pvmeas}. The Gaussian shape of the probability distribution is a reminiscence of Gaussian Wigner function of the vacuum state. Since $\calN_L(\tilde{p}_v)$ is monotonically decreasing with $\tilde{p}_v$, we set $\tilde{p}_v^{-}=-\infty$ and take $\tilde{p}_v^{+}$ that satisfies Eq.~\eq{psuc}. This monotonic behavior of $\calN_L(\tilde{p}_v)$ in terms of $\tilde{p}_v$ was also seen in the other choice of $s_{ini}$ for $0\leq s_{ini}\leq 1.4$. 
Fig.~\fig{NL_post_psuc} shows the change in the logarithmic negativity by our protocol in terms of the success probability.
There is a clear trade-off between the negativity increase and the success probability, but positive increase in the negativity is realized at the most of the success probabilities.
Fig.~\fig{NL_ave_psuc} shows the ratios of the average logarithmic negativity after the protocol to the initial logarithmic negativity in terms of the success probability. It confirms that the selective monotonicity of the logarithmic negativity Eq.~\eq{bound} is satisfied. Note that the description of the output state is also known because the Wigner function of the output state contains all the information about the state. Thus, larger genuine non-Gaussianity of the output state will be helpful to realize tasks such as universal quantum computation under noisy environment.

As expected from Fig.~\fig{homodyne_mechanism} and the argument above, it turns out that the same protocol also allows for the ``cubic phase state distillation'' in the sense that it increases the fidelity from another imperfect (but better) cubic phase state with higher squeezing parameter. Again, we take an imperfect cubic phase state $\ket{\gamma,P_{ini},s_{ini}}$ as an input state and look at the fidelity from another imperfect cubic phase state $\ket{\gamma,P_{ini},s_{targ}}$ with $s_{targ} > s_{ini}$,
\ba
 F^{ini} &\equiv& F(\ket{\gamma,P_{ini}, s_{ini}}, \ket{\gamma,P_{ini}, s_{targ}})
 \nonumber\\
 &=& |\braket{\gamma, P_{ini}, s_{ini}|\gamma, P_{ini}, s_{targ}}|^2 
 \nonumber\\
 &=& \left[\cosh (s_{ini} - s_{targ})\right]^{-1},
\ea
where we used Eq.~\eq{wfunc_cubic}. Suppose we obtain $\tilde{p}_v$ as the outcome of the homodyne measurement. Then, the fidelity between the output state and the target state is 
\begin{eqnarray}
 F(\tilde{p}_v) &\equiv& F(\hat{\rho}_{{\rm out}|\tilde{p}_v}, \ket{\gamma,P',s_{targ}}) 
 \nonumber\\
 &=& |\bra{\gamma, P', s_{targ}} \hat{\rho}_{{\rm out}|\tilde{p}_v} \ket{\gamma, P', s_{targ}}|
 \nonumber\\
 &=& 4\pi\int dqdp W(q,p;\hat{\rho}_{{\rm out}|\tilde{p}_v})\nonumber \\ 
 &&\hspace{1.6cm}\times W(q,p;\ket{\gamma,P',s_{targ}}),
\end{eqnarray}
where we set $P'=\sqrt{\frac{1-t}{t}}\,\tilde{p}_v$ to take into account the slight shift of the Wigner function of the resource state along $p$ direction due to the measurement. The form of $W(q,p;\ket{\gamma,P,s})$ is given in Appendix \ref{app_wigners}. We choose $\tilde{p}_v^{\pm}$ in the same fashion so given $P_{suc}$ it maximizes $F^{post}$ where 
\ba
 F^{post} \equiv \left<F^{fin}\right>/P_{suc}
\ea
 and
 \ba
 \left<F^{fin}\right> = \int_{\tilde{p}_v^{-}}^{\tilde{p}_v^{+}}d\tilde{p}_v P(\tilde{p}_v)F(\tilde{p}_v).
\ea
Fig.~\fig{Fpost_psuc} shows the change in the fidelity before and after the protocol for different initial squeezing parameters. We show the result for $s_{targ} = 4$, but similar behaviors were seen for the other choices of $s_{targ}$. As can be seen in the figure, we obtain 13\% increase in the fidelity with success probability 1\% for $s_{ini}=1$.
As expected, the increase in the fidelity is not monotonic with $s_{ini}$; $F^{post}/F^{ini}$ increases until about $s_{ini} = 1.2$, but starts decreasing as $s_{ini}$ increases from that point. For $s_{ini}=1.6$ the fidelity even decreases from the initial fidelity. 
One may wonder whether sequential applications of the protocol keeps increasing the fidelity. However, it turns out that after some numbers of applications, the increase saturates at some point that is not very far from the fidelity obtained after the first application. It is not very surprising that the saturation occurs because although it increases the negativity, mismatch of the shape of the Wigner function can accumulate over the number of applications of the protocol. We leave the construction of a sequential protocol that keeps increasing the fidelity for future work.

\section{Conclusions}
\label{sec_conclusions}
We developed a resource theory of genuine non-Gaussianity that is relevant to continuous variable quantum computation. 
Our theory is operation-driven where we first fix a set of operations that are easily accessible as free operations and a set of free states are defined as the maximal set of the states generated by them.
We formally introduced the Gaussian protocol as free operations and the convex hull of the Gaussian states was naturally chosen to be a corresponding set of free states. 
We showed that the logarithmic negativity of Wigner function, the logarithm of the integral of the absolute value of the Wigner function, is a valid monotone under the Gaussian protocols and examined its properties. 
We computed the logarithmic negativity of well known genuine non-Gaussian resource states and compared them with respect to the mean photon number. We found that number states, ON states, and cubic phase states show similar behaviors. 
As an application of our theory, we discussed state conversions under the Gaussian protocols, where monotonicity serves as a necessary condition for such conversions. 
We examined the recently proposed protocol implementing the cubic phase gate using ON state and found that their protocol is efficient in terms of genuine non-Gaussianity. 
For another application, we proposed a simple protocol that may non-deterministically increase the genuine non-Gaussianity using the Gaussian protocol with postselection. We numerically verified that it indeed increases the logarithmic negativity of imperfect cubic phase states.
We further showed that it may also work as a ``cubic phase state distillator'', which takes an imperfect cubic phase state as input and outputs a state with higher fidelity to the better cubic phase state with larger squeezing parameter.
It needs to be noted that it is different from a conventional state distillation protocol that keeps purifying the state to a specific state under sequential applications of the protocols. 
Under our protocol, fidelity saturates after some applications due to the accumulation of the mismatch of the Wigner functions.

From this work, there is much room to explore in the future.
When counting resources, the notion of ``golden unit'' of the resource is often helpful. 
In the quantum computation in discrete systems, one could consider the $T$ state as the golden unit and evaluate a state by, for instance, the number of $T$ states required to create it.
It is subtle to define the corresponding notion in continuous-variable systems (with infinite dimension) because the genuine non-Gaussianity diverges at the infinite photon number limit. 
One may think that it can be still defined for fixed mean photon numbers, but it is problematic because Gaussian protocols can freely change the mean photon number. A formal treatment awaits.
Another concern would be that a protocol that keeps purifying a state to a specific resource state under sequential applications has not been known. 
Although our protocol increases the fidelity of imperfect cubic phase states, it ceases to increase the fidelity after some numbers of applications. 
The notion of the golden unit will make more sense if we have such a protocol.

The design of distillation protocols is related to the search of sufficient resources for continuous-variable universal quantum computation. 
The perfect cubic phase state, which implements the perfect cubic phase gate, is not a valid quantum state.
What is physically relevant is whether Gaussian protocols, with an infinite supply of an imperfect cubic phase state, can perform universal quantum computation.
In the case of the qubit computation, stabilizer protocols with an infinite supply of a noisy magic state are sufficient for the universal quantum computation. This is because well-established magic-state distillation protocols allow one to prepare a state that is arbitrarily close to the target magic state.
We could exploit the same argument for the continuous-variable universal quantum computation if we had a corresponding protocol. 
The protocol proposed in this paper may serve as a first step toward answering this question.

Another possible future work would be to consider genuine-non-Gaussianity generating power of processes.
The resource generating power of a quantum process indicates the maximum amount of the resource that it can create per use.  
It will play an important role in evaluating the cost of classical simulation of the circuit that involves operations that are not free~\cite{pashayan2015estimating}. 

Although we only discussed one monotone, the logarithmic negativity, it will be interesting to examine other monotones as well.
Especially, the logarithmic negativity is not faithful, and there are many potential choices for faithful measures.
For instance, distance based measures which quantify how far the state is from the closest free state are natural choices of faithful monotones. 
Since the set of free states is convex, well-known robustness measure~\cite{Vidal1999, howard2017application} and more general measure~\cite{Regula2018} should also work.
On introducing such faithful measures, the following two points should be cared. One is that such faithful measures will be hard to compute~\cite{veitch2014resource}.---even telling whether a state belongs to the set of free states is difficult~\cite{deMelo2013}.
The other point is that one needs to clarify the tasks on which bound genuine non-Gaussian states---genuine non-Gaussian states with non-negative Wigner function---have advantage over the states that can be written by convex mixtures of Gaussian states. 
The ways to detect bound genuine non-Gaussian states have been investigated, motivated by a desire to tell whether some genuinely non-Gaussian protocols have been applied before certain noisy channels may have caused the Wigner function to be positive again~\cite{filip2011detecting,genoni2013detecting,park2015testing,hughes2014quantum,jevzek2011experimental,Happ2018}. However, the operational meaning of bound genuine non-Gaussian states is not known yet. 
As an analog, in entanglement resource theory, bound entanglement can still improve the task of channel discrimination~\cite{Piani2009}. It will be important to identify such tasks in the Gaussian context.  

Finally, if we take the perspective that the negativity is a useful resource, one could develop a state-driven theory starting from the set of free states being the states with non-negative Wigner function. 
It would be nice to clarify the characterization of the ``non-negative Wigner function preserving operations'' and examine how feasible it is to realize such operations.

{\it Note added.}---Recently we became aware of a related work by F. Albarelli, M. G. Genoni, M. G. A. Paris and A. Ferraro~\cite{albarelli2018resource}.

\begin{acknowledgments}
The authors thank Daniel Gottesman, Jeffrey H. Shapiro, Yasunari Suzuki, and Seth Lloyd for fruitful discussions. R.T. acknowledges the support of the Takenaka scholarship foundation.
Q.Z. is supported by the Air Force Office of Scientific Research Grant No. FA9550-14-1-0052. Q.Z. also acknowledges the Claude E. Shannon Research Assistantship.
\end{acknowledgments}

\begin{appendix}
\begin{widetext}

\section{Proof of \lemm{Wigner_Gaussian_unitary}}
\label{proof_lemma1}
\ba
\chi \left({  \xi};\hat{U}_{  S,  d} \hat{\rho} \hat{U}_{  S,  d}^\dagger\right)&=& {\rm Tr} \left[\hat{U}_{  S,  d} \hat{\rho} \hat{U}_{  S,  d}^\dagger\exp\left(i \hat{  x}^T {  \Omega} {  \xi}  \right)\right]
\nonumber
\\
&=&{\rm Tr} \left[ \hat{\rho} \hat{U}_{  S,  d}^\dagger\exp\left(i \hat{  x}^T {  \Omega} {  \xi}  \right)\hat{U}_{  S,  d}\right]
\nonumber
\\
&=&{\rm Tr} \left[ \hat{\rho} \exp\left(i \left(  S  \hat{x}+  d\right)^T {  \Omega} {  \xi}  \right)\right]
\nonumber
\\
&=&\chi \left({  S}^{-1}{  \xi};\hat{\rho}\right) \exp\left(i {  d}^T {  \Omega   \xi }\right).
\ea 
In the last step we used $  S^T   \Omega=  \Omega   S^{-1}$. To be more rigorous, the cyclic property of trace is ensured by the Fubini-Tonelli theorem.

Similarly, for Wigner function, we have
\ba
&&W\left({  x};\hat{U}_{  S,  d}\hat{\rho} \hat{U}_{  S,  d}^\dagger\right)
\nonumber
\\
&=&\int \frac{d ^{2N}{  \xi}^\prime}{\left(2\pi\right)^{2N}}\exp\left(-i {  x}^T {  \Omega} {  S} {  \xi}^\prime\right) \chi \left({  \xi^\prime};\hat{\rho}\right) \exp\left(i {  d}^T {  \Omega   S  \xi^\prime }\right),
\nonumber
\\
&=&\int \frac{d ^{2N}{  \xi}^\prime}{\left(2\pi\right)^{2N}}\exp\left(-i \left({  S}^{-1}\left(  x-  d\right)\right)^T {  \Omega}  {  \xi}^\prime\right) \chi \left({  \xi^\prime};\hat{\rho}\right) 
\nonumber
\\
&=&W\left(  S^{-1}\left(  x-  d\right);\hat{\rho}\right).
\ea
where we let $  \xi^\prime=  S^{-1}  \xi$ and used $|\det\left(  S\right)|=1$, and $  \Omega   S=\left(  S^{-1}\right)^T   \Omega$.

\section{Proof of \ref{Wigner_xspace}, \ref{Wigner_trace_states} in \lemm{Wigner}}
\label{Convention_proof}

Note that in accordance with ref.~\cite{Weedbrook_2012}, we use the following convention: commutation relation $\left[\hat{q},\hat{p}\right]=2i$, delta-function $\int dp \exp\left(ip x\right)=2\pi \delta\left(x\right)$, and displacement of position $\hat{D}_q\left(x\right)\equiv \exp\left(-i\hat{p}x/2\right)$, which satisfies $\hat{D}_q\left(x\right)\ket{y}=\ket{y+x}$. It suffices to prove the single mode case. Denote $\xi=\left(\xi_1,\xi_2\right)$, thus $\hat{D}\left({  \xi}\right)=e^{-i\hat{p}\xi_1+i\hat{q}\xi_2}=e^{-i\hat{p}\xi_1} e^{i\hat{q}\xi_2} e^{i \xi_1 \xi_2}$.

From the definition, for single mode pure states $\hat{\rho}=\state{\psi}$, we have (Note all integration are from $-\infty$ to $\infty$)
\ba
W\left(q,p;\hat{\rho}\right)&=&\int \frac{d ^{2}\xi_1 \xi_2}{\left(2\pi\right)^{2}}
e^{ip\xi_1} e^{-iq\xi_2} \int dx \braket{x|\hat{\rho} e^{-i\hat{p}\xi_1} e^{i\hat{q}\xi_2} e^{i \xi_1 \xi_2}|x}
\\
&=&\int \frac{d ^{2}\xi_1 \xi_2}{\left(2\pi\right)^{2}}
e^{ip\xi_1} e^{-iq\xi_2} \int dx \braket{x|\psi}\braket{\psi|x+2\xi_1} e^{ix\xi_2} e^{i \xi_1 \xi_2}
\\
&=&\int \frac{d\xi_1 }{2\pi}
e^{ip\xi_1}  \int dx \braket{x|\psi}\braket{\psi|x+2\xi_1} \delta\left(x+\xi_1-q\right)
\\
&=&\int \frac{d\xi_1 }{2\pi}
e^{-ip\xi_1}   \braket{q+\xi_1|\psi}\braket{\psi|q-\xi_1}.
\ea
We arrive at \ref{Wigner_xspace}.

For two single-mode quantum states,
\ba
&&4\pi \int d x W\left(x;\hat{\rho}\right)W\left(x;\hat{\sigma}\right)
\\
&=&4\pi \int d q dp 
\int \frac{d ^{2}\xi_1 \xi_2}{\left(2\pi\right)^{2}}
e^{ip\xi_1} e^{-iq\xi_2} \int dx \braket{x|\hat{\rho}|x+2\xi_1} e^{ix\xi_2} e^{i \xi_1 \xi_2}
\\
&&\times 
\int \frac{d ^{2}\xi_1^\prime \xi_2^\prime}{\left(2\pi\right)^{2}}
e^{ip\xi_1^\prime} e^{-iq\xi_2^\prime} \int dx^\prime \braket{x^\prime|\hat{\sigma}|x^\prime+2\xi_1^\prime} e^{ix^\prime\xi_2^\prime} e^{i \xi_1^\prime \xi_2^\prime}
\\
&=&
4\pi 
\int \frac{d ^{2}\xi_1 \xi_2}{\left(2\pi\right)^{2}}\int \frac{d ^{2}\xi_1^\prime \xi_2^\prime}{\left(2\pi\right)^{2}}
2\pi\delta\left(\xi_1+\xi_1^\prime\right)2\pi\delta \left(\xi_2+\xi_2^\prime\right)
\\
&&\times \int dx \braket{x|\hat{\rho}|x+2\xi_1} e^{ix\xi_2} e^{i \xi_1 \xi_2}
 \int dx^\prime \braket{x^\prime|\hat{\sigma}|x^\prime+2\xi_1^\prime} e^{ix^\prime\xi_2^\prime} e^{i \xi_1^\prime \xi_2^\prime}
\\
&=&
4\pi 
\int \frac{d ^{2}\xi_1 \xi_2}{\left(2\pi\right)^{2}} 
\int dx \braket{x|\hat{\rho}|x+2\xi_1} e^{ix\xi_2} e^{i \xi_1 \xi_2}
\int dx^\prime \braket{x^\prime|\hat{\sigma}|x^\prime-2\xi_1} e^{-ix^\prime\xi_2} e^{i \xi_1 \xi_2}
\\
&=&
4\pi 
\int \frac{d \xi_1}{\left(2\pi\right)^{2}} 
\int dx \braket{x|\hat{\rho}|x+2\xi_1} 
\int dx^\prime \braket{x^\prime|\hat{\sigma}|x^\prime-2\xi_1} 
2\pi \delta\left(x-x^\prime +2\xi_1\right)
\\
&=&
\int d 2\xi_1
\int dx \braket{x|\hat{\rho}|x+2\xi_1}  \braket{x+2\xi_1|\hat{\sigma}|x} =\tr\left(\hat{\rho} \hat{\sigma} \right).
\ea
We arrive at \ref{Wigner_trace_states}.

\section{Proof of \lemm{mana}}
\label{proof_mana}
\begin{proof}

\ref{mana_unitary} directly comes from \lemm{Wigner_Gaussian_unitary} as follows.  
\ba
\calN_L\left(\hat{U}_{S,d}\hat{\rho}\hat{U}_{S,d}^\dagger\right)&=&\ln \int d^{2N} x |W\left(x; \hat{U}_{S,d}\hat{\rho}\hat{U}_{S,d}^\dagger\right)|
\\
&=&\ln \int d^{2N} x|W\left(S^{-1}x-S^{-1}d;\hat{\rho}\right)|
\\
&=&\calN_L\left(\hat{\rho}\right).
\ea 

\ref{mana_partial_trace} comes from \ref{Wigner_trace} and triangle inequality. Consider bipartite state $\hat{\rho}_{AB}$ with two parts $A$ and $B$.
\ba
\calN_L\left({\rm Tr}_A \hat{\rho}_{AB}\right)&=&\ln \int d^{2N_B} x_B |W\left(x_B; {\rm Tr}_A \hat{\rho}_{AB} \right)|
\\
&=& \ln \int d^{2N_B} x_B|\int d^{2N_A} x_A W\left(x_A,x_B; \hat{\rho}_{AB}\right)|
\\
&\le &\ln \int d^{2N_B} x_B\int d^{2N_A} x_A |W\left(x_A,x_B; \hat{\rho}_{AB}\right)|
\\
&=&\calN_L\left( \hat{\rho}_{AB}\right).
\ea 

\ref{mana_additive} follows directly from \ref{Wigner_product}.
\ba 
\calN_L\left(\hat{\rho}_{A}\otimes \hat{\rho}_{B}\right)&=&
\ln \int d^{2N} x |W\left(x; \hat{\rho}_{A}\otimes \hat{\rho}_{B}\right)|.
\\
&=&
\ln \int d^{2N_A} x_A \int d^{2N_B} x_B
|W\left(x_A; \hat{\rho}_{A} \right)|| W\left(x_B;  \hat{\rho}_{B} \right)|
\\
&=&\calN_L\left(\hat{\rho}_{A}\right)+\calN_L\left( \hat{\rho}_{B}\right).
\ea 

\ref{mana_zero} is simply from equality condition of triangular inequality.  If $\int d^{2N} x |W\left(x; \hat{\rho}\right)|=\int d^{2N} x W\left(x; \hat{\rho}\right)=1$, we have $\calN_L\left(\hat{\rho}\right)=0$. If $\calN_L\left(\hat{\rho}\right)=0$, we have $W\left(x; \hat{\rho}\right)\ge 0$ except for points with measure zero. These measure zero negative points have no relevance to experiments in reality, and we have included these cases in $W_+$.

\ref{mana_channel} directly follows from \ref{mana_unitary}-\ref{mana_zero} and the fact that any Gaussian channel has Stinespring dilation of a Gaussian unitary $\hat{U}_\calG$ with ancilla $E$ in vacuum state $\hat{0}$. $\calN_L\left(\Phi_\calG\left(\hat{\rho}\right)\right)=\calN_L\left(\tr_E \left(\hat{\rho}\otimes \hat{0}\right)\right)\le \calN_L\left(\hat{\rho}\otimes \hat{0}\right)= \calN_L\left(\hat{\rho}\right)$.


To prove \ref{mana_prob}, it suffices to prove non-increasing under the Gaussian protocol in~\lemm{Gprotocol2} on state $\hat{\rho}_{AB}$: perform a homodyne on $A$ and conditioned on the measurement result $q_A$, perform a Gaussian channel $\Phi_{q_A}$ on $B$. 
By \ref{Wigner_measurement}, the measurement result's distribution is given by
\be
P_{q_A}=\int d^{N_A} p_A d^{2N_B} x_B W\left(p_A,q_A,x_B; \hat{\rho}_{AB}\right),
\ee
and the Wigner function of $B$ conditioned on measurement result $q_A$ is given by 
\be 
W\left(x_B; \hat{\rho}_{B|q_A}\right)=\frac{1}{P_{q_A}}\int d^{N_A} p_A W\left(p_A,q_A,x_B; \hat{\rho}_{AB}\right).
\ee

The logarithmic negativity of the overall output of the channels $\{\Phi_{q_A}\}$ conditioned on measurement result $q_A$ is 
\ba 
&&  \calN_L\left(\int d^{N_A} q_A P_{q_A}\Phi_{q_A}\left(\hat{\rho}_{B|q_A}\right)\right)
\\
&= &  \ln \int d^{2N_B} x_B |\int d^{N_A} q_A P_{q_A} W\left(x_B; \Phi_{q_A}\left(\hat{\rho}_{B|q_A}\right)\right)|
\\
&\le & \ln \int d^{N_A} q_A P_{q_A} \int d^{2N_B} x_B | W\left(x_B; \Phi_{q_A}\left(\hat{\rho}_{B|q_A}\right)\right)|
\\
&\le & \ln \int d^{N_A} q_A P_{q_A} \int d^{2N_B} x_B | W\left(x_B; \hat{\rho}_{B|q_A}\right)|
\\
&=& \ln \int d q^{N_A}_A P_{q_A} \int d^{2N_B} x_B | \frac{1}{P_{q_A}}\int d^{N_A} p_A W\left(p_A,q_A,x_B; \hat{\rho}_{AB}\right)|
\\
&\le &   \ln \int d^{N_A} q_A\int d^{2N_B} x_B \int d^{N_A} p_A|  W\left(p_A,q_A,x_B; \hat{\rho}_{AB}\right)|
\\
&=& \calN_L \left(\hat{\rho}_{AB}\right).
\ea 
The first equality is due to the linearity of Wigner functions. The first inequality is due to triangular inequality. The second inequality is due to \ref{mana_channel} and the monotonicity of $\ln\left(x\right)$. The third inequality is due to triangular inequality.

The non-increasing of the average logarithmic negativity can be proved similarly as follows
\ba 
\overline{\calN_L}&\equiv &\int d^{N_A} q_A P_{q_A}  \calN_L\left(\Phi_{q_A}\left(\hat{\rho}_{B|q_A}\right)\right)
\\
&\le &  \int d^{N_A} q_A P_{q_A}  \calN_L\left(\hat{\rho}_{B|q_A}\right)
\\
&=&\int d^{N_A} q_A P_{q_A}  \ln \int d^{2N_B} x_B |\frac{1}{P_{q_A}}\int d^{N_A} p_A  W\left(p_A,q_A,x_B; \hat{\rho}_{AB}\right)|
\\
&\le &   \ln \int d^{N_A} q_A\int d^{2N_B} x_B |\int d^{N_A} p_A  W\left(p_A,q_A,x_B; \hat{\rho}_{AB}\right)|
\\
&\le &   \ln \int d^{N_A} q_A\int d^{2N_B} x_B \int d^{N_A} p_A|  W\left(p_A,q_A,x_B; \hat{\rho}_{AB}\right)|
\\
&=& \calN_L \left(\hat{\rho}_{AB}\right).
\ea 
The first inequality is due to \ref{mana_channel}. The rest of inequalities is due to triangular inequality, $\ln\left(x\right)$ being concave and Jensen's inequality.

\end{proof}

\section{Wigner functions for resource states}
\label{app_wigners}
\subsection{Number state}
The Wigner function of the number state $\ket{n}$ is
\be 
W\left(p,q; \ket{n}\right)=\frac{1}{2\pi} \left(-1\right)^n L_n \left(p^2+q^2\right)e^{-\left(p^2+q^2\right)/2},
\ee 
where $L_n\left(x\right)$ is the Laguerre polynomial.

\subsection{Single-photon-added and single-photon-subtracted states}
The Wigner function of such states are analytically calculable~\cite{walschaers2017entanglement,walschaers2017statistical}. 
The Wigner function of the single-photon-added (+) or subtracted (-) state starting from $\ket{\theta,s}=R\left(\theta\right)S\left(s\right)\ket{0}$ is given by
\be
W^{\pm}\left(x\right)=\frac{1}{2}\left[x V^{-1}A_g^{\pm} V^{-1} x^T-\tr \left(V^{-1} A_g^{\pm}\right)+2\right] W_0\left(x\right),
\ee
where $V$ is the covariance matrix of $\ket{\theta,s}$ and
\be
A_g^{\pm}=2\frac{\left(V\pm I\right)^2}{\tr\left(V\pm I\right)}.
\ee 

\subsection{Cubic phase state} 
The unnormalized wave function of $\ket{\gamma,P}$ is
\ba
 \psi(q)\propto \exp\left(i\gamma q^3 + i\frac{Pq}{2}\right).
\ea
Using $W(q,p) = \frac{1}{2\pi}\int dy \psi^*(q-y) \psi(q+y)e^{-ipy}$, the unnormalized Wigner function of $\ket{\gamma, P}$ is given by
\ba
W\left(q,p; \ket{\gamma,P}\right)&\propto& \int_{-\infty}^{\infty} dy \exp\left[i\left(2\gamma y^3 + 2\left(3\gamma q^2 - \frac{p-P}{2}\right)y\right)\right]\\
&=& \int_{0}^{\infty} dy\  2\cos\left[2\gamma y^3+2\left(3\gamma q^2-\frac{p-P}{2}\right)y\right] \\
&\propto&
{\rm Ai}\left(\left(\frac{4}{3\gamma}\right)^{1/3}\left(3\gamma q^2-\frac{p-P}{2}\right)\right),
\ea 
where ${\rm Ai}\left(x\right)$ is the Airy function. 

The Wigner function of $\ket{\gamma,P,s}$ is obtained by
\ba
 W(q,p;\ket{\gamma,P,s}) &=& (8\pi^3 e^{2s})^{-1/2}\exp\left[\frac{-q^2}{2e^{2s}}\right]\int_{-\infty}^{\infty} dy \exp\left[i\left(2\gamma y^3 + 2\left(3\gamma q^2 - \frac{p-P}{2}\right)y\right)\right]\exp\left[-\frac{y^2}{2e^{2s}}\right] \nonumber\\
 &=& (8\pi^3 e^{2s})^{-1/2}\exp\left[\frac{-q^2}{2e^{2s}}\right]\int_{0}^{\infty} dy\  2\cos\left[2\gamma y^3+2\left(3\gamma q^2-\frac{p-P}{2}\right)y\right]\exp\left[-\frac{y^2}{2e^{2s}}\right]. \nonumber
 \label{eq:wig_cubic_finite}
\ea
As an example, the Wigner function with $\gamma = 0.05$, $P=0$, $s=1$ for $q>0$ is shown in Fig.~\fig{wigner_pure} (it is symmetric for $q\rightarrow -q$).

\begin{figure}
    \centering
    \includegraphics[width=0.3\textwidth]{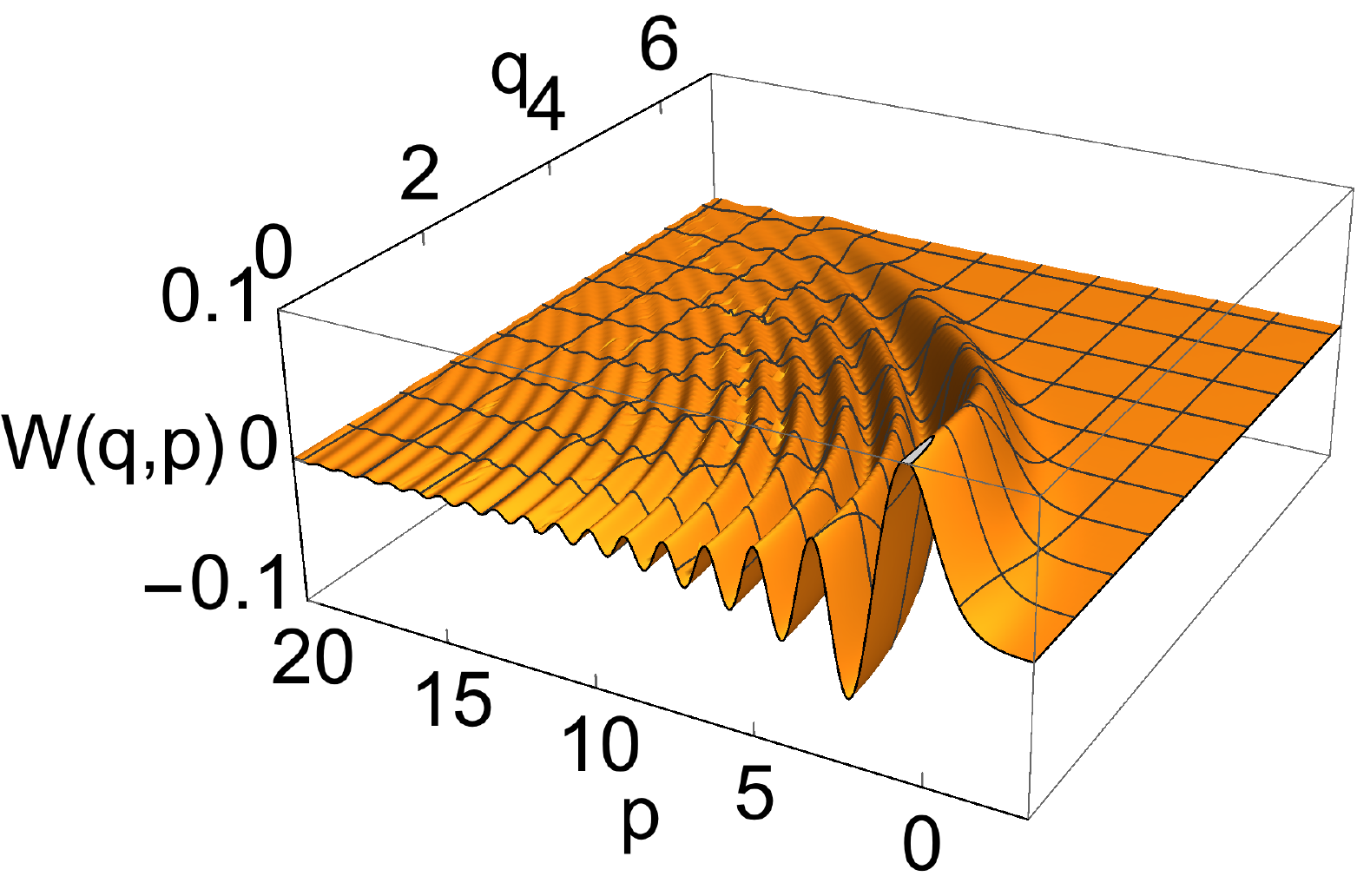}
    \caption{Wigner function of $\ket{\gamma, P, s}$ with $\gamma = 0.05$, $P=0$, $s=1$} 
    \label{fig:wigner_pure}
\end{figure}

To calculate the mean photon number of state $\ket{\gamma,P,s}$, we notice that $\hat{a}^\dagger \hat{a}=\left(\hat{p}^2+\hat{q}^2\right)/4-1/2$ and by definition Eq.~\eq{cubic_phase_definition} $\ket{\gamma,P,s}  = \hat{V}(\gamma)\hat{D}_p\left(P\right)\hat{S}\left(-s\right)\ket{0}$, thus the mean photon number is given by
\ba
N_S&=&\braket{0|\hat{S}^\dagger\left(-s\right) \hat{D}_p^\dagger\left(P\right) \hat{V}^\dagger(\gamma)\left(\hat{p}^2+\hat{q}^2\right) \hat{V}(\gamma)\hat{D}_p\left(P\right)\hat{S}\left(-s\right)|0}/4-1/2
\\
&=&\braket{0|\hat{S}^\dagger\left(-s\right)  \hat{V}^\dagger(\gamma)\left(\left(\hat{p}+P\right)^2+\hat{q}^2\right) \hat{V}(\gamma)\hat{S}\left(-s\right)|0}/4-1/2,
\ea
where we have used $\left[\hat{V}(\gamma),\hat{D}_p\left(P\right)\right]=0$ and $\hat{D}_p^\dagger\left(P\right) \hat{p} \hat{D}_p\left(P\right)=\hat{p}+P$. Since $\left[\hat{q},\hat{p}\right]=2i$, by the correspondence $\hat{p}=-2i\frac{d}{d\hat{q}}$ we have $\left[\hat{p}, \hat{V}(\gamma)\right]=\hat{V}(\gamma) 6 \gamma \hat{q}^2$, thus
\be
\left[\left(\hat{p}+P\right)^2, \hat{V}(\gamma)\right]=\left[\hat{p}^2, \hat{V}(\gamma)\right]+2P \left[\hat{p}, \hat{V}(\gamma)\right]=4\hat{V}(\gamma) \left(9\gamma^2 \hat{q}^4+3 P \gamma \hat{q}^2\right)+6 \hat{V}(\gamma) \left(\hat{p}\hat{q}^2+\hat{q}^2\hat{p}\right).
\ee
With these in hand, one finds that
\ba
N_S&=&\braket{0|\hat{S}^\dagger\left(-s\right) \left(\left(\hat{p}+P\right)^2+\hat{q}^2+4\left(9\gamma^2 \hat{q}^4+3 P \gamma \hat{q}^2 \right)+6 \left(\hat{p}\hat{q}^2+\hat{q}^2\hat{p}\right)\right) \hat{S}\left(-s\right)|0}/4-1/2.
\label{eq:Nstemp}
\ea
Now we are evaluating the expectation value of some polynomial of quadrature operators on a zero-mean squeezed vacuum state with $\braket{\hat{p}^2}=e^{-2s}, \braket{\hat{q}^2}=e^{2s}$. By Gaussian moment factoring (Wick's theorem) we have $\braket{\hat{q}^4}=3\braket{\hat{q}^2}^2=3 e^{4s}, \braket{\hat{q}^2 \hat{p}}=\braket{\hat{p}\hat{q}^2 }=0$. Combing the above into Eq.~\eq{Nstemp}, we have 
\be
N_S=\frac{1}{2}\left(\cosh(2s)-1\right)+18 \gamma^2 e^{4s}+\frac{1}{4}\left(P+6\gamma e^{2s}\right)^2,
\ee
which is Eq.~\eq{Ns_cubic_phase}.

\subsection{ON-state}
The Wigner function can be obtained analytically 
\ba
&W\left(p,q; \ket{ON}\right)=\frac{1}{1+|a|^2}W\left(p,q; \ket{0}\right)+\frac{|a|^2}{1+|a|^2}W\left(p,q; \ket{n}\right)
\nonumber
\\
&+\frac{1}{1+|a|^2}\sqrt{\frac{1}{n!}} \frac{1}{2\pi} \exp\left(-p^2-x^2\right) \left(a\left(x-i p\right)^n+a^\star \left(x+i p\right)^n \right).
\ea


\end{widetext}
\end{appendix}

\bibliography{myref.bib}

\begin{thebibliography}{98}%
\makeatletter
\providecommand \@ifxundefined [1]{%
 \@ifx{#1\undefined}
}%
\providecommand \@ifnum [1]{%
 \ifnum #1\expandafter \@firstoftwo
 \else \expandafter \@secondoftwo
 \fi
}%
\providecommand \@ifx [1]{%
 \ifx #1\expandafter \@firstoftwo
 \else \expandafter \@secondoftwo
 \fi
}%
\providecommand \natexlab [1]{#1}%
\providecommand \enquote  [1]{``#1''}%
\providecommand \bibnamefont  [1]{#1}%
\providecommand \bibfnamefont [1]{#1}%
\providecommand \citenamefont [1]{#1}%
\providecommand \href@noop [0]{\@secondoftwo}%
\providecommand \href [0]{\begingroup \@sanitize@url \@href}%
\providecommand \@href[1]{\@@startlink{#1}\@@href}%
\providecommand \@@href[1]{\endgroup#1\@@endlink}%
\providecommand \@sanitize@url [0]{\catcode `\\12\catcode `\$12\catcode
  `\&12\catcode `\#12\catcode `\^12\catcode `\_12\catcode `\%12\relax}%
\providecommand \@@startlink[1]{}%
\providecommand \@@endlink[0]{}%
\providecommand \url  [0]{\begingroup\@sanitize@url \@url }%
\providecommand \@url [1]{\endgroup\@href {#1}{\urlprefix }}%
\providecommand \urlprefix  [0]{URL }%
\providecommand \Eprint [0]{\href }%
\providecommand \doibase [0]{http://dx.doi.org/}%
\providecommand \selectlanguage [0]{\@gobble}%
\providecommand \bibinfo  [0]{\@secondoftwo}%
\providecommand \bibfield  [0]{\@secondoftwo}%
\providecommand \translation [1]{[#1]}%
\providecommand \BibitemOpen [0]{}%
\providecommand \bibitemStop [0]{}%
\providecommand \bibitemNoStop [0]{.\EOS\space}%
\providecommand \EOS [0]{\spacefactor3000\relax}%
\providecommand \BibitemShut  [1]{\csname bibitem#1\endcsname}%
\let\auto@bib@innerbib\@empty
\bibitem [{\citenamefont {Walls}\ and\ \citenamefont
  {Milburn}(2007)}]{walls2007quantum}%
  \BibitemOpen
  \bibfield  {author} {\bibinfo {author} {\bibfnamefont {D.~F.}\ \bibnamefont
  {Walls}}\ and\ \bibinfo {author} {\bibfnamefont {G.~J.}\ \bibnamefont
  {Milburn}},\ }\href@noop {} {\emph {\bibinfo {title} {Quantum Opt.}}}\
  (\bibinfo  {publisher} {Springer Science \& Business Media},\ \bibinfo {year}
  {2007})\BibitemShut {NoStop}%
\bibitem [{\citenamefont {Weedbrook}\ \emph {et~al.}(2014)\citenamefont
  {Weedbrook}, \citenamefont {Ottaviani},\ and\ \citenamefont
  {Pirandola}}]{Weedbrook_2014}%
  \BibitemOpen
  \bibfield  {author} {\bibinfo {author} {\bibfnamefont {C.}~\bibnamefont
  {Weedbrook}}, \bibinfo {author} {\bibfnamefont {C.}~\bibnamefont
  {Ottaviani}}, \ and\ \bibinfo {author} {\bibfnamefont {S.}~\bibnamefont
  {Pirandola}},\ }\href {\doibase 10.1103/PhysRevA.89.012309} {\bibfield
  {journal} {\bibinfo  {journal} {Phys. Rev. A}\ }\textbf {\bibinfo {volume}
  {89}},\ \bibinfo {pages} {012309} (\bibinfo {year} {2014})}\BibitemShut
  {NoStop}%
\bibitem [{\citenamefont {Giedke}\ and\ \citenamefont
  {Cirac}(2002)}]{giedke2002characterization}%
  \BibitemOpen
  \bibfield  {author} {\bibinfo {author} {\bibfnamefont {G.}~\bibnamefont
  {Giedke}}\ and\ \bibinfo {author} {\bibfnamefont {J.~I.}\ \bibnamefont
  {Cirac}},\ }\href@noop {} {\bibfield  {journal} {\bibinfo  {journal} {Phys.
  Rev. A}\ }\textbf {\bibinfo {volume} {66}},\ \bibinfo {pages} {032316}
  (\bibinfo {year} {2002})}\BibitemShut {NoStop}%
\bibitem [{\citenamefont {De~Palma}\ \emph {et~al.}(2015)\citenamefont
  {De~Palma}, \citenamefont {Mari}, \citenamefont {Giovannetti},\ and\
  \citenamefont {Holevo}}]{de2015normal}%
  \BibitemOpen
  \bibfield  {author} {\bibinfo {author} {\bibfnamefont {G.}~\bibnamefont
  {De~Palma}}, \bibinfo {author} {\bibfnamefont {A.}~\bibnamefont {Mari}},
  \bibinfo {author} {\bibfnamefont {V.}~\bibnamefont {Giovannetti}}, \ and\
  \bibinfo {author} {\bibfnamefont {A.~S.}\ \bibnamefont {Holevo}},\
  }\href@noop {} {\bibfield  {journal} {\bibinfo  {journal} {J. Math. Phys.}\
  }\textbf {\bibinfo {volume} {56}},\ \bibinfo {pages} {052202} (\bibinfo
  {year} {2015})}\BibitemShut {NoStop}%
\bibitem [{\citenamefont {De~Palma}(2017)}]{de2017gaussian}%
  \BibitemOpen
  \bibfield  {author} {\bibinfo {author} {\bibfnamefont {G.}~\bibnamefont
  {De~Palma}},\ }\href@noop {} {\bibfield  {journal} {\bibinfo  {journal}
  {arXiv:1710.09395}\ } (\bibinfo {year} {2017})}\BibitemShut {NoStop}%
\bibitem [{\citenamefont {Vaidman}(1994)}]{vaidman1994teleportation}%
  \BibitemOpen
  \bibfield  {author} {\bibinfo {author} {\bibfnamefont {L.}~\bibnamefont
  {Vaidman}},\ }\href@noop {} {\bibfield  {journal} {\bibinfo  {journal} {Phys.
  Rev. A}\ }\textbf {\bibinfo {volume} {49}},\ \bibinfo {pages} {1473}
  (\bibinfo {year} {1994})}\BibitemShut {NoStop}%
\bibitem [{\citenamefont {Braunstein}\ and\ \citenamefont
  {Kimble}(1998)}]{braunstein1998teleportation}%
  \BibitemOpen
  \bibfield  {author} {\bibinfo {author} {\bibfnamefont {S.~L.}\ \bibnamefont
  {Braunstein}}\ and\ \bibinfo {author} {\bibfnamefont {H.~J.}\ \bibnamefont
  {Kimble}},\ }\href {\doibase 10.1103/PhysRevLett.80.869} {\bibfield
  {journal} {\bibinfo  {journal} {Phys. Rev. Lett.}\ }\textbf {\bibinfo
  {volume} {80}},\ \bibinfo {pages} {869} (\bibinfo {year} {1998})}\BibitemShut
  {NoStop}%
\bibitem [{\citenamefont {Ralph}\ and\ \citenamefont
  {Lam}(1998)}]{ralpha1998teleportation}%
  \BibitemOpen
  \bibfield  {author} {\bibinfo {author} {\bibfnamefont {T.~C.}\ \bibnamefont
  {Ralph}}\ and\ \bibinfo {author} {\bibfnamefont {P.~K.}\ \bibnamefont
  {Lam}},\ }\href {\doibase 10.1103/PhysRevLett.81.5668} {\bibfield  {journal}
  {\bibinfo  {journal} {Phys. Rev. Lett.}\ }\textbf {\bibinfo {volume} {81}},\
  \bibinfo {pages} {5668} (\bibinfo {year} {1998})}\BibitemShut {NoStop}%
\bibitem [{\citenamefont {Cerf}\ and\ \citenamefont
  {Iblisdir}(2000)}]{cerf2000optimal}%
  \BibitemOpen
  \bibfield  {author} {\bibinfo {author} {\bibfnamefont {N.~J.}\ \bibnamefont
  {Cerf}}\ and\ \bibinfo {author} {\bibfnamefont {S.}~\bibnamefont
  {Iblisdir}},\ }\href {\doibase 10.1103/PhysRevA.62.040301} {\bibfield
  {journal} {\bibinfo  {journal} {Phys. Rev. A}\ }\textbf {\bibinfo {volume}
  {62}},\ \bibinfo {pages} {040301} (\bibinfo {year} {2000})}\BibitemShut
  {NoStop}%
\bibitem [{\citenamefont {Lindblad}(2000)}]{lindblad2000cloning}%
  \BibitemOpen
  \bibfield  {author} {\bibinfo {author} {\bibfnamefont {G.}~\bibnamefont
  {Lindblad}},\ }\href@noop {} {\bibfield  {journal} {\bibinfo  {journal} {J.
  Phys. A: Math. Gen.}\ }\textbf {\bibinfo {volume} {33}},\ \bibinfo {pages}
  {5059} (\bibinfo {year} {2000})}\BibitemShut {NoStop}%
\bibitem [{\citenamefont {Caves}(1981)}]{caves1981quantum}%
  \BibitemOpen
  \bibfield  {author} {\bibinfo {author} {\bibfnamefont {C.~M.}\ \bibnamefont
  {Caves}},\ }\href@noop {} {\bibfield  {journal} {\bibinfo  {journal} {Phys.
  Rev. D}\ }\textbf {\bibinfo {volume} {23}},\ \bibinfo {pages} {1693}
  (\bibinfo {year} {1981})}\BibitemShut {NoStop}%
\bibitem [{\citenamefont {Bondurant}\ and\ \citenamefont
  {Shapiro}(1984)}]{bondurant1984squeezed}%
  \BibitemOpen
  \bibfield  {author} {\bibinfo {author} {\bibfnamefont {R.~S.}\ \bibnamefont
  {Bondurant}}\ and\ \bibinfo {author} {\bibfnamefont {J.~H.}\ \bibnamefont
  {Shapiro}},\ }\href@noop {} {\bibfield  {journal} {\bibinfo  {journal} {Phys.
  Rev. D}\ }\textbf {\bibinfo {volume} {30}},\ \bibinfo {pages} {2548}
  (\bibinfo {year} {1984})}\BibitemShut {NoStop}%
\bibitem [{\citenamefont {Tan}\ \emph {et~al.}(2008)\citenamefont {Tan},
  \citenamefont {Erkmen}, \citenamefont {Giovannetti}, \citenamefont {Guha},
  \citenamefont {Lloyd}, \citenamefont {Maccone}, \citenamefont {Pirandola},\
  and\ \citenamefont {Shapiro}}]{tan2008quantum}%
  \BibitemOpen
  \bibfield  {author} {\bibinfo {author} {\bibfnamefont {S.-H.}\ \bibnamefont
  {Tan}}, \bibinfo {author} {\bibfnamefont {B.~I.}\ \bibnamefont {Erkmen}},
  \bibinfo {author} {\bibfnamefont {V.}~\bibnamefont {Giovannetti}}, \bibinfo
  {author} {\bibfnamefont {S.}~\bibnamefont {Guha}}, \bibinfo {author}
  {\bibfnamefont {S.}~\bibnamefont {Lloyd}}, \bibinfo {author} {\bibfnamefont
  {L.}~\bibnamefont {Maccone}}, \bibinfo {author} {\bibfnamefont
  {S.}~\bibnamefont {Pirandola}}, \ and\ \bibinfo {author} {\bibfnamefont
  {J.~H.}\ \bibnamefont {Shapiro}},\ }\href@noop {} {\bibfield  {journal}
  {\bibinfo  {journal} {Phys. Rev. Lett.}\ }\textbf {\bibinfo {volume} {101}},\
  \bibinfo {pages} {253601} (\bibinfo {year} {2008})}\BibitemShut {NoStop}%
\bibitem [{\citenamefont {Zhuang}\ \emph {et~al.}(2017)\citenamefont {Zhuang},
  \citenamefont {Zhang},\ and\ \citenamefont {Shapiro}}]{zhuang2017optimum}%
  \BibitemOpen
  \bibfield  {author} {\bibinfo {author} {\bibfnamefont {Q.}~\bibnamefont
  {Zhuang}}, \bibinfo {author} {\bibfnamefont {Z.}~\bibnamefont {Zhang}}, \
  and\ \bibinfo {author} {\bibfnamefont {J.~H.}\ \bibnamefont {Shapiro}},\
  }\href@noop {} {\bibfield  {journal} {\bibinfo  {journal} {Phys. Rev. Lett.}\
  }\textbf {\bibinfo {volume} {118}},\ \bibinfo {pages} {040801} (\bibinfo
  {year} {2017})}\BibitemShut {NoStop}%
\bibitem [{\citenamefont {Zhuang}\ \emph
  {et~al.}(2018{\natexlab{a}})\citenamefont {Zhuang}, \citenamefont {Zhang},\
  and\ \citenamefont {Shapiro}}]{zhuang2017distributed}%
  \BibitemOpen
  \bibfield  {author} {\bibinfo {author} {\bibfnamefont {Q.}~\bibnamefont
  {Zhuang}}, \bibinfo {author} {\bibfnamefont {Z.}~\bibnamefont {Zhang}}, \
  and\ \bibinfo {author} {\bibfnamefont {J.~H.}\ \bibnamefont {Shapiro}},\
  }\href {\doibase 10.1103/PhysRevA.97.032329} {\bibfield  {journal} {\bibinfo
  {journal} {Phys. Rev. A}\ }\textbf {\bibinfo {volume} {97}},\ \bibinfo
  {pages} {032329} (\bibinfo {year} {2018}{\natexlab{a}})}\BibitemShut
  {NoStop}%
\bibitem [{\citenamefont {Grosshans}\ and\ \citenamefont
  {Grangier}(2002)}]{grosshans2002continuous}%
  \BibitemOpen
  \bibfield  {author} {\bibinfo {author} {\bibfnamefont {F.}~\bibnamefont
  {Grosshans}}\ and\ \bibinfo {author} {\bibfnamefont {P.}~\bibnamefont
  {Grangier}},\ }\href@noop {} {\bibfield  {journal} {\bibinfo  {journal}
  {Phys. Rev. Lett.}\ }\textbf {\bibinfo {volume} {88}},\ \bibinfo {pages}
  {057902} (\bibinfo {year} {2002})}\BibitemShut {NoStop}%
\bibitem [{\citenamefont {Eisert}\ \emph {et~al.}(2002)\citenamefont {Eisert},
  \citenamefont {Scheel},\ and\ \citenamefont {Plenio}}]{eisert2002distilling}%
  \BibitemOpen
  \bibfield  {author} {\bibinfo {author} {\bibfnamefont {J.}~\bibnamefont
  {Eisert}}, \bibinfo {author} {\bibfnamefont {S.}~\bibnamefont {Scheel}}, \
  and\ \bibinfo {author} {\bibfnamefont {M.}~\bibnamefont {Plenio}},\
  }\href@noop {} {\bibfield  {journal} {\bibinfo  {journal} {Phys. Rev. Lett.}\
  }\textbf {\bibinfo {volume} {89}},\ \bibinfo {pages} {137903} (\bibinfo
  {year} {2002})}\BibitemShut {NoStop}%
\bibitem [{\citenamefont {Fiur{\'a}{\v{s}}ek}(2002)}]{fiuravsek2002gaussian}%
  \BibitemOpen
  \bibfield  {author} {\bibinfo {author} {\bibfnamefont {J.}~\bibnamefont
  {Fiur{\'a}{\v{s}}ek}},\ }\href@noop {} {\bibfield  {journal} {\bibinfo
  {journal} {Phys. Rev. Lett.}\ }\textbf {\bibinfo {volume} {89}},\ \bibinfo
  {pages} {137904} (\bibinfo {year} {2002})}\BibitemShut {NoStop}%
\bibitem [{\citenamefont {Zhang}\ and\ \citenamefont {van
  Loock}(2010)}]{zhang2010distillation}%
  \BibitemOpen
  \bibfield  {author} {\bibinfo {author} {\bibfnamefont {S.~L.}\ \bibnamefont
  {Zhang}}\ and\ \bibinfo {author} {\bibfnamefont {P.}~\bibnamefont {van
  Loock}},\ }\href@noop {} {\bibfield  {journal} {\bibinfo  {journal} {Phys.
  Rev.A}\ }\textbf {\bibinfo {volume} {82}},\ \bibinfo {pages} {062316}
  (\bibinfo {year} {2010})}\BibitemShut {NoStop}%
\bibitem [{\citenamefont {Niset}\ \emph {et~al.}(2009)\citenamefont {Niset},
  \citenamefont {Fiur{\'a}{\v{s}}ek},\ and\ \citenamefont
  {Cerf}}]{niset2009no}%
  \BibitemOpen
  \bibfield  {author} {\bibinfo {author} {\bibfnamefont {J.}~\bibnamefont
  {Niset}}, \bibinfo {author} {\bibfnamefont {J.}~\bibnamefont
  {Fiur{\'a}{\v{s}}ek}}, \ and\ \bibinfo {author} {\bibfnamefont {N.~J.}\
  \bibnamefont {Cerf}},\ }\href@noop {} {\bibfield  {journal} {\bibinfo
  {journal} {Phys. Rev. Lett.}\ }\textbf {\bibinfo {volume} {102}},\ \bibinfo
  {pages} {120501} (\bibinfo {year} {2009})}\BibitemShut {NoStop}%
\bibitem [{Note1()}]{Note1}%
  \BibitemOpen
  \bibinfo {note} {If one trusts the device, then Gaussian states and
  operations suffices for Bell's inequality violation~\cite
  {ralph2000proposal,Oliver2018violation}.}\BibitemShut {Stop}%
\bibitem [{\citenamefont {Banaszek}\ and\ \citenamefont
  {W{\'o}dkiewicz}(1998)}]{banaszek1998nonlocality}%
  \BibitemOpen
  \bibfield  {author} {\bibinfo {author} {\bibfnamefont {K.}~\bibnamefont
  {Banaszek}}\ and\ \bibinfo {author} {\bibfnamefont {K.}~\bibnamefont
  {W{\'o}dkiewicz}},\ }\href@noop {} {\bibfield  {journal} {\bibinfo  {journal}
  {Phys. Rev. A}\ }\textbf {\bibinfo {volume} {58}},\ \bibinfo {pages} {4345}
  (\bibinfo {year} {1998})}\BibitemShut {NoStop}%
\bibitem [{\citenamefont {Banaszek}\ and\ \citenamefont
  {W{\'o}dkiewicz}(1999)}]{banaszek1999testing}%
  \BibitemOpen
  \bibfield  {author} {\bibinfo {author} {\bibfnamefont {K.}~\bibnamefont
  {Banaszek}}\ and\ \bibinfo {author} {\bibfnamefont {K.}~\bibnamefont
  {W{\'o}dkiewicz}},\ }\href@noop {} {\bibfield  {journal} {\bibinfo  {journal}
  {Phys. Rev. Lett.}\ }\textbf {\bibinfo {volume} {82}},\ \bibinfo {pages}
  {2009} (\bibinfo {year} {1999})}\BibitemShut {NoStop}%
\bibitem [{\citenamefont {Filip}\ and\ \citenamefont
  {Mi{\v{s}}ta~Jr}(2002)}]{filip2002violation}%
  \BibitemOpen
  \bibfield  {author} {\bibinfo {author} {\bibfnamefont {R.}~\bibnamefont
  {Filip}}\ and\ \bibinfo {author} {\bibfnamefont {L.}~\bibnamefont
  {Mi{\v{s}}ta~Jr}},\ }\href@noop {} {\bibfield  {journal} {\bibinfo  {journal}
  {Phys. Rev. A}\ }\textbf {\bibinfo {volume} {66}},\ \bibinfo {pages} {044309}
  (\bibinfo {year} {2002})}\BibitemShut {NoStop}%
\bibitem [{\citenamefont {Chen}\ \emph {et~al.}(2002)\citenamefont {Chen},
  \citenamefont {Pan}, \citenamefont {Hou},\ and\ \citenamefont
  {Zhang}}]{chen2002maximal}%
  \BibitemOpen
  \bibfield  {author} {\bibinfo {author} {\bibfnamefont {Z.-B.}\ \bibnamefont
  {Chen}}, \bibinfo {author} {\bibfnamefont {J.-W.}\ \bibnamefont {Pan}},
  \bibinfo {author} {\bibfnamefont {G.}~\bibnamefont {Hou}}, \ and\ \bibinfo
  {author} {\bibfnamefont {Y.-D.}\ \bibnamefont {Zhang}},\ }\href@noop {}
  {\bibfield  {journal} {\bibinfo  {journal} {Phys. Rev. Lett.}\ }\textbf
  {\bibinfo {volume} {88}},\ \bibinfo {pages} {040406} (\bibinfo {year}
  {2002})}\BibitemShut {NoStop}%
\bibitem [{\citenamefont {Nha}\ and\ \citenamefont
  {Carmichael}(2004)}]{nha2004proposed}%
  \BibitemOpen
  \bibfield  {author} {\bibinfo {author} {\bibfnamefont {H.}~\bibnamefont
  {Nha}}\ and\ \bibinfo {author} {\bibfnamefont {H.}~\bibnamefont
  {Carmichael}},\ }\href@noop {} {\bibfield  {journal} {\bibinfo  {journal}
  {Phys. Rev. Lett.}\ }\textbf {\bibinfo {volume} {93}},\ \bibinfo {pages}
  {020401} (\bibinfo {year} {2004})}\BibitemShut {NoStop}%
\bibitem [{\citenamefont {Invernizzi}\ \emph {et~al.}(2005)\citenamefont
  {Invernizzi}, \citenamefont {Olivares}, \citenamefont {Paris},\ and\
  \citenamefont {Banaszek}}]{invernizzi2005effect}%
  \BibitemOpen
  \bibfield  {author} {\bibinfo {author} {\bibfnamefont {C.}~\bibnamefont
  {Invernizzi}}, \bibinfo {author} {\bibfnamefont {S.}~\bibnamefont
  {Olivares}}, \bibinfo {author} {\bibfnamefont {M.~G.}\ \bibnamefont {Paris}},
  \ and\ \bibinfo {author} {\bibfnamefont {K.}~\bibnamefont {Banaszek}},\
  }\href@noop {} {\bibfield  {journal} {\bibinfo  {journal} {Phys. Rev. A}\
  }\textbf {\bibinfo {volume} {72}},\ \bibinfo {pages} {042105} (\bibinfo
  {year} {2005})}\BibitemShut {NoStop}%
\bibitem [{\citenamefont {Garc{\'\i}a-Patr{\'o}n}\ \emph
  {et~al.}(2005)\citenamefont {Garc{\'\i}a-Patr{\'o}n}, \citenamefont
  {Fiur{\'a}{\v{s}}ek},\ and\ \citenamefont {Cerf}}]{garcia2005loophole}%
  \BibitemOpen
  \bibfield  {author} {\bibinfo {author} {\bibfnamefont {R.}~\bibnamefont
  {Garc{\'\i}a-Patr{\'o}n}}, \bibinfo {author} {\bibfnamefont {J.}~\bibnamefont
  {Fiur{\'a}{\v{s}}ek}}, \ and\ \bibinfo {author} {\bibfnamefont {N.~J.}\
  \bibnamefont {Cerf}},\ }\href@noop {} {\bibfield  {journal} {\bibinfo
  {journal} {Phys. Rev. A}\ }\textbf {\bibinfo {volume} {71}},\ \bibinfo
  {pages} {022105} (\bibinfo {year} {2005})}\BibitemShut {NoStop}%
\bibitem [{\citenamefont {Ferraro}\ and\ \citenamefont
  {Paris}(2005)}]{ferraro2005nonlocality}%
  \BibitemOpen
  \bibfield  {author} {\bibinfo {author} {\bibfnamefont {A.}~\bibnamefont
  {Ferraro}}\ and\ \bibinfo {author} {\bibfnamefont {M.~G.}\ \bibnamefont
  {Paris}},\ }\href@noop {} {\bibfield  {journal} {\bibinfo  {journal} {J. Opt.
  B: Quantum Semiclass. Opt.}\ }\textbf {\bibinfo {volume} {7}},\ \bibinfo
  {pages} {174} (\bibinfo {year} {2005})}\BibitemShut {NoStop}%
\bibitem [{\citenamefont {Lloyd}\ and\ \citenamefont
  {Braunstein}(1999)}]{lloyd1999quantum}%
  \BibitemOpen
  \bibfield  {author} {\bibinfo {author} {\bibfnamefont {S.}~\bibnamefont
  {Lloyd}}\ and\ \bibinfo {author} {\bibfnamefont {S.~L.}\ \bibnamefont
  {Braunstein}},\ }\href@noop {} {\bibfield  {journal} {\bibinfo  {journal}
  {Phys. Rev. Lett.}\ }\textbf {\bibinfo {volume} {82}},\ \bibinfo {pages}
  {1784} (\bibinfo {year} {1999})}\BibitemShut {NoStop}%
\bibitem [{\citenamefont {Bartlett}\ and\ \citenamefont
  {Sanders}(2002)}]{bartlett2002universal}%
  \BibitemOpen
  \bibfield  {author} {\bibinfo {author} {\bibfnamefont {S.~D.}\ \bibnamefont
  {Bartlett}}\ and\ \bibinfo {author} {\bibfnamefont {B.~C.}\ \bibnamefont
  {Sanders}},\ }\href@noop {} {\bibfield  {journal} {\bibinfo  {journal} {Phys.
  Rev. A}\ }\textbf {\bibinfo {volume} {65}},\ \bibinfo {pages} {042304}
  (\bibinfo {year} {2002})}\BibitemShut {NoStop}%
\bibitem [{\citenamefont {Ohliger}\ \emph {et~al.}(2010)\citenamefont
  {Ohliger}, \citenamefont {Kieling},\ and\ \citenamefont
  {Eisert}}]{ohliger2010limitations}%
  \BibitemOpen
  \bibfield  {author} {\bibinfo {author} {\bibfnamefont {M.}~\bibnamefont
  {Ohliger}}, \bibinfo {author} {\bibfnamefont {K.}~\bibnamefont {Kieling}}, \
  and\ \bibinfo {author} {\bibfnamefont {J.}~\bibnamefont {Eisert}},\
  }\href@noop {} {\bibfield  {journal} {\bibinfo  {journal} {Phys. Rev. A}\
  }\textbf {\bibinfo {volume} {82}},\ \bibinfo {pages} {042336} (\bibinfo
  {year} {2010})}\BibitemShut {NoStop}%
\bibitem [{\citenamefont {Menicucci}\ \emph {et~al.}(2006)\citenamefont
  {Menicucci}, \citenamefont {van Loock}, \citenamefont {Gu}, \citenamefont
  {Weedbrook}, \citenamefont {Ralph},\ and\ \citenamefont
  {Nielsen}}]{menicucci2006universal}%
  \BibitemOpen
  \bibfield  {author} {\bibinfo {author} {\bibfnamefont {N.~C.}\ \bibnamefont
  {Menicucci}}, \bibinfo {author} {\bibfnamefont {P.}~\bibnamefont {van
  Loock}}, \bibinfo {author} {\bibfnamefont {M.}~\bibnamefont {Gu}}, \bibinfo
  {author} {\bibfnamefont {C.}~\bibnamefont {Weedbrook}}, \bibinfo {author}
  {\bibfnamefont {T.~C.}\ \bibnamefont {Ralph}}, \ and\ \bibinfo {author}
  {\bibfnamefont {M.~A.}\ \bibnamefont {Nielsen}},\ }\href@noop {} {\bibfield
  {journal} {\bibinfo  {journal} {Phys. Rev. Lett.}\ }\textbf {\bibinfo
  {volume} {97}},\ \bibinfo {pages} {110501} (\bibinfo {year}
  {2006})}\BibitemShut {NoStop}%
\bibitem [{\citenamefont {Lami}\ \emph {et~al.}(2018)\citenamefont {Lami},
  \citenamefont {Regula}, \citenamefont {Wang}, \citenamefont {Nichols},
  \citenamefont {Winter},\ and\ \citenamefont {Adesso}}]{lami2018gaussian}%
  \BibitemOpen
  \bibfield  {author} {\bibinfo {author} {\bibfnamefont {L.}~\bibnamefont
  {Lami}}, \bibinfo {author} {\bibfnamefont {B.}~\bibnamefont {Regula}},
  \bibinfo {author} {\bibfnamefont {X.}~\bibnamefont {Wang}}, \bibinfo {author}
  {\bibfnamefont {R.}~\bibnamefont {Nichols}}, \bibinfo {author} {\bibfnamefont
  {A.}~\bibnamefont {Winter}}, \ and\ \bibinfo {author} {\bibfnamefont
  {G.}~\bibnamefont {Adesso}},\ }\href@noop {} {\bibfield  {journal} {\bibinfo
  {journal} {arXiv:1801.05450}\ } (\bibinfo {year} {2018})}\BibitemShut
  {NoStop}%
\bibitem [{\citenamefont {Marian}\ and\ \citenamefont
  {Marian}(2013)}]{marian2013relative}%
  \BibitemOpen
  \bibfield  {author} {\bibinfo {author} {\bibfnamefont {P.}~\bibnamefont
  {Marian}}\ and\ \bibinfo {author} {\bibfnamefont {T.~A.}\ \bibnamefont
  {Marian}},\ }\href@noop {} {\bibfield  {journal} {\bibinfo  {journal} {Phys.
  Rev. A}\ }\textbf {\bibinfo {volume} {88}},\ \bibinfo {pages} {012322}
  (\bibinfo {year} {2013})}\BibitemShut {NoStop}%
\bibitem [{\citenamefont {Genoni}\ \emph {et~al.}(2008)\citenamefont {Genoni},
  \citenamefont {Paris},\ and\ \citenamefont
  {Banaszek}}]{genoni2008quantifying}%
  \BibitemOpen
  \bibfield  {author} {\bibinfo {author} {\bibfnamefont {M.~G.}\ \bibnamefont
  {Genoni}}, \bibinfo {author} {\bibfnamefont {M.~G.}\ \bibnamefont {Paris}}, \
  and\ \bibinfo {author} {\bibfnamefont {K.}~\bibnamefont {Banaszek}},\
  }\href@noop {} {\bibfield  {journal} {\bibinfo  {journal} {Phys. Rev. A}\
  }\textbf {\bibinfo {volume} {78}},\ \bibinfo {pages} {060303} (\bibinfo
  {year} {2008})}\BibitemShut {NoStop}%
\bibitem [{\citenamefont {Genoni}\ and\ \citenamefont
  {Paris}(2010)}]{genoni2010quantifying}%
  \BibitemOpen
  \bibfield  {author} {\bibinfo {author} {\bibfnamefont {M.~G.}\ \bibnamefont
  {Genoni}}\ and\ \bibinfo {author} {\bibfnamefont {M.~G.}\ \bibnamefont
  {Paris}},\ }\href@noop {} {\bibfield  {journal} {\bibinfo  {journal} {Phys.
  Rev. A}\ }\textbf {\bibinfo {volume} {82}},\ \bibinfo {pages} {052341}
  (\bibinfo {year} {2010})}\BibitemShut {NoStop}%
\bibitem [{\citenamefont {Zhuang}\ \emph
  {et~al.}(2018{\natexlab{b}})\citenamefont {Zhuang}, \citenamefont {Shor},\
  and\ \citenamefont {Shapiro}}]{zhuang2018resource}%
  \BibitemOpen
  \bibfield  {author} {\bibinfo {author} {\bibfnamefont {Q.}~\bibnamefont
  {Zhuang}}, \bibinfo {author} {\bibfnamefont {P.~W.}\ \bibnamefont {Shor}}, \
  and\ \bibinfo {author} {\bibfnamefont {J.~H.}\ \bibnamefont {Shapiro}},\
  }\href@noop {} {\bibfield  {journal} {\bibinfo  {journal} {arXiv:1803.07580}\
  } (\bibinfo {year} {2018}{\natexlab{b}})}\BibitemShut {NoStop}%
\bibitem [{\citenamefont {Horodecki}\ and\ \citenamefont
  {Oppenheim}(2013)}]{Horodecki2013}%
  \BibitemOpen
  \bibfield  {author} {\bibinfo {author} {\bibfnamefont {M.}~\bibnamefont
  {Horodecki}}\ and\ \bibinfo {author} {\bibfnamefont {J.}~\bibnamefont
  {Oppenheim}},\ }\href {\doibase 10.1142/S0217979213450197} {\bibfield
  {journal} {\bibinfo  {journal} {Int. J. Mod. Phys. B}\ }\textbf {\bibinfo
  {volume} {27}},\ \bibinfo {pages} {1345019} (\bibinfo {year}
  {2013})}\BibitemShut {NoStop}%
\bibitem [{\citenamefont {Brand\~ao}\ and\ \citenamefont
  {Gour}(2015)}]{Brandao2015}%
  \BibitemOpen
  \bibfield  {author} {\bibinfo {author} {\bibfnamefont {F.~G. S.~L.}\
  \bibnamefont {Brand\~ao}}\ and\ \bibinfo {author} {\bibfnamefont
  {G.}~\bibnamefont {Gour}},\ }\href {\doibase 10.1103/PhysRevLett.115.070503}
  {\bibfield  {journal} {\bibinfo  {journal} {Phys. Rev. Lett.}\ }\textbf
  {\bibinfo {volume} {115}},\ \bibinfo {pages} {070503} (\bibinfo {year}
  {2015})}\BibitemShut {NoStop}%
\bibitem [{\citenamefont {Chitambar}\ and\ \citenamefont
  {Gour}(2018)}]{chitambar2018quantum}%
  \BibitemOpen
  \bibfield  {author} {\bibinfo {author} {\bibfnamefont {E.}~\bibnamefont
  {Chitambar}}\ and\ \bibinfo {author} {\bibfnamefont {G.}~\bibnamefont
  {Gour}},\ }\href@noop {} {\bibfield  {journal} {\bibinfo  {journal}
  {arXiv:1806.06107}\ } (\bibinfo {year} {2018})}\BibitemShut {NoStop}%
\bibitem [{Note2()}]{Note2}%
  \BibitemOpen
  \bibinfo {note} {The converse is not true. And various ways to test whether a
  quantum state is in the such a free set are being considered~\cite
  {filip2011detecting,genoni2013detecting,park2015testing,hughes2014quantum,jevzek2011experimental,Happ2018}}\BibitemShut
  {NoStop}%
\bibitem [{\citenamefont {Mari}\ and\ \citenamefont
  {Eisert}(2012)}]{mari2012positive}%
  \BibitemOpen
  \bibfield  {author} {\bibinfo {author} {\bibfnamefont {A.}~\bibnamefont
  {Mari}}\ and\ \bibinfo {author} {\bibfnamefont {J.}~\bibnamefont {Eisert}},\
  }\href@noop {} {\bibfield  {journal} {\bibinfo  {journal} {Phys. Rev. Lett.}\
  }\textbf {\bibinfo {volume} {109}},\ \bibinfo {pages} {230503} (\bibinfo
  {year} {2012})}\BibitemShut {NoStop}%
\bibitem [{\citenamefont {Veitch}\ \emph {et~al.}(2013)\citenamefont {Veitch},
  \citenamefont {Wiebe}, \citenamefont {Ferrie},\ and\ \citenamefont
  {Emerson}}]{Veitch2013}%
  \BibitemOpen
  \bibfield  {author} {\bibinfo {author} {\bibfnamefont {V.}~\bibnamefont
  {Veitch}}, \bibinfo {author} {\bibfnamefont {N.}~\bibnamefont {Wiebe}},
  \bibinfo {author} {\bibfnamefont {C.}~\bibnamefont {Ferrie}}, \ and\ \bibinfo
  {author} {\bibfnamefont {J.}~\bibnamefont {Emerson}},\ }\href
  {http://stacks.iop.org/1367-2630/15/i=1/a=013037} {\bibfield  {journal}
  {\bibinfo  {journal} {New J. Phys.}\ }\textbf {\bibinfo {volume} {15}},\
  \bibinfo {pages} {013037} (\bibinfo {year} {2013})}\BibitemShut {NoStop}%
\bibitem [{\citenamefont {Veitch}\ \emph {et~al.}(2014)\citenamefont {Veitch},
  \citenamefont {Mousavian}, \citenamefont {Gottesman},\ and\ \citenamefont
  {Emerson}}]{veitch2014resource}%
  \BibitemOpen
  \bibfield  {author} {\bibinfo {author} {\bibfnamefont {V.}~\bibnamefont
  {Veitch}}, \bibinfo {author} {\bibfnamefont {S.~H.}\ \bibnamefont
  {Mousavian}}, \bibinfo {author} {\bibfnamefont {D.}~\bibnamefont
  {Gottesman}}, \ and\ \bibinfo {author} {\bibfnamefont {J.}~\bibnamefont
  {Emerson}},\ }\href@noop {} {\bibfield  {journal} {\bibinfo  {journal} {New
  J. Phys.}\ }\textbf {\bibinfo {volume} {16}},\ \bibinfo {pages} {013009}
  (\bibinfo {year} {2014})}\BibitemShut {NoStop}%
\bibitem [{\citenamefont {Howard}\ and\ \citenamefont
  {Campbell}(2017)}]{howard2017application}%
  \BibitemOpen
  \bibfield  {author} {\bibinfo {author} {\bibfnamefont {M.}~\bibnamefont
  {Howard}}\ and\ \bibinfo {author} {\bibfnamefont {E.}~\bibnamefont
  {Campbell}},\ }\href@noop {} {\bibfield  {journal} {\bibinfo  {journal}
  {Phys. Rev. Lett.}\ }\textbf {\bibinfo {volume} {118}},\ \bibinfo {pages}
  {090501} (\bibinfo {year} {2017})}\BibitemShut {NoStop}%
\bibitem [{\citenamefont {Pashayan}\ \emph {et~al.}(2015)\citenamefont
  {Pashayan}, \citenamefont {Wallman},\ and\ \citenamefont
  {Bartlett}}]{pashayan2015estimating}%
  \BibitemOpen
  \bibfield  {author} {\bibinfo {author} {\bibfnamefont {H.}~\bibnamefont
  {Pashayan}}, \bibinfo {author} {\bibfnamefont {J.~J.}\ \bibnamefont
  {Wallman}}, \ and\ \bibinfo {author} {\bibfnamefont {S.~D.}\ \bibnamefont
  {Bartlett}},\ }\href@noop {} {\bibfield  {journal} {\bibinfo  {journal}
  {Phys. Rev. Lett.}\ }\textbf {\bibinfo {volume} {115}},\ \bibinfo {pages}
  {070501} (\bibinfo {year} {2015})}\BibitemShut {NoStop}%
\bibitem [{\citenamefont {Gottesman}\ \emph {et~al.}(2001)\citenamefont
  {Gottesman}, \citenamefont {Kitaev},\ and\ \citenamefont
  {Preskill}}]{gottesman2001encoding}%
  \BibitemOpen
  \bibfield  {author} {\bibinfo {author} {\bibfnamefont {D.}~\bibnamefont
  {Gottesman}}, \bibinfo {author} {\bibfnamefont {A.}~\bibnamefont {Kitaev}}, \
  and\ \bibinfo {author} {\bibfnamefont {J.}~\bibnamefont {Preskill}},\
  }\href@noop {} {\bibfield  {journal} {\bibinfo  {journal} {Phys. Rev. A}\
  }\textbf {\bibinfo {volume} {64}},\ \bibinfo {pages} {012310} (\bibinfo
  {year} {2001})}\BibitemShut {NoStop}%
\bibitem [{\citenamefont {Sabapathy}\ and\ \citenamefont
  {Weedbrook}(2018)}]{sabapathy2018states}%
  \BibitemOpen
  \bibfield  {author} {\bibinfo {author} {\bibfnamefont {K.~K.}\ \bibnamefont
  {Sabapathy}}\ and\ \bibinfo {author} {\bibfnamefont {C.}~\bibnamefont
  {Weedbrook}},\ }\href {\doibase 10.1103/PhysRevA.97.062315} {\bibfield
  {journal} {\bibinfo  {journal} {Phys. Rev. A}\ }\textbf {\bibinfo {volume}
  {97}},\ \bibinfo {pages} {062315} (\bibinfo {year} {2018})}\BibitemShut
  {NoStop}%
\bibitem [{\citenamefont {Suzuki}\ \emph {et~al.}(2006)\citenamefont {Suzuki},
  \citenamefont {Takeoka}, \citenamefont {Sasaki}, \citenamefont {Andersen},\
  and\ \citenamefont {Kannari}}]{Suzuki2006}%
  \BibitemOpen
  \bibfield  {author} {\bibinfo {author} {\bibfnamefont {S.}~\bibnamefont
  {Suzuki}}, \bibinfo {author} {\bibfnamefont {M.}~\bibnamefont {Takeoka}},
  \bibinfo {author} {\bibfnamefont {M.}~\bibnamefont {Sasaki}}, \bibinfo
  {author} {\bibfnamefont {U.~L.}\ \bibnamefont {Andersen}}, \ and\ \bibinfo
  {author} {\bibfnamefont {F.}~\bibnamefont {Kannari}},\ }\href {\doibase
  10.1103/PhysRevA.73.042304} {\bibfield  {journal} {\bibinfo  {journal} {Phys.
  Rev. A}\ }\textbf {\bibinfo {volume} {73}},\ \bibinfo {pages} {042304}
  (\bibinfo {year} {2006})}\BibitemShut {NoStop}%
\bibitem [{\citenamefont {Heersink}\ \emph {et~al.}(2006)\citenamefont
  {Heersink}, \citenamefont {Marquardt}, \citenamefont {Dong}, \citenamefont
  {Filip}, \citenamefont {Lorenz}, \citenamefont {Leuchs},\ and\ \citenamefont
  {Andersen}}]{Heersink2006}%
  \BibitemOpen
  \bibfield  {author} {\bibinfo {author} {\bibfnamefont {J.}~\bibnamefont
  {Heersink}}, \bibinfo {author} {\bibfnamefont {C.}~\bibnamefont {Marquardt}},
  \bibinfo {author} {\bibfnamefont {R.}~\bibnamefont {Dong}}, \bibinfo {author}
  {\bibfnamefont {R.}~\bibnamefont {Filip}}, \bibinfo {author} {\bibfnamefont
  {S.}~\bibnamefont {Lorenz}}, \bibinfo {author} {\bibfnamefont
  {G.}~\bibnamefont {Leuchs}}, \ and\ \bibinfo {author} {\bibfnamefont {U.~L.}\
  \bibnamefont {Andersen}},\ }\href {\doibase 10.1103/PhysRevLett.96.253601}
  {\bibfield  {journal} {\bibinfo  {journal} {Phys. Rev. Lett.}\ }\textbf
  {\bibinfo {volume} {96}},\ \bibinfo {pages} {253601} (\bibinfo {year}
  {2006})}\BibitemShut {NoStop}%
\bibitem [{\citenamefont {Dennis}(2001)}]{Dennis2001}%
  \BibitemOpen
  \bibfield  {author} {\bibinfo {author} {\bibfnamefont {E.}~\bibnamefont
  {Dennis}},\ }\href {\doibase 10.1103/PhysRevA.63.052314} {\bibfield
  {journal} {\bibinfo  {journal} {Phys. Rev. A}\ }\textbf {\bibinfo {volume}
  {63}},\ \bibinfo {pages} {052314} (\bibinfo {year} {2001})}\BibitemShut
  {NoStop}%
\bibitem [{\citenamefont {Bravyi}\ and\ \citenamefont
  {Kitaev}(2005)}]{Bravyi2005}%
  \BibitemOpen
  \bibfield  {author} {\bibinfo {author} {\bibfnamefont {S.}~\bibnamefont
  {Bravyi}}\ and\ \bibinfo {author} {\bibfnamefont {A.}~\bibnamefont
  {Kitaev}},\ }\href {\doibase 10.1103/PhysRevA.71.022316} {\bibfield
  {journal} {\bibinfo  {journal} {Phys. Rev. A}\ }\textbf {\bibinfo {volume}
  {71}},\ \bibinfo {pages} {022316} (\bibinfo {year} {2005})}\BibitemShut
  {NoStop}%
\bibitem [{\citenamefont {Reichardt}(2005)}]{Reichardt2005}%
  \BibitemOpen
  \bibfield  {author} {\bibinfo {author} {\bibfnamefont {B.~W.}\ \bibnamefont
  {Reichardt}},\ }\href {\doibase 10.1007/s11128-005-7654-8} {\bibfield
  {journal} {\bibinfo  {journal} {Quantum Inf. Process.}\ }\textbf {\bibinfo
  {volume} {4}},\ \bibinfo {pages} {251} (\bibinfo {year} {2005})}\BibitemShut
  {NoStop}%
\bibitem [{\citenamefont {Bravyi}\ and\ \citenamefont
  {Haah}(2012)}]{Bravyi2012}%
  \BibitemOpen
  \bibfield  {author} {\bibinfo {author} {\bibfnamefont {S.}~\bibnamefont
  {Bravyi}}\ and\ \bibinfo {author} {\bibfnamefont {J.}~\bibnamefont {Haah}},\
  }\href {\doibase 10.1103/PhysRevA.86.052329} {\bibfield  {journal} {\bibinfo
  {journal} {Phys. Rev. A}\ }\textbf {\bibinfo {volume} {86}},\ \bibinfo
  {pages} {052329} (\bibinfo {year} {2012})}\BibitemShut {NoStop}%
\bibitem [{\citenamefont {Meier}\ \emph {et~al.}(2013)\citenamefont {Meier},
  \citenamefont {Eastin},\ and\ \citenamefont {Knill}}]{Meier2013}%
  \BibitemOpen
  \bibfield  {author} {\bibinfo {author} {\bibfnamefont {A.~M.}\ \bibnamefont
  {Meier}}, \bibinfo {author} {\bibfnamefont {B.}~\bibnamefont {Eastin}}, \
  and\ \bibinfo {author} {\bibfnamefont {E.}~\bibnamefont {Knill}},\ }\href
  {http://www.rintonpress.com/journals/qiconline.html} {\bibfield  {journal}
  {\bibinfo  {journal} {Quantum Inf. Comput.}\ }\textbf {\bibinfo {volume}
  {13}},\ \bibinfo {pages} {195} (\bibinfo {year} {2013})}\BibitemShut
  {NoStop}%
\bibitem [{\citenamefont {Eastin}(2013)}]{Eastin2013}%
  \BibitemOpen
  \bibfield  {author} {\bibinfo {author} {\bibfnamefont {B.}~\bibnamefont
  {Eastin}},\ }\href {\doibase 10.1103/PhysRevA.87.032321} {\bibfield
  {journal} {\bibinfo  {journal} {Phys. Rev. A}\ }\textbf {\bibinfo {volume}
  {87}},\ \bibinfo {pages} {032321} (\bibinfo {year} {2013})}\BibitemShut
  {NoStop}%
\bibitem [{\citenamefont {Weedbrook}\ \emph {et~al.}(2012)\citenamefont
  {Weedbrook}, \citenamefont {Pirandola}, \citenamefont {Garc\'{\i}a-Patr\'on},
  \citenamefont {Cerf}, \citenamefont {Ralph}, \citenamefont {Shapiro},\ and\
  \citenamefont {Lloyd}}]{Weedbrook_2012}%
  \BibitemOpen
  \bibfield  {author} {\bibinfo {author} {\bibfnamefont {C.}~\bibnamefont
  {Weedbrook}}, \bibinfo {author} {\bibfnamefont {S.}~\bibnamefont
  {Pirandola}}, \bibinfo {author} {\bibfnamefont {R.}~\bibnamefont
  {Garc\'{\i}a-Patr\'on}}, \bibinfo {author} {\bibfnamefont {N.~J.}\
  \bibnamefont {Cerf}}, \bibinfo {author} {\bibfnamefont {T.~C.}\ \bibnamefont
  {Ralph}}, \bibinfo {author} {\bibfnamefont {J.~H.}\ \bibnamefont {Shapiro}},
  \ and\ \bibinfo {author} {\bibfnamefont {S.}~\bibnamefont {Lloyd}},\ }\href
  {\doibase 10.1103/RevModPhys.84.621} {\bibfield  {journal} {\bibinfo
  {journal} {Rev. Mod. Phys.}\ }\textbf {\bibinfo {volume} {84}},\ \bibinfo
  {pages} {621} (\bibinfo {year} {2012})}\BibitemShut {NoStop}%
\bibitem [{\citenamefont {Ghose}\ and\ \citenamefont
  {Sanders}(2007)}]{ghose2007non}%
  \BibitemOpen
  \bibfield  {author} {\bibinfo {author} {\bibfnamefont {S.}~\bibnamefont
  {Ghose}}\ and\ \bibinfo {author} {\bibfnamefont {B.~C.}\ \bibnamefont
  {Sanders}},\ }\href@noop {} {\bibfield  {journal} {\bibinfo  {journal} {J.
  Mod. Opt.}\ }\textbf {\bibinfo {volume} {54}},\ \bibinfo {pages} {855}
  (\bibinfo {year} {2007})}\BibitemShut {NoStop}%
\bibitem [{\citenamefont {Sefi}\ and\ \citenamefont {van
  Loock}(2011)}]{sefi2011decompose}%
  \BibitemOpen
  \bibfield  {author} {\bibinfo {author} {\bibfnamefont {S.}~\bibnamefont
  {Sefi}}\ and\ \bibinfo {author} {\bibfnamefont {P.}~\bibnamefont {van
  Loock}},\ }\href@noop {} {\bibfield  {journal} {\bibinfo  {journal} {Phys.
  Rev. Lett.}\ }\textbf {\bibinfo {volume} {107}},\ \bibinfo {pages} {170501}
  (\bibinfo {year} {2011})}\BibitemShut {NoStop}%
\bibitem [{\citenamefont {Marek}\ \emph {et~al.}(2011)\citenamefont {Marek},
  \citenamefont {Filip},\ and\ \citenamefont
  {Furusawa}}]{marek2011deterministic}%
  \BibitemOpen
  \bibfield  {author} {\bibinfo {author} {\bibfnamefont {P.}~\bibnamefont
  {Marek}}, \bibinfo {author} {\bibfnamefont {R.}~\bibnamefont {Filip}}, \ and\
  \bibinfo {author} {\bibfnamefont {A.}~\bibnamefont {Furusawa}},\ }\href@noop
  {} {\bibfield  {journal} {\bibinfo  {journal} {Phys. Rev. A}\ }\textbf
  {\bibinfo {volume} {84}},\ \bibinfo {pages} {053802} (\bibinfo {year}
  {2011})}\BibitemShut {NoStop}%
\bibitem [{\citenamefont {Yukawa}\ \emph {et~al.}(2013)\citenamefont {Yukawa},
  \citenamefont {Miyata}, \citenamefont {Yonezawa}, \citenamefont {Marek},
  \citenamefont {Filip},\ and\ \citenamefont {Furusawa}}]{yukawa2013emulating}%
  \BibitemOpen
  \bibfield  {author} {\bibinfo {author} {\bibfnamefont {M.}~\bibnamefont
  {Yukawa}}, \bibinfo {author} {\bibfnamefont {K.}~\bibnamefont {Miyata}},
  \bibinfo {author} {\bibfnamefont {H.}~\bibnamefont {Yonezawa}}, \bibinfo
  {author} {\bibfnamefont {P.}~\bibnamefont {Marek}}, \bibinfo {author}
  {\bibfnamefont {R.}~\bibnamefont {Filip}}, \ and\ \bibinfo {author}
  {\bibfnamefont {A.}~\bibnamefont {Furusawa}},\ }\href@noop {} {\bibfield
  {journal} {\bibinfo  {journal} {Phys. Rev. A}\ }\textbf {\bibinfo {volume}
  {88}},\ \bibinfo {pages} {053816} (\bibinfo {year} {2013})}\BibitemShut
  {NoStop}%
\bibitem [{\citenamefont {Arzani}\ \emph {et~al.}(2017)\citenamefont {Arzani},
  \citenamefont {Treps},\ and\ \citenamefont {Ferrini}}]{arzani2017polynomial}%
  \BibitemOpen
  \bibfield  {author} {\bibinfo {author} {\bibfnamefont {F.}~\bibnamefont
  {Arzani}}, \bibinfo {author} {\bibfnamefont {N.}~\bibnamefont {Treps}}, \
  and\ \bibinfo {author} {\bibfnamefont {G.}~\bibnamefont {Ferrini}},\
  }\href@noop {} {\bibfield  {journal} {\bibinfo  {journal} {Phys. Rev. A}\
  }\textbf {\bibinfo {volume} {95}},\ \bibinfo {pages} {052352} (\bibinfo
  {year} {2017})}\BibitemShut {NoStop}%
\bibitem [{\citenamefont {Aberg}(2006)}]{aberg2006quantifying}%
  \BibitemOpen
  \bibfield  {author} {\bibinfo {author} {\bibfnamefont {J.}~\bibnamefont
  {Aberg}},\ }\href@noop {} {\bibfield  {journal} {\bibinfo  {journal}
  {arXiv:quant-ph/0612146}\ } (\bibinfo {year} {2006})}\BibitemShut {NoStop}%
\bibitem [{\citenamefont {Baumgratz}\ \emph {et~al.}(2014)\citenamefont
  {Baumgratz}, \citenamefont {Cramer},\ and\ \citenamefont
  {Plenio}}]{Baumgratz2014}%
  \BibitemOpen
  \bibfield  {author} {\bibinfo {author} {\bibfnamefont {T.}~\bibnamefont
  {Baumgratz}}, \bibinfo {author} {\bibfnamefont {M.}~\bibnamefont {Cramer}}, \
  and\ \bibinfo {author} {\bibfnamefont {M.~B.}\ \bibnamefont {Plenio}},\
  }\href {\doibase 10.1103/PhysRevLett.113.140401} {\bibfield  {journal}
  {\bibinfo  {journal} {Phys. Rev. Lett.}\ }\textbf {\bibinfo {volume} {113}},\
  \bibinfo {pages} {140401} (\bibinfo {year} {2014})}\BibitemShut {NoStop}%
\bibitem [{\citenamefont {Winter}\ and\ \citenamefont
  {Yang}(2016)}]{winter2016operational}%
  \BibitemOpen
  \bibfield  {author} {\bibinfo {author} {\bibfnamefont {A.}~\bibnamefont
  {Winter}}\ and\ \bibinfo {author} {\bibfnamefont {D.}~\bibnamefont {Yang}},\
  }\href@noop {} {\bibfield  {journal} {\bibinfo  {journal} {Phys. Rev. Lett.}\
  }\textbf {\bibinfo {volume} {116}},\ \bibinfo {pages} {120404} (\bibinfo
  {year} {2016})}\BibitemShut {NoStop}%
\bibitem [{\citenamefont {Yadin}\ \emph {et~al.}(2016)\citenamefont {Yadin},
  \citenamefont {Ma}, \citenamefont {Girolami}, \citenamefont {Gu},\ and\
  \citenamefont {Vedral}}]{Yadin2016}%
  \BibitemOpen
  \bibfield  {author} {\bibinfo {author} {\bibfnamefont {B.}~\bibnamefont
  {Yadin}}, \bibinfo {author} {\bibfnamefont {J.}~\bibnamefont {Ma}}, \bibinfo
  {author} {\bibfnamefont {D.}~\bibnamefont {Girolami}}, \bibinfo {author}
  {\bibfnamefont {M.}~\bibnamefont {Gu}}, \ and\ \bibinfo {author}
  {\bibfnamefont {V.}~\bibnamefont {Vedral}},\ }\href {\doibase
  10.1103/PhysRevX.6.041028} {\bibfield  {journal} {\bibinfo  {journal} {Phys.
  Rev. X}\ }\textbf {\bibinfo {volume} {6}},\ \bibinfo {pages} {041028}
  (\bibinfo {year} {2016})}\BibitemShut {NoStop}%
\bibitem [{\citenamefont {Chitambar}\ and\ \citenamefont
  {Gour}(2016)}]{Chitambar2016}%
  \BibitemOpen
  \bibfield  {author} {\bibinfo {author} {\bibfnamefont {E.}~\bibnamefont
  {Chitambar}}\ and\ \bibinfo {author} {\bibfnamefont {G.}~\bibnamefont
  {Gour}},\ }\href {\doibase 10.1103/PhysRevA.94.052336} {\bibfield  {journal}
  {\bibinfo  {journal} {Phys. Rev. A}\ }\textbf {\bibinfo {volume} {94}},\
  \bibinfo {pages} {052336} (\bibinfo {year} {2016})}\BibitemShut {NoStop}%
\bibitem [{\citenamefont {Nielsen}(1999)}]{nielsen1999conditions}%
  \BibitemOpen
  \bibfield  {author} {\bibinfo {author} {\bibfnamefont {M.~A.}\ \bibnamefont
  {Nielsen}},\ }\href@noop {} {\bibfield  {journal} {\bibinfo  {journal} {Phys.
  Rev. Lett.}\ }\textbf {\bibinfo {volume} {83}},\ \bibinfo {pages} {436}
  (\bibinfo {year} {1999})}\BibitemShut {NoStop}%
\bibitem [{\citenamefont {Genoni}\ \emph {et~al.}(2013)\citenamefont {Genoni},
  \citenamefont {Palma}, \citenamefont {Tufarelli}, \citenamefont {Olivares},
  \citenamefont {Kim},\ and\ \citenamefont {Paris}}]{genoni2013detecting}%
  \BibitemOpen
  \bibfield  {author} {\bibinfo {author} {\bibfnamefont {M.~G.}\ \bibnamefont
  {Genoni}}, \bibinfo {author} {\bibfnamefont {M.~L.}\ \bibnamefont {Palma}},
  \bibinfo {author} {\bibfnamefont {T.}~\bibnamefont {Tufarelli}}, \bibinfo
  {author} {\bibfnamefont {S.}~\bibnamefont {Olivares}}, \bibinfo {author}
  {\bibfnamefont {M.}~\bibnamefont {Kim}}, \ and\ \bibinfo {author}
  {\bibfnamefont {M.~G.}\ \bibnamefont {Paris}},\ }\href@noop {} {\bibfield
  {journal} {\bibinfo  {journal} {Phys. Rev. A}\ }\textbf {\bibinfo {volume}
  {87}},\ \bibinfo {pages} {062104} (\bibinfo {year} {2013})}\BibitemShut
  {NoStop}%
\bibitem [{Note3()}]{Note3}%
  \BibitemOpen
  \bibinfo {note} {We include states of Wigner function with negative values at
  measure zero points in this set. These measure zero negative points should be
  irrelevant to any experiments in reality.}\BibitemShut {Stop}%
\bibitem [{\citenamefont {Filip}\ and\ \citenamefont
  {Mi{\v{s}}ta~Jr}(2011)}]{filip2011detecting}%
  \BibitemOpen
  \bibfield  {author} {\bibinfo {author} {\bibfnamefont {R.}~\bibnamefont
  {Filip}}\ and\ \bibinfo {author} {\bibfnamefont {L.}~\bibnamefont
  {Mi{\v{s}}ta~Jr}},\ }\href@noop {} {\bibfield  {journal} {\bibinfo  {journal}
  {Phys. Rev. Lett.}\ }\textbf {\bibinfo {volume} {106}},\ \bibinfo {pages}
  {200401} (\bibinfo {year} {2011})}\BibitemShut {NoStop}%
\bibitem [{\citenamefont {Horodecki}\ \emph {et~al.}(1998)\citenamefont
  {Horodecki}, \citenamefont {Horodecki},\ and\ \citenamefont
  {Horodecki}}]{horodecki1998mixed}%
  \BibitemOpen
  \bibfield  {author} {\bibinfo {author} {\bibfnamefont {M.}~\bibnamefont
  {Horodecki}}, \bibinfo {author} {\bibfnamefont {P.}~\bibnamefont
  {Horodecki}}, \ and\ \bibinfo {author} {\bibfnamefont {R.}~\bibnamefont
  {Horodecki}},\ }\href@noop {} {\bibfield  {journal} {\bibinfo  {journal}
  {Phys. Rev. Lett.}\ }\textbf {\bibinfo {volume} {80}},\ \bibinfo {pages}
  {5239} (\bibinfo {year} {1998})}\BibitemShut {NoStop}%
\bibitem [{\citenamefont {DiVincenzo}\ \emph {et~al.}(2000)\citenamefont
  {DiVincenzo}, \citenamefont {Shor}, \citenamefont {Smolin}, \citenamefont
  {Terhal},\ and\ \citenamefont {Thapliyal}}]{divincenzo2000evidence}%
  \BibitemOpen
  \bibfield  {author} {\bibinfo {author} {\bibfnamefont {D.~P.}\ \bibnamefont
  {DiVincenzo}}, \bibinfo {author} {\bibfnamefont {P.~W.}\ \bibnamefont
  {Shor}}, \bibinfo {author} {\bibfnamefont {J.~A.}\ \bibnamefont {Smolin}},
  \bibinfo {author} {\bibfnamefont {B.~M.}\ \bibnamefont {Terhal}}, \ and\
  \bibinfo {author} {\bibfnamefont {A.~V.}\ \bibnamefont {Thapliyal}},\
  }\href@noop {} {\bibfield  {journal} {\bibinfo  {journal} {Phys. Rev. A}\
  }\textbf {\bibinfo {volume} {61}},\ \bibinfo {pages} {062312} (\bibinfo
  {year} {2000})}\BibitemShut {NoStop}%
\bibitem [{\citenamefont {Veitch}\ \emph {et~al.}(2012)\citenamefont {Veitch},
  \citenamefont {Ferrie}, \citenamefont {Gross},\ and\ \citenamefont
  {Emerson}}]{Veitch2012}%
  \BibitemOpen
  \bibfield  {author} {\bibinfo {author} {\bibfnamefont {V.}~\bibnamefont
  {Veitch}}, \bibinfo {author} {\bibfnamefont {C.}~\bibnamefont {Ferrie}},
  \bibinfo {author} {\bibfnamefont {D.}~\bibnamefont {Gross}}, \ and\ \bibinfo
  {author} {\bibfnamefont {J.}~\bibnamefont {Emerson}},\ }\href
  {http://stacks.iop.org/1367-2630/14/i=11/a=113011} {\bibfield  {journal}
  {\bibinfo  {journal} {New J. Phys.}\ }\textbf {\bibinfo {volume} {14}},\
  \bibinfo {pages} {113011} (\bibinfo {year} {2012})}\BibitemShut {NoStop}%
\bibitem [{\citenamefont {Navarrete-Benlloch}\ \emph
  {et~al.}(2012)\citenamefont {Navarrete-Benlloch}, \citenamefont
  {Garc{\'\i}a-Patr{\'o}n}, \citenamefont {Shapiro},\ and\ \citenamefont
  {Cerf}}]{navarrete2012enhancing}%
  \BibitemOpen
  \bibfield  {author} {\bibinfo {author} {\bibfnamefont {C.}~\bibnamefont
  {Navarrete-Benlloch}}, \bibinfo {author} {\bibfnamefont {R.}~\bibnamefont
  {Garc{\'\i}a-Patr{\'o}n}}, \bibinfo {author} {\bibfnamefont {J.~H.}\
  \bibnamefont {Shapiro}}, \ and\ \bibinfo {author} {\bibfnamefont {N.~J.}\
  \bibnamefont {Cerf}},\ }\href@noop {} {\bibfield  {journal} {\bibinfo
  {journal} {Phys. Rev. A}\ }\textbf {\bibinfo {volume} {86}},\ \bibinfo
  {pages} {012328} (\bibinfo {year} {2012})}\BibitemShut {NoStop}%
\bibitem [{\citenamefont {Kim}\ \emph {et~al.}(2008)\citenamefont {Kim},
  \citenamefont {Jeong}, \citenamefont {Zavatta}, \citenamefont {Parigi},\ and\
  \citenamefont {Bellini}}]{kim2008scheme}%
  \BibitemOpen
  \bibfield  {author} {\bibinfo {author} {\bibfnamefont {M.}~\bibnamefont
  {Kim}}, \bibinfo {author} {\bibfnamefont {H.}~\bibnamefont {Jeong}}, \bibinfo
  {author} {\bibfnamefont {A.}~\bibnamefont {Zavatta}}, \bibinfo {author}
  {\bibfnamefont {V.}~\bibnamefont {Parigi}}, \ and\ \bibinfo {author}
  {\bibfnamefont {M.}~\bibnamefont {Bellini}},\ }\href@noop {} {\bibfield
  {journal} {\bibinfo  {journal} {Phys. Rev. Lett.}\ }\textbf {\bibinfo
  {volume} {101}},\ \bibinfo {pages} {260401} (\bibinfo {year}
  {2008})}\BibitemShut {NoStop}%
\bibitem [{\citenamefont {Parigi}\ \emph {et~al.}(2007)\citenamefont {Parigi},
  \citenamefont {Zavatta}, \citenamefont {Kim},\ and\ \citenamefont
  {Bellini}}]{parigi2007probing}%
  \BibitemOpen
  \bibfield  {author} {\bibinfo {author} {\bibfnamefont {V.}~\bibnamefont
  {Parigi}}, \bibinfo {author} {\bibfnamefont {A.}~\bibnamefont {Zavatta}},
  \bibinfo {author} {\bibfnamefont {M.}~\bibnamefont {Kim}}, \ and\ \bibinfo
  {author} {\bibfnamefont {M.}~\bibnamefont {Bellini}},\ }\href@noop {}
  {\bibfield  {journal} {\bibinfo  {journal} {Science}\ }\textbf {\bibinfo
  {volume} {317}},\ \bibinfo {pages} {1890} (\bibinfo {year}
  {2007})}\BibitemShut {NoStop}%
\bibitem [{\citenamefont
  {Fiur{\'a}{\v{s}}ek}(2009)}]{fiuravsek2009engineering}%
  \BibitemOpen
  \bibfield  {author} {\bibinfo {author} {\bibfnamefont {J.}~\bibnamefont
  {Fiur{\'a}{\v{s}}ek}},\ }\href@noop {} {\bibfield  {journal} {\bibinfo
  {journal} {Phys. Rev. A}\ }\textbf {\bibinfo {volume} {80}},\ \bibinfo
  {pages} {053822} (\bibinfo {year} {2009})}\BibitemShut {NoStop}%
\bibitem [{\citenamefont {Marek}\ \emph {et~al.}(2008)\citenamefont {Marek},
  \citenamefont {Jeong},\ and\ \citenamefont {Kim}}]{marek2008generating}%
  \BibitemOpen
  \bibfield  {author} {\bibinfo {author} {\bibfnamefont {P.}~\bibnamefont
  {Marek}}, \bibinfo {author} {\bibfnamefont {H.}~\bibnamefont {Jeong}}, \ and\
  \bibinfo {author} {\bibfnamefont {M.}~\bibnamefont {Kim}},\ }\href@noop {}
  {\bibfield  {journal} {\bibinfo  {journal} {Phys. Rev. A}\ }\textbf {\bibinfo
  {volume} {78}},\ \bibinfo {pages} {063811} (\bibinfo {year}
  {2008})}\BibitemShut {NoStop}%
\bibitem [{\citenamefont {Kitagawa}\ \emph {et~al.}(2006)\citenamefont
  {Kitagawa}, \citenamefont {Takeoka}, \citenamefont {Sasaki},\ and\
  \citenamefont {Chefles}}]{kitagawa2006entanglement}%
  \BibitemOpen
  \bibfield  {author} {\bibinfo {author} {\bibfnamefont {A.}~\bibnamefont
  {Kitagawa}}, \bibinfo {author} {\bibfnamefont {M.}~\bibnamefont {Takeoka}},
  \bibinfo {author} {\bibfnamefont {M.}~\bibnamefont {Sasaki}}, \ and\ \bibinfo
  {author} {\bibfnamefont {A.}~\bibnamefont {Chefles}},\ }\href@noop {}
  {\bibfield  {journal} {\bibinfo  {journal} {Phys. Rev. A}\ }\textbf {\bibinfo
  {volume} {73}},\ \bibinfo {pages} {042310} (\bibinfo {year}
  {2006})}\BibitemShut {NoStop}%
\bibitem [{\citenamefont {Namekata}\ \emph {et~al.}(2010)\citenamefont
  {Namekata}, \citenamefont {Takahashi}, \citenamefont {Fujii}, \citenamefont
  {Fukuda}, \citenamefont {Kurimura},\ and\ \citenamefont
  {Inoue}}]{namekata2010non}%
  \BibitemOpen
  \bibfield  {author} {\bibinfo {author} {\bibfnamefont {N.}~\bibnamefont
  {Namekata}}, \bibinfo {author} {\bibfnamefont {Y.}~\bibnamefont {Takahashi}},
  \bibinfo {author} {\bibfnamefont {G.}~\bibnamefont {Fujii}}, \bibinfo
  {author} {\bibfnamefont {D.}~\bibnamefont {Fukuda}}, \bibinfo {author}
  {\bibfnamefont {S.}~\bibnamefont {Kurimura}}, \ and\ \bibinfo {author}
  {\bibfnamefont {S.}~\bibnamefont {Inoue}},\ }\href@noop {} {\bibfield
  {journal} {\bibinfo  {journal} {Nat. Photonics}\ }\textbf {\bibinfo {volume}
  {4}},\ \bibinfo {pages} {655} (\bibinfo {year} {2010})}\BibitemShut {NoStop}%
\bibitem [{\citenamefont {Fiur{\'a}{\v{s}}ek}\ \emph
  {et~al.}(2005)\citenamefont {Fiur{\'a}{\v{s}}ek}, \citenamefont
  {Garc{\'\i}a-Patr{\'o}n},\ and\ \citenamefont
  {Cerf}}]{fiuravsek2005conditional}%
  \BibitemOpen
  \bibfield  {author} {\bibinfo {author} {\bibfnamefont {J.}~\bibnamefont
  {Fiur{\'a}{\v{s}}ek}}, \bibinfo {author} {\bibfnamefont {R.}~\bibnamefont
  {Garc{\'\i}a-Patr{\'o}n}}, \ and\ \bibinfo {author} {\bibfnamefont {N.~J.}\
  \bibnamefont {Cerf}},\ }\href@noop {} {\bibfield  {journal} {\bibinfo
  {journal} {Phys. Rev. A}\ }\textbf {\bibinfo {volume} {72}},\ \bibinfo
  {pages} {033822} (\bibinfo {year} {2005})}\BibitemShut {NoStop}%
\bibitem [{\citenamefont {Wakui}\ \emph {et~al.}(2007)\citenamefont {Wakui},
  \citenamefont {Takahashi}, \citenamefont {Furusawa},\ and\ \citenamefont
  {Sasaki}}]{wakui2007photon}%
  \BibitemOpen
  \bibfield  {author} {\bibinfo {author} {\bibfnamefont {K.}~\bibnamefont
  {Wakui}}, \bibinfo {author} {\bibfnamefont {H.}~\bibnamefont {Takahashi}},
  \bibinfo {author} {\bibfnamefont {A.}~\bibnamefont {Furusawa}}, \ and\
  \bibinfo {author} {\bibfnamefont {M.}~\bibnamefont {Sasaki}},\ }\href@noop {}
  {\bibfield  {journal} {\bibinfo  {journal} {Opt. Express}\ }\textbf {\bibinfo
  {volume} {15}},\ \bibinfo {pages} {3568} (\bibinfo {year}
  {2007})}\BibitemShut {NoStop}%
\bibitem [{\citenamefont {Walschaers}\ \emph
  {et~al.}(2017{\natexlab{a}})\citenamefont {Walschaers}, \citenamefont
  {Fabre}, \citenamefont {Parigi},\ and\ \citenamefont
  {Treps}}]{walschaers2017entanglement}%
  \BibitemOpen
  \bibfield  {author} {\bibinfo {author} {\bibfnamefont {M.}~\bibnamefont
  {Walschaers}}, \bibinfo {author} {\bibfnamefont {C.}~\bibnamefont {Fabre}},
  \bibinfo {author} {\bibfnamefont {V.}~\bibnamefont {Parigi}}, \ and\ \bibinfo
  {author} {\bibfnamefont {N.}~\bibnamefont {Treps}},\ }\href@noop {}
  {\bibfield  {journal} {\bibinfo  {journal} {Phys. Rev. Lett.}\ }\textbf
  {\bibinfo {volume} {119}},\ \bibinfo {pages} {183601} (\bibinfo {year}
  {2017}{\natexlab{a}})}\BibitemShut {NoStop}%
\bibitem [{\citenamefont {Walschaers}\ \emph
  {et~al.}(2017{\natexlab{b}})\citenamefont {Walschaers}, \citenamefont
  {Fabre}, \citenamefont {Parigi},\ and\ \citenamefont
  {Treps}}]{walschaers2017statistical}%
  \BibitemOpen
  \bibfield  {author} {\bibinfo {author} {\bibfnamefont {M.}~\bibnamefont
  {Walschaers}}, \bibinfo {author} {\bibfnamefont {C.}~\bibnamefont {Fabre}},
  \bibinfo {author} {\bibfnamefont {V.}~\bibnamefont {Parigi}}, \ and\ \bibinfo
  {author} {\bibfnamefont {N.}~\bibnamefont {Treps}},\ }\href@noop {}
  {\bibfield  {journal} {\bibinfo  {journal} {Phys. Rev. A}\ }\textbf {\bibinfo
  {volume} {96}},\ \bibinfo {pages} {053835} (\bibinfo {year}
  {2017}{\natexlab{b}})}\BibitemShut {NoStop}%
\bibitem [{\citenamefont {Genoni}\ \emph {et~al.}(2007)\citenamefont {Genoni},
  \citenamefont {Paris},\ and\ \citenamefont {Banaszek}}]{genoni2007measure}%
  \BibitemOpen
  \bibfield  {author} {\bibinfo {author} {\bibfnamefont {M.~G.}\ \bibnamefont
  {Genoni}}, \bibinfo {author} {\bibfnamefont {M.~G.}\ \bibnamefont {Paris}}, \
  and\ \bibinfo {author} {\bibfnamefont {K.}~\bibnamefont {Banaszek}},\
  }\href@noop {} {\bibfield  {journal} {\bibinfo  {journal} {Phys. Rev. A}\
  }\textbf {\bibinfo {volume} {76}},\ \bibinfo {pages} {042327} (\bibinfo
  {year} {2007})}\BibitemShut {NoStop}%
\bibitem [{\citenamefont {Vidal}\ and\ \citenamefont
  {Tarrach}(1999)}]{Vidal1999}%
  \BibitemOpen
  \bibfield  {author} {\bibinfo {author} {\bibfnamefont {G.}~\bibnamefont
  {Vidal}}\ and\ \bibinfo {author} {\bibfnamefont {R.}~\bibnamefont
  {Tarrach}},\ }\href {\doibase 10.1103/PhysRevA.59.141} {\bibfield  {journal}
  {\bibinfo  {journal} {Phys. Rev. A}\ }\textbf {\bibinfo {volume} {59}},\
  \bibinfo {pages} {141} (\bibinfo {year} {1999})}\BibitemShut {NoStop}%
\bibitem [{\citenamefont {Regula}(2018)}]{Regula2018}%
  \BibitemOpen
  \bibfield  {author} {\bibinfo {author} {\bibfnamefont {B.}~\bibnamefont
  {Regula}},\ }\href {http://stacks.iop.org/1751-8121/51/i=4/a=045303}
  {\bibfield  {journal} {\bibinfo  {journal} {J. Phys. A: Math. Theor.}\
  }\textbf {\bibinfo {volume} {51}},\ \bibinfo {pages} {045303} (\bibinfo
  {year} {2018})}\BibitemShut {NoStop}%
\bibitem [{\citenamefont {de~Melo}\ \emph {et~al.}(2013)\citenamefont
  {de~Melo}, \citenamefont {Ćwikliński},\ and\ \citenamefont
  {Terhal}}]{deMelo2013}%
  \BibitemOpen
  \bibfield  {author} {\bibinfo {author} {\bibfnamefont {F.}~\bibnamefont
  {de~Melo}}, \bibinfo {author} {\bibfnamefont {P.}~\bibnamefont
  {Ćwikliński}}, \ and\ \bibinfo {author} {\bibfnamefont {B.~M.}\
  \bibnamefont {Terhal}},\ }\href
  {http://stacks.iop.org/1367-2630/15/i=1/a=013015} {\bibfield  {journal}
  {\bibinfo  {journal} {New J. Phys.}\ }\textbf {\bibinfo {volume} {15}},\
  \bibinfo {pages} {013015} (\bibinfo {year} {2013})}\BibitemShut {NoStop}%
\bibitem [{\citenamefont {Park}\ \emph {et~al.}(2015)\citenamefont {Park},
  \citenamefont {Zhang}, \citenamefont {Lee}, \citenamefont {Ji}, \citenamefont
  {Um}, \citenamefont {Lv}, \citenamefont {Kim},\ and\ \citenamefont
  {Nha}}]{park2015testing}%
  \BibitemOpen
  \bibfield  {author} {\bibinfo {author} {\bibfnamefont {J.}~\bibnamefont
  {Park}}, \bibinfo {author} {\bibfnamefont {J.}~\bibnamefont {Zhang}},
  \bibinfo {author} {\bibfnamefont {J.}~\bibnamefont {Lee}}, \bibinfo {author}
  {\bibfnamefont {S.-W.}\ \bibnamefont {Ji}}, \bibinfo {author} {\bibfnamefont
  {M.}~\bibnamefont {Um}}, \bibinfo {author} {\bibfnamefont {D.}~\bibnamefont
  {Lv}}, \bibinfo {author} {\bibfnamefont {K.}~\bibnamefont {Kim}}, \ and\
  \bibinfo {author} {\bibfnamefont {H.}~\bibnamefont {Nha}},\ }\href@noop {}
  {\bibfield  {journal} {\bibinfo  {journal} {Phys. Rev. Lett.}\ }\textbf
  {\bibinfo {volume} {114}},\ \bibinfo {pages} {190402} (\bibinfo {year}
  {2015})}\BibitemShut {NoStop}%
\bibitem [{\citenamefont {Hughes}\ \emph {et~al.}(2014)\citenamefont {Hughes},
  \citenamefont {Genoni}, \citenamefont {Tufarelli}, \citenamefont {Paris},\
  and\ \citenamefont {Kim}}]{hughes2014quantum}%
  \BibitemOpen
  \bibfield  {author} {\bibinfo {author} {\bibfnamefont {C.}~\bibnamefont
  {Hughes}}, \bibinfo {author} {\bibfnamefont {M.~G.}\ \bibnamefont {Genoni}},
  \bibinfo {author} {\bibfnamefont {T.}~\bibnamefont {Tufarelli}}, \bibinfo
  {author} {\bibfnamefont {M.~G.}\ \bibnamefont {Paris}}, \ and\ \bibinfo
  {author} {\bibfnamefont {M.}~\bibnamefont {Kim}},\ }\href@noop {} {\bibfield
  {journal} {\bibinfo  {journal} {Phys. Rev. A}\ }\textbf {\bibinfo {volume}
  {90}},\ \bibinfo {pages} {013810} (\bibinfo {year} {2014})}\BibitemShut
  {NoStop}%
\bibitem [{\citenamefont {Je{\v{z}}ek}\ \emph {et~al.}(2011)\citenamefont
  {Je{\v{z}}ek}, \citenamefont {Straka}, \citenamefont {Mi{\v{c}}uda},
  \citenamefont {Du{\v{s}}ek}, \citenamefont {Fiur{\'a}{\v{s}}ek},\ and\
  \citenamefont {Filip}}]{jevzek2011experimental}%
  \BibitemOpen
  \bibfield  {author} {\bibinfo {author} {\bibfnamefont {M.}~\bibnamefont
  {Je{\v{z}}ek}}, \bibinfo {author} {\bibfnamefont {I.}~\bibnamefont {Straka}},
  \bibinfo {author} {\bibfnamefont {M.}~\bibnamefont {Mi{\v{c}}uda}}, \bibinfo
  {author} {\bibfnamefont {M.}~\bibnamefont {Du{\v{s}}ek}}, \bibinfo {author}
  {\bibfnamefont {J.}~\bibnamefont {Fiur{\'a}{\v{s}}ek}}, \ and\ \bibinfo
  {author} {\bibfnamefont {R.}~\bibnamefont {Filip}},\ }\href@noop {}
  {\bibfield  {journal} {\bibinfo  {journal} {Phys. Rev. Lett.}\ }\textbf
  {\bibinfo {volume} {107}},\ \bibinfo {pages} {213602} (\bibinfo {year}
  {2011})}\BibitemShut {NoStop}%
\bibitem [{\citenamefont {Happ}\ \emph {et~al.}(2018)\citenamefont {Happ},
  \citenamefont {Efremov}, \citenamefont {Nha},\ and\ \citenamefont
  {Schleich}}]{Happ2018}%
  \BibitemOpen
  \bibfield  {author} {\bibinfo {author} {\bibfnamefont {L.}~\bibnamefont
  {Happ}}, \bibinfo {author} {\bibfnamefont {M.~A.}\ \bibnamefont {Efremov}},
  \bibinfo {author} {\bibfnamefont {H.}~\bibnamefont {Nha}}, \ and\ \bibinfo
  {author} {\bibfnamefont {W.~P.}\ \bibnamefont {Schleich}},\ }\href
  {http://stacks.iop.org/1367-2630/20/i=2/a=023046} {\bibfield  {journal}
  {\bibinfo  {journal} {New J. Phys.}\ }\textbf {\bibinfo {volume} {20}},\
  \bibinfo {pages} {023046} (\bibinfo {year} {2018})}\BibitemShut {NoStop}%
\bibitem [{\citenamefont {Piani}\ and\ \citenamefont
  {Watrous}(2009)}]{Piani2009}%
  \BibitemOpen
  \bibfield  {author} {\bibinfo {author} {\bibfnamefont {M.}~\bibnamefont
  {Piani}}\ and\ \bibinfo {author} {\bibfnamefont {J.}~\bibnamefont
  {Watrous}},\ }\href {\doibase 10.1103/PhysRevLett.102.250501} {\bibfield
  {journal} {\bibinfo  {journal} {Phys. Rev. Lett.}\ }\textbf {\bibinfo
  {volume} {102}},\ \bibinfo {pages} {250501} (\bibinfo {year}
  {2009})}\BibitemShut {NoStop}%
\bibitem [{\citenamefont {Albarelli}\ \emph {et~al.}(2018)\citenamefont
  {Albarelli}, \citenamefont {Genoni}, \citenamefont {Paris},\ and\
  \citenamefont {Ferraro}}]{albarelli2018resource}%
  \BibitemOpen
  \bibfield  {author} {\bibinfo {author} {\bibfnamefont {F.}~\bibnamefont
  {Albarelli}}, \bibinfo {author} {\bibfnamefont {M.~G.}\ \bibnamefont
  {Genoni}}, \bibinfo {author} {\bibfnamefont {M.~G.~A.}\ \bibnamefont
  {Paris}}, \ and\ \bibinfo {author} {\bibfnamefont {A.}~\bibnamefont
  {Ferraro}},\ }\href@noop {} {\bibfield  {journal} {\bibinfo  {journal}
  {arXiv:1804.05763v1}\ } (\bibinfo {year} {2018})}\BibitemShut {NoStop}%
\bibitem [{\citenamefont {Ralph}\ \emph {et~al.}(2000)\citenamefont {Ralph},
  \citenamefont {Munro},\ and\ \citenamefont
  {Polkinghorne}}]{ralph2000proposal}%
  \BibitemOpen
  \bibfield  {author} {\bibinfo {author} {\bibfnamefont {T.}~\bibnamefont
  {Ralph}}, \bibinfo {author} {\bibfnamefont {W.}~\bibnamefont {Munro}}, \ and\
  \bibinfo {author} {\bibfnamefont {R.}~\bibnamefont {Polkinghorne}},\
  }\href@noop {} {\bibfield  {journal} {\bibinfo  {journal} {Phys. Rev. Lett.}\
  }\textbf {\bibinfo {volume} {85}},\ \bibinfo {pages} {2035} (\bibinfo {year}
  {2000})}\BibitemShut {NoStop}%
\bibitem [{\citenamefont {Thearle}\ \emph {et~al.}(2018)\citenamefont
  {Thearle}, \citenamefont {Janousek}, \citenamefont {Armstrong}, \citenamefont
  {Hosseini}, \citenamefont {Sch\"unemann~(Mraz)}, \citenamefont {Assad},
  \citenamefont {Symul}, \citenamefont {James}, \citenamefont {Huntington},
  \citenamefont {Ralph},\ and\ \citenamefont {Lam}}]{Oliver2018violation}%
  \BibitemOpen
  \bibfield  {author} {\bibinfo {author} {\bibfnamefont {O.}~\bibnamefont
  {Thearle}}, \bibinfo {author} {\bibfnamefont {J.}~\bibnamefont {Janousek}},
  \bibinfo {author} {\bibfnamefont {S.}~\bibnamefont {Armstrong}}, \bibinfo
  {author} {\bibfnamefont {S.}~\bibnamefont {Hosseini}}, \bibinfo {author}
  {\bibfnamefont {M.}~\bibnamefont {Sch\"unemann~(Mraz)}}, \bibinfo {author}
  {\bibfnamefont {S.}~\bibnamefont {Assad}}, \bibinfo {author} {\bibfnamefont
  {T.}~\bibnamefont {Symul}}, \bibinfo {author} {\bibfnamefont {M.~R.}\
  \bibnamefont {James}}, \bibinfo {author} {\bibfnamefont {E.}~\bibnamefont
  {Huntington}}, \bibinfo {author} {\bibfnamefont {T.~C.}\ \bibnamefont
  {Ralph}}, \ and\ \bibinfo {author} {\bibfnamefont {P.~K.}\ \bibnamefont
  {Lam}},\ }\href {\doibase 10.1103/PhysRevLett.120.040406} {\bibfield
  {journal} {\bibinfo  {journal} {Phys. Rev. Lett.}\ }\textbf {\bibinfo
  {volume} {120}},\ \bibinfo {pages} {040406} (\bibinfo {year}
  {2018})}\BibitemShut {NoStop}%
\end{thebibliography}%

\end{document}